\documentclass[oribibl]{llncs}
\usepackage{dla}

\title{Data Linkage Algebra, Data Linkage Dynamics, and Priority
       Rewriting}
\author{J.A. Bergstra \and C.A. Middelburg}
\institute{Informatics Institute, Faculty of Science, University of
           Amsterdam, \\
           Science Park~904, 1098~XH Amsterdam, the Netherlands \\
           \email{J.A.Bergstra@uva.nl,C.A.Middelburg@uva.nl}}

\begin{document}
\maketitle

\begin{abstract}
We introduce an algebra of data linkages.
Data linkages are intended for modelling the states of computations in
which dynamic data structures are involved.
We present a simple model of computation in which states of computations
are modelled as data linkages and state changes take place by means of
certain actions.
We describe the state changes and replies that result from performing
those actions by means of a term rewriting system with rule priorities.
The model in question is an upgrade of molecular dynamics.
The upgrading is mainly concerned with the features to deal with values
and the features to reclaim garbage.
\\[1.5ex]
{\sl Keywords:}
data linkage algebra, data linkage dynamics, garbage collection, 
priority rewrite system.
\\[1.5ex]
{\sl 1998 ACM Computing Classification:}
D.3.3, D.4.2, F.1.1, F.3.2, F.3.3.
\end{abstract}

\section{Introduction}
\label{sect-intro}

The data structures involved in programming are quite often dynamic data 
structures, i.e.\ data structures that may vary in size and shape.
Dynamic data structures are data structures whose constituent parts are 
linked in one way or another to facilitate the insertion and deletion of 
constituent parts.
Scientific work on dynamic data structure is generally concerned with 
specific dynamic data structures that show a simple shape, such as 
linked lists and various kinds of tree structures, or with dynamic data 
structures in general.
In the former case, the dynamic data structures concerned are regularly 
considered in abstraction from their representation by means of 
pointers.
In the latter case, however, dynamic data structures are primarily 
considered at the level of their representation by means of pointers.
Although it is often useful to abstract from the representation by means 
of pointers, it seems that no serious attempts have been made to provide 
a setting in which this is possible.
The aim of the current paper is to provide such a setting.

We introduce an algebra, called data linkage algebra, of which the
elements are intended for modelling the states of computations in which
dynamic data structures are involved.
We also present a simple model of computation, called data linkage
dynamics, in which states of computations are modelled as elements of
data linkage algebra and state changes take place by means of certain
actions.
We describe the state changes and replies that result from performing
those actions by means of a term rewriting system with rule
priorities~\cite{BBKW89a}.

Term rewriting systems take an important place in theoretical computer
science.
Moreover, because term rewriting is a practical mechanism for doing
calculations, term rewriting systems have many applications in software
engineering.
Term rewriting systems with rule priorities, also called priority
rewrite systems, were first proposed in~\cite{BBKW89a}.
Further studies of priority rewrite systems can, for example, be found
in~\cite{Gna09a,Moh89a,Pol98a,ST98a}.
Applications of priority rewrite systems are found in various areas, see 
e.g.~\cite{BKV98a,KKO09a,UY09a}. 
The rule priorities add expressive power: the reduction relation of a
priority rewrite system is not decidable in general.
It happens that it is quite convenient to describe the state changes and
replies that result from performing the actions of data linkage dynamics
by means of a priority rewrite system.
Moreover, the priority rewrite system in question turns out
computationally unproblematic: its reduction relation is decidable.

We take the view that the behaviours produced by sequential programs 
under execution are threads as considered in basic thread 
algebra~\cite{BL02a} (see also \cite[Chapter~2]{BM12b}).%
\footnote
{In~\cite{BL02a}, basic thread algebra is introduced under the name
 basic polarized process algebra.
 Prompted by the development of thread algebra~\cite{BM04c}, which is a
 design on top of it, basic polarized process algebra has been renamed
 to basic thread algebra.
}
These threads represent in a direct way the behaviours produced by 
instruction sequences under execution.
A thread proceeds by performing actions in a sequential fashion.
Upon each action performed by a thread, a reply from the execution 
environment, which takes the action as an instruction to be processed, 
determines how the thread proceeds.
In~\cite{BP02a}, basic thread algebra has been extended with services, 
which represent an abstract view on the behaviours exhibited by the 
components of an execution environment that are capable of processing 
particular instructions independently, and use and apply operators, 
which have to do with the effects of the interaction between threads and 
services that takes place during instruction processing (see 
also~\cite[Chapter~3]{BM12b}).

The state changes and replies that result from performing the actions
of data linkage dynamics can be achieved by means of services.
In the current paper, we explain how basic thread algebra extended with
services and use operators can be combined with data linkage dynamics 
such that the whole can be used for studying issues concerning the use 
of dynamic data structures in programming.

Data linkage dynamics is an upgrade of molecular dynamics, which was 
first described in~\cite{BB02a}.
The name molecular dynamics refers to the molecule metaphor used to
explain the model.
By that, there is no clue in the name itself to what it stands for.
To remedy this defect, the upgrade has been renamed to data linkage
dynamics.
The upgrading is mainly concerned with the features to deal with values
and the features to reclaim garbage.
In data linkage dynamics, calculations in a non-trivial finite
meadow~\cite{BBP13a,BT07a,BR07a}, such as a finite field with the 
multiplicative inverse operation made total by imposing that the 
multiplicative inverse of zero is zero, can be done.
The features to reclaim garbage include: full garbage collection,
restricted garbage collection (as if reference counts are used), safe
disposal of potential garbage, and unsafe disposal of potential garbage.

In~\cite{BM06c}, a description of the state changes and replies that
result from performing the actions of molecular dynamics was given in
the world of sets.
In the current paper, we relate this description to the description
based on data linkage algebra by widening the former to a description
for data linkage dynamics and showing that the widened description
agrees with the description based on data linkage algebra.

Data linkage dynamics in itself is meant to convey a theoretical 
understanding of the pragmatic concept of a dynamic data structure as 
exploited in the practice of programming.
Such a theoretical understanding is a valuable complement of the 
understanding of how specific programs use dynamic data structures, 
which is acquired by means of dynamic analysis tools that analyze how 
programs build and modify them (see e.g.~\cite{PV06a}).
We expect that a theoretical understanding will become increasingly 
important to the development of successful software systems.
Below, we give a first impression of what is to be expected from data 
linkage dynamics as a setting in which issues concerning dynamic data 
structures are studied.

In work on dynamic data structures, on the one hand issues concerning 
specific dynamic data structures, in particular performance issues and 
computational complexity issues, are studied (see 
e.g.~\cite{FFV79a,Knu77a,OL82a,Tam88a}).
In the studies in question, the dynamic data structures concerned are 
mostly considered in abstraction from their representation by means of 
pointers.
We believe that issues like these ones can also be studied in the 
setting of data linkage dynamics, and we expect that the use of a single 
setting facilitates uniformity in the way in which the same issue is 
approached for different dynamic data structures.

In work on dynamic data structures, on the other hand issues concerning 
dynamic data structures in general, particularly issues related to 
optimization and parallelization of sequential programs that make use of 
dynamic data structures, are studied (see 
e.g.~\cite{Car96a,CAZ04a,HHN94a,RCRH95a}).
In the studies in question, dynamic data structures are mostly 
considered at the level of their representation by means of pointers.
We consider it likely that, for issues like these ones, more general 
results can be obtained by studying the issues concerned in the setting 
of data linkage dynamics, where they can be considered in abstraction 
from their representation by means of pointers. 

This paper is organized as follows.
First, we introduce data linkage algebra (Section~\ref{sect-DLA}).
Next, after a short review of priority rewrite systems 
(Section~\ref{sect-PRS}), we present data linkage dynamics 
(Sections~\ref{sect-DLD-K}, \ref{sect-DLD}, \ref{sect-properties}, 
and~\ref{sect-garb-coll}).
After that, we review basic thread algebra and its extension with 
services and use operators (Sections~\ref{sect-BTA} and~\ref{sect-TSI})
and explain how this extension of basic thread algebra can be combined 
with data linkage dynamics (Section~\ref{sect-comb-TA-DLD}).
Following this, we give the alternative description of data linkage
dynamics in the world of sets (Sections~\ref{sect-DLD-alt-descr},
\ref{sect-correctness}, and~\ref{sect-garb-coll-maps}).
Finally, we make some concluding remarks (Section~\ref{sect-concl}).

Some familiarity with term rewriting systems is assumed.
The desirable background can, for example, be found
in~\cite{DJ90a,Klo92a,KV03a}.
For convenience, the basic definitions and results regarding term 
rewriting systems are collected in an appendix.

\section{Data Linkage Algebra}
\label{sect-DLA}

In this section, we introduce the algebraic theory \DLA\ (Data Linkage
Algebra).

The elements of the initial algebra of \DLA\ are intended for modelling
the states of computations in which dynamic data structures are 
involved.
These states resemble collections of molecules composed of atoms.
An atom can have fields and each of those fields can contain an atom.
An atom together with the ones it has links to via fields can be viewed
as a sub-molecule, and a sub-molecule that is not contained in a larger
sub-molecule can be viewed as a molecule.
Thus, the collection of molecules that make up a state can be viewed as
a fluid.
To make atoms reachable, there are spots and each spot can contain an
atom.

Disengaging from the molecule metaphor, atoms will henceforth be called
atomic objects.
Moreover, sub-molecules, molecules and fluids will henceforth not be
distinguished and commonly be called data linkages.

In \DLA, it is assumed that a fixed but arbitrary finite set $\Spot$ of
\emph{spots}, a fixed but arbitrary finite set $\Field$ of
\emph{fields}, a fixed but arbitrary finite set $\AtObj$ of \emph{atomic
objects}, and a fixed but arbitrary finite set $\Value$ of \emph{values}
have been given.

\DLA\ has one sort: the sort $\DaLi$ of \emph{data linkages}.
To build terms of sort $\DaLi$, \DLA\ has the following constants and
operators:
\begin{itemize}
\item
for each $s \in \Spot$ and $a \in \AtObj$,
the \emph{spot link} constant $\const{\slink{s}{a}}{\DaLi}$;
\item
for each $a \in \AtObj$ and $f \in \Field$,
the \emph{partial field link} constant $\const{\pflink{a}{f}}{\DaLi}$;
\item
for each $a,b \in \AtObj$ and $f \in \Field$,
the \emph{field link} constant $\const{\flink{a}{f}{b}}{\DaLi}$;
\item
for each $a \in \AtObj$ and $n \in \Value$,
the \emph{value association} constant $\const{\valass{a}{n}}{\DaLi}$;
\item
the \emph{empty data linkage} constant $\const{\emptydl}{\DaLi}$;
\item
the binary \emph{data linkage combination} operator
$\funct{\dlcom}{\DaLi \x \DaLi}{\DaLi}$;
\item
the binary \emph{data linkage overriding combination} operator
$\funct{\dlori}{\DaLi \x \DaLi}{\DaLi}$.
\end{itemize}
Terms of sort $\DaLi$ are built as usual (see e.g.~\cite{ST99a,Wir90a}).
Throughout the paper, we assume that there are infinitely many variables
of sort $\DaLi$, including $X$, $Y$, $Z$.
We use infix notation for data linkage combination and data linkage
overriding combination.

Let $L$ and $L'$ be closed \DLA\ terms.
Then the constants and operators of \DLA\ can be explained as follows:
\begin{itemize}
\item
$\slink{s}{a}$ is the atomic data linkage that consists of a link via
spot $s$ to atomic object $a$;
\item
$\pflink{a}{f}$ is the atomic data linkage that consists of a partial
link from atomic object $a$ via field $f$;
\item
$\flink{a}{f}{b}$ is the atomic data linkage that consists of a link
from atomic object~$a$ via field $f$ to atomic object $b$;
\item
$\valass{a}{n}$ is the atomic data linkage that consists of an
association of the value~$n$ with atomic object $a$;
\item
$\emptydl$ is the data linkage that does not contain any atomic data
linkage;
\item
$L \dlcom L'$ is the union of the data linkages $L$ and $L'$;
\item
$L \dlori L'$ differs from $L \dlcom L'$ as follows:
\begin{itemize}
\item
if $L$ contains spot links via spot $s$ and $L'$ contains spot links via
spot $s$, then the former links are overridden by the latter ones;
\pagebreak[2]
\item
if $L$ contains partial field links and/or field links from atomic
object $a$ via field $f$ and $L'$ contains partial field links and/or
field links from atomic object $a$ via field $f$, then the former
partial field links and/or field links are overridden by the latter
ones;
\item
if $L$ contains value associations with atomic object $a$ and $L'$
contains value associations with atomic object $a$, then the former
value associations are overridden by the latter ones.
\end{itemize}
\end{itemize}

Following the introduction of \DLA, we will present a simple model of
computation that pertains to the use of dynamic data structures in
programming.
\DLA\ provides a notation that enables us to get a clear picture of
computations in the context of that model.

The axioms of \DLA\ are given in Table~\ref{axioms-DLA}.%
\begin{table}[!t]
\caption{Axioms of \DLA}
\label{axioms-DLA}
\begin{eqntbl}
\begin{seqncol}
X \dlcom Y = Y \dlcom X \\
X \dlcom (Y  \dlcom Z) = (X \dlcom Y)  \dlcom Z \\
X \dlcom X = X \\
X \dlcom \emptydl = X
\eqnsep
\emptydl \dlori X = X \\
X \dlori \emptydl = X \\
X \dlori (Y \dlcom Z) = (X \dlori Y) \dlcom (X \dlori Z) \\
(X \dlcom \slinkp{s}{a}) \dlori \slinkp{s}{b} =
X \dlori \slinkp{s}{b} \\
(X \dlcom \pflinkp{a}{f}) \dlori \pflinkp{a}{f} =
X \dlori \pflinkp{a}{f} \\
(X \dlcom \flinkp{a}{f}{b}) \dlori \pflinkp{a}{f} =
X \dlori \pflinkp{a}{f} \\
(X \dlcom \pflinkp{a}{f}) \dlori \flinkp{a}{f}{b} =
X \dlori \flinkp{a}{f}{b} \\
(X \dlcom \flinkp{a}{f}{b}) \dlori \flinkp{a}{f}{c} =
X \dlori \flinkp{a}{f}{c} \\
(X \dlcom \valass{a}{n}) \dlori \valass{a}{m} =
X \dlori \valass{a}{m} \\
(X \dlcom \slinkp{s}{a}) \dlori \slinkp{t}{b} =
(X \dlori \slinkp{t}{b}) \dlcom \slinkp{s}{a}
 & \mif s \neq t \\
(X \dlcom \pflinkp{a}{f}) \dlori \slinkp{s}{b} =
(X \dlori \slinkp{s}{b}) \dlcom \pflinkp{a}{f} \\
(X \dlcom \flinkp{a}{f}{b}) \dlori \slinkp{s}{c} =
(X \dlori \slinkp{s}{c}) \dlcom \flinkp{a}{f}{b} \\
(X \dlcom \valass{a}{n}) \dlori \slinkp{s}{b} =
(X \dlori \slinkp{s}{b}) \dlcom \valass{a}{n} \\
(X \dlcom \slinkp{s}{a}) \dlori \pflinkp{b}{f} =
(X \dlori \pflinkp{b}{f}) \dlcom \slinkp{s}{a} \\
(X \dlcom \pflinkp{a}{f}) \dlori \pflinkp{b}{g} =
(X \dlori \pflinkp{b}{g}) \dlcom \pflinkp{a}{f}
 & \mif a \neq b \Lor f \neq g \\
(X \dlcom \flinkp{a}{f}{b}) \dlori \pflinkp{c}{g} =
(X \dlori \pflinkp{c}{g}) \dlcom \flinkp{a}{f}{b}
 & \mif a \neq c \Lor f \neq g \\
(X \dlcom \valass{a}{n}) \dlori \pflinkp{b}{f} =
(X \dlori \pflinkp{b}{f}) \dlcom \valass{a}{n} \\
(X \dlcom \slinkp{s}{a}) \dlori \flinkp{b}{f}{c} =
(X \dlori \flinkp{b}{f}{c}) \dlcom \slinkp{s}{a} \\
(X \dlcom \pflinkp{a}{f}) \dlori \flinkp{b}{g}{c} =
(X \dlori \flinkp{b}{g}{c}) \dlcom \pflinkp{a}{f}
 & \mif a \neq b \Lor f \neq g \\
(X \dlcom \flinkp{a}{f}{b}) \dlori \flinkp{c}{g}{d} =
(X \dlori \flinkp{c}{g}{d}) \dlcom \flinkp{a}{f}{b}
 & \mif a \neq c \Lor f \neq g \\
(X \dlcom \valass{a}{n}) \dlori \flinkp{b}{f}{c} =
(X \dlori \flinkp{b}{f}{c}) \dlcom \valass{a}{n} \\
(X \dlcom \slinkp{s}{a}) \dlori \valass{b}{n} =
(X \dlori \valass{b}{n}) \dlcom \slinkp{s}{a} \\
(X \dlcom \pflinkp{a}{f}) \dlori \valass{b}{n} =
(X \dlori \valass{b}{n}) \dlcom \pflinkp{a}{f} \\
(X \dlcom \flinkp{a}{f}{b}) \dlori \valass{c}{n} =
(X \dlori \valass{c}{n}) \dlcom \flinkp{a}{f}{b} \\
(X \dlcom \valass{a}{n}) \dlori \valass{b}{m} =
(X \dlori \valass{b}{m}) \dlcom \valass{a}{n}
 & \mif a \neq b
\end{seqncol}
\end{eqntbl}
\end{table}
In this table, $s$ and $t$ stand for arbitrary spots from $\Spot$,
$f$ and $g$ stand for arbitrary fields from $\Field$,
$a$, $b$, $c$ and $d$ stand for arbitrary atomic objects from $\AtObj$,
and $n$ and $m$ stand for arbitrary values from $\Value$.

In the examples given in this paper, we take the set 
$\set{\ul{n} \where n \in \set{0,\ldots,9}}$ for $\AtObj$.
\begin{example}
\label{example-DLA-term}
We consider the following closed \DLA\ term:
\begin{ldispl}
(\slinkp{r}{\ul{0}} \dlcom
 \flinkp{\ul{0}}{\nm{up}}{\ul{1}} \dlcom
 \flinkp{\ul{1}}{\nm{dn}}{\ul{0}} \dlcom
 \flinkp{\ul{1}}{\nm{up}}{\ul{2}} \dlcom
 \flinkp{\ul{2}}{\nm{dn}}{\ul{0}} \dlcom
 \slinkp{s}{\ul{2}}) \dlori
 \flinkp{\ul{2}}{\nm{dn}}{\ul{1}}\;.
\end{ldispl}%
Using the axioms of \DLA, we can establish that this term denotes the 
same data linkage as the following term:
\begin{ldispl}
\phantom{(}
 \slinkp{r}{\ul{0}} \dlcom
 \flinkp{\ul{0}}{\nm{up}}{\ul{1}} \dlcom
 \flinkp{\ul{1}}{\nm{dn}}{\ul{0}} \dlcom
 \flinkp{\ul{1}}{\nm{up}}{\ul{2}} \dlcom
 \flinkp{\ul{2}}{\nm{dn}}{\ul{1}} \dlcom
 \slinkp{s}{\ul{2}}\;.
\end{ldispl}%
The data linkage concerned is represented graphically in 
Figure~\ref{fig-graph-repr}.
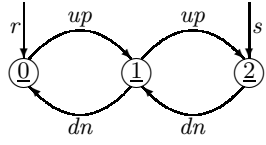
\begin{figure}[!t]
\begin{center}
\setlength{\unitlength}{0.75mm}
\begin{picture}(60,25)(0,0)
%
\multiput(10,10)(20,0){3}{\circle{5}}
\put(10,10){\makebox(0,0){$\ul{0}$}}
\put(30,10){\makebox(0,0){$\ul{1}$}}
\put(50,10){\makebox(0,0){$\ul{2}$}}
\put(10,22.5){\vector(0,-1){10}}
\put(8.5,18){\makebox(0,0){$\mathit{r}$}}
\put(50,22.5){\vector(0,-1){10}}
\put(51.5,18){\makebox(0,0){$\mathit{s}$}}
\qbezier(11,12.5)(20,22.5)(29,12.5)
\put(27.2,14.5){\vector(1,-1){2}}
\qbezier(31,12.5)(40,22.5)(49,12.5)
\put(47.2,14.5){\vector(1,-1){2}}
\multiput(20,19.5)(20,0){2}{\makebox(0,0){$\mathit{up}$}}
\qbezier(11,7.5)(20,-2.5)(29,7.5)
\put(13.2,5.5){\vector(-1,1){2}}
\qbezier(31,7.5)(40,-2.5)(49,7.5)
\put(33.2,5.5){\vector(-1,1){2}}
\multiput(20,0.5)(20,0){2}{\makebox(0,0){$\mathit{dn}$}}
\end{picture}
\end{center}
\caption{Graphical representation of a data linkage}
\label{fig-graph-repr}
\end{figure}

\end{example}

All closed \DLA\ terms are derivably equal to basic terms over \DLA,
i.e.\ closed \DLA\ terms in which the data linkage overriding 
combination operator does not occur.

The set $\cB$ of \emph{basic terms} over \DLA\ is inductively defined by
the following rules:
\begin{itemize}
\item
$\emptydl \in \cB$;
\item
if $s \in \Spot$ and $a \in \AtObj$, then $\slink{s}{a} \in \cB$;
\item
if $a \in \AtObj$ and $f \in \Field$, then $\pflink{a}{f} \in \cB$;
\item
if $a,b \in \AtObj$ and $f \in \Field$, then $\flink{a}{f}{b} \in \cB$;
\item
if $a \in \AtObj$ and $n \in \Value$, then $\valass{a}{n} \in \cB$;
\item
if $L_1,L_2 \in \cB$, then $L_1 \dlcom L_2 \in \cB$.
\end{itemize}
\begin{theorem}[Elimination]
\label{theorem-elimination}
For all closed \DLA\ terms $L$, there exists a basic term $L' \in \cB$
such that $L = L'$ is derivable from the axioms of \DLA.
\end{theorem}
\begin{proof}
This is easily proved by induction on the structure of $L$.
In the case where $L \equiv L_1 \dlori L_2$, we use the fact that for
all basic terms $L'_1,L'_2 \in \cB$, there exists a basic term $L'' \in
\cB$ such that $L'_1 \dlori L'_2 = L''$ is derivable from the axioms of
\DLA.
This is easily proved by induction on the structure of $L'_2$.
\qed
\end{proof}

We are only interested in the initial model of \DLA.
We write $\DL$ for the set of all elements of the initial model of
$\DLA$.

$\DL$ consists of the equivalence classes of basic terms over \DLA\ with
respect to the equivalence induced by the axioms of \DLA.
In other words, modulo equivalence, $\cB$ is $\DL$.
Henceforth, we will identify basic terms over \DLA\ and their
equivalence classes.

Let $L \in \DL$.
Then $L$ is \emph{locally deterministic} with regard to a spot 
$s \in \Spot$ if the following holds:
\begin{itemize}
\item
$L \dlcom \slinkp{s}{a} = L$ for some $a \in \AtObj$;
\item
$L \dlcom \slinkp{s}{a} = L \dlcom \slinkp{s}{b}$ implies $a = b$ for 
all $a,b \in \AtObj$;
\end{itemize}
and $L$ is \emph{locally deterministic} with regard to a field 
$f \in \Field$ for an atomic object $a \in \AtObj$ if the following 
holds:
\begin{itemize}
\item
$L \dlcom \flinkp{a}{f}{b} = L$ for some $b \in \AtObj$ or 
$L \dlcom \pflinkp{a}{f} = L$;
\item
$L \dlcom \flinkp{a}{f}{b} = L \dlcom \flinkp{a}{f}{c}$ implies $b = c$ 
for all $b,c \in \AtObj$;
\item
$L \dlcom \flinkp{a}{f}{b} \neq L \dlcom \pflinkp{a}{f}$ for all
$b \in \AtObj$;
\end{itemize}
and $L$ is \emph{locally deterministic} with regard to the value 
assignment for an atomic object $a \in \AtObj$ if the following holds:
\begin{itemize}
\item
$L \dlcom \valass{a}{n} = L$ for some $n \in \Value$;
\item
$L \dlcom \valass{a}{n} = L \dlcom \valass{a}{m}$ implies $n = m$ for 
all $n,m \in \Value$;
\end{itemize}
and $L$ is \emph{deterministic} if the following holds:
\begin{itemize}
\item
for all $s \in \Spot$, $L$ is locally deterministic with regard to $s$;
\item
for all $f \in \Field$ and all $a \in \AtObj$, $L$ is locally 
deterministic with regard to $f$ for $a$; 
\item
for all $a \in \AtObj$, $L$ is locally deterministic with regard to the 
value assignment for $a$.
\end{itemize}
A data linkage $L \in \DL$ is \emph{non-deterministic} if it is not
deterministic.

In Section~\ref{sect-DLD-alt-descr}, deterministic data linkages are
represented by means of functions and data linkage overriding
combination is modelled by means of function overriding.

\section{Interlude: Priority Rewrite Systems}
\label{sect-PRS}

In Sections~\ref{sect-DLD-K} and~\ref{sect-DLD}, we will present data 
linkage dynamics, a model of computation in which states of computations 
are modelled as data linkages and state changes take place by means of 
certain actions.
We will describe the state changes and replies that result from 
performing these actions by means of a priority rewrite system.
Therefore, we shortly review priority rewrite systems first.
A comprehensive account of priority rewrite systems can be found 
in~\cite{BBKW89a}.
For convenience, the basic definitions and results regarding term 
rewriting systems are collected in an appendix.

A \emph{priority rewrite system} is a pair $\tup{\cR,<}$, where $\cR$ is
a term rewriting system and $<$ is a partial order on the set of rewrite
rules of $\cR$.

Informally, the procedural meaning of the partial order on the set of 
rewrite rules of a priority rewrite system is that a term of the form 
$f(t_1,\ldots,t_n)$ is allowed to be rewritten according to some 
applicable rewrite rule only if it cannot be rewritten to a term 
$f(t'_1,\ldots,t'_n)$ to which a rewrite rule of a higher priority is 
applicable.

Let $\tup{\cR,<}$ be a priority rewrite system, and 
let $R$ be a set of closed instances of rewrite rules of $\cR$.
Then an $R$-\emph{reduction} is a reduction of $\cR$ that belongs to
the closure of $R$ under closed contexts, transitivity and reflexivity.
Let $r$ be a rewrite rule.
Then an $r$-\emph{rewrite} is a closed substitution instance of $r$ and
an $r$-\emph{redex} is the left-hand side of a substitution instance of 
$r$.

Let $\tup{\cR,<}$ be a priority rewrite system.
Assume that there exists a unique set $R$ of closed instances of rewrite
rules of $\cR$ such that an $r$-rewrite $t \osred s \in R$ if there does
not exist an $R$-reduction $t \msred t'$ that\linebreak[2] leaves the 
head symbol of $t$ unaffected and an $r'$-rewrite $t' \osred s' \in R$ 
with $r < r'$.
Then $\tup{\cR,<}$ deter\-mines a \emph{one-step reduction relation} as
follows: $\osred$ is the closure of $R$ under closed contexts.
Moreover, let $E$ be a set of equations between terms over the signature
of $\cR$.
Then $\tup{\cR,<}$ determines a \emph{one-step reduction relation modulo
$E$} as follows: $t \osredm{E} s$ if and only if $t' \osred s'$ for some
$t'$ and $s'$ such that $t = t'$ and $s = s'$ are derivable from $E$
(where $\osred$ denotes the one-step reduction relation determined by
$\tup{\cR,<}$).
If a unique $R$ as described above exists, $\tup{\cR,<}$ is called
\emph{well-defined}.
If a priority rewrite system is not well-defined, then it does not
determine a one-step reduction relation.

Let $\tup{\cR,<}$ be a priority rewrite system.
Then $\cR$ is called the underlying term rewriting system of 
$\tup{\cR,<}$.

The priority rewrite system for data linkage dynamics is actually a
many-sorted priority rewrite system.
The definitions and results concerning term re\-writing systems extend
easily to the many-sorted case, see e.g.~\cite{BT95a}, and likewise for
priority rewrite systems.

Equations can serve as rewrite rules.
Taken as rewrite rules, equations are only used in the direction from
left to right.
In the priority rewrite system for data linkage dynamics, equations that
serve as axioms of \DLA\ are taken as rewrite rules.

Henceforth, all rewrite rules will be written as equations.
Moreover, the notation
\begin{ldispl}
\begin{array}[t]{@{}l@{\quad}l@{}}
\bprio{n_1} & r_1
\\ \qquad \vdots \\
\eprio{n_k} & r_k
\end{array}
\end{ldispl}%
will be used in a table of rewrite rules to indicate that each of the
rewrite rules $r_1,\ldots,r_k$ is incomparable with each of the other
rewrite rules in the table and, for $i,j \in \set{1,\ldots,k}$,
$r_i < r_j$ if and only if $n_i > n_j$.

\section{The Kernel of Data Linkage Dynamics}
\label{sect-DLD-K}

\DLD\ (Data Linkage Dynamics) is a simple model of computation that
pertains to the use of dynamic data structures in programming.
It comprises states, basic actions, i.e.\ indivisible actions, and the 
state changes and replies that result from performing the basic actions.
The states of \DLD\ are data linkages.

In connection with the value-related basic actions of \DLD, it is 
assumed that a fixed but arbitrary model of a certain set of equations 
has been given and that the set $\Value$ consists of the elements of 
that algebra.
A different set of equations would give rise to a variant of \DLD\ that
includes a large part of \DLD, namely the part with its features to 
structure data dynamically.
In this section, we introduce this part of \DLD, called \DLD-K 
(\DLD\ Kernel).
In Section~\ref{sect-DLD}, we will introduce the remaining part of \DLD,
i.e.\ the part with its features to deal with values found in 
dynamically structured data.
Basic actions related to reclaiming garbage are treated separately in 
Section~\ref{sect-garb-coll}.

Like in \DLA, it is assumed that a fixed but arbitrary finite set
$\Spot$ of spots, a fixed but arbitrary finite set $\Field$ of fields,
a fixed but arbitrary finite set $\AtObj$ of atomic objects, and a fixed
but arbitrary set $\Value$ of values have been given.
It is also assumed that a fixed but arbitrary \emph{choice}
function $\funct{\cf}{(\setof{(\AtObj)} \diff \emptyset)}{\AtObj}$ such
that, for all $A \in \setof{(\AtObj)} \diff \emptyset$, $\cf(A) \in A$
has been given.
The function $\cf$ is used whenever a fresh atomic object must be
obtained.

By means of the basic actions of \DLD-K, fresh atomic objects can be 
obtained, fields can be added to and removed from atomic objects, and 
the contents of fields of atomic objects can be examined and modified.
A few basic actions use a spot to put an atomic object in or to get an 
atomic object from.
The contents of spots can be compared and modified.

\DLD-K has the following basic actions:
\begin{itemize}
\item
for each $s \in \Spot$,
a \emph{get fresh atomic object action} $\getatobj{s}$;
\item
for each $s,t \in \Spot$, a \emph{set spot action} $\setspot{s}{t}$;
\item
for each $s \in \Spot$, a \emph{clear spot action} $\clrspot{s}$;
\item
for each $s,t \in \Spot$,
an \emph{equality test action} $\equaltst{s}{t}$;
\item
for each $s \in \Spot$,
an \emph{undefinedness test action} $\undeftst{s}$;
\item
for each $s \in \Spot$ and $f \in \Field$,
a \emph{add field action} $\addfield{s}{f}$;
\item
for each $s \in \Spot$ and $f \in \Field$,
a \emph{remove field action} $\rmvfield{s}{f}$;
\item
for each $s \in \Spot$ and $f \in \Field$,
a \emph{has field action} $\hasfield{s}{f}$;
\item
for each $s,t \in \Spot$ and $f \in \Field$,
a \emph{set field action} $\setfield{s}{f}{t}$;
\item
for each $s \in \Spot$ and $f \in \Field$,
a \emph{clear field action} $\clrfield{s}{f}$;
\item
for each $s,t \in \Spot$ and $f \in \Field$,
a \emph{get field action} $\getfield{s}{t}{f}$.
\end{itemize}
We write $\Act_\DLDK$ for the set of all basic actions of \DLD-K.

The intuition is that performing a basic action may cause a state change 
and will produce a reply. 
The possible replies are $\True$ (standing for true) and $\False$
(standing for false), and the actual reply is state-dependent.
Basic actions intended for examining the state do not cause a state 
change, and basic actions intended for changing the state produce the 
reply $\False$ only if something precludes a state change.

When speaking informally about a state $L$ of \DLD-K, we say:
\begin{itemize}
\item
if $L$ is locally deterministic with regard to spot $s$, 
\emph{the content of spot $s$} instead of the unique atomic object $a$ 
for which $\slink{s}{a}$ is contained in $L$;
\item
if $L$ is locally deterministic with regard to field $f$ for atomic 
object $a$, \emph{the content of field $f$ of atomic object $a$} instead 
of the unique atomic object $b$ for which $\flink{a}{f}{b}$ is contained 
in $L$;
\item
\emph{the fields of atomic object $a$} instead of the set of all fields
$f$ such that either $\pflink{a}{f}$ is contained in $L$ or there
exists an atomic object $b$ such that $\flink{a}{f}{b}$ is contained in
$L$.
\end{itemize}

The basic actions of \DLD-K can be explained as follows if all spots and 
fields involved in performing them are spots and fields with regard to 
which the current state is locally deterministic:
\begin{itemize}
\item
$\getatobj{s}$:
if a fresh atomic object can be allocated, then the content of spot $s$
becomes that fresh atomic object and the reply is $\True$; otherwise,
nothing changes and the reply is $\False$;
\item
$\setspot{s}{t}$:
the content of spot $s$ becomes the same as the content of spot $t$
and the reply is $\True$;
\item
$\clrspot{s}$:
the content of spot $s$ becomes undefined and the reply is $\True$;
\item
$\equaltst{s}{t}$:
if the content of spot $s$ equals the content of spot $t$, then
nothing changes and the reply is $\True$; otherwise, nothing changes and
the reply is $\False$;
\item
$\undeftst{s}$:
if the content of spot $s$ is undefined, then nothing changes and the
reply is $\True$; otherwise, nothing changes and the reply is $\False$;
\item
$\addfield{s}{f}$:
if the content of spot $s$ is an atomic object and $f$ does not yet
belong to the fields of that atomic object, then $f$ is added (with
undefined content) to the fields of that atomic object and the reply is
$\True$; otherwise, nothing changes and the reply is $\False$;
\item
$\rmvfield{s}{f}$:
if the content of spot $s$ is an atomic object and $f$ belongs to the
fields of that atomic object, then $f$ is removed from the fields of
that atomic object and the reply is $\True$; otherwise, nothing changes
and the reply is $\False$;
\item
$\hasfield{s}{f}$:
if the content of spot $s$ is an atomic object and $f$ belongs to the
fields of that atomic object, then nothing changes and the reply is
$\True$; otherwise, nothing changes and the reply is $\False$;
\item
$\setfield{s}{f}{t}$:
if the content of spot $s$ is an atomic object and $f$ belongs to the
fields of that atomic object, then the content of that field becomes the
same as the content of spot $t$ and the reply is $\True$; otherwise,
nothing changes and the reply is $\False$;
\item
$\clrfield{s}{f}$:
if the content of spot $s$ is an atomic object and $f$ belongs to the
fields of that atomic object, then the content of that field becomes
undefined and the reply is $\True$; otherwise, nothing changes and the
reply is $\False$;
\item
$\getfield{s}{t}{f}$:
if the content of spot $t$ is an atomic object and $f$ belongs to the
fields of that atomic object, then the content of spot $s$ becomes the
same as the content of that field and the reply is $\True$; otherwise,
nothing changes and the reply is~$\False$.
\end{itemize}
In the explanation given above, wherever we say that the content of a
spot or field becomes the same as the content of another spot or field,
this is meant to imply that the former content becomes undefined if the
latter content is undefined.
If not all spots and fields involved in performing a basic action of 
\DLD-K are spots and fields with regard to which the current state is 
locally deterministic, there is no state change and the reply is 
$\False$.

The choice is made to deal uniformly with all cases in which not all 
spots and fields involved in performing a basic action of \DLD-K are 
spots and fields with regard to which the current state is locally 
deterministic.
However, if the field involved in performing a remove field action or a 
has field action is not a field with regard to which the current state 
is locally deterministic, there are other imaginable ways to deal with 
it.
For example, the state change and reply could be the same as in the case 
where the field involved is a field with regard to which the current 
state is locally deterministic.

Recall that, in the examples given in this paper, we take 
$\set{\ul{n} \where n \in \set{0,\ldots,9}}$ for $\AtObj$.
Moreover, we take the choice function $\cf$ such that $\cf(A) = \ul{n}$
if and only if $\ul{n} \in A$ and there does not exist an $n' < n$ such 
that $\ul{n'} \in A$.

\begin{example}
\label{example-DLD-K-actions}
We consider the closed \DLA\ term
\begin{ldispl}
\phantom{(}
 \slinkp{r}{\ul{0}} \dlcom
 \flinkp{\ul{0}}{\nm{up}}{\ul{1}} \dlcom
 \flinkp{\ul{1}}{\nm{dn}}{\ul{0}} \dlcom
 \flinkp{\ul{1}}{\nm{up}}{\ul{2}} \dlcom
 \flinkp{\ul{2}}{\nm{dn}}{\ul{1}} \dlcom
 \slinkp{s}{\ul{2}}
\end{ldispl}%
from Example~\ref{example-DLA-term} again.
The data linkage denoted by this term can be obtained from the empty 
data linkage by performing
\begin{itemize}
\item
first $\getatobj{r}$ and $\setspot{s}{r}$ in that order, yielding
$\slinkp{r}{\ul{0}} \dlcom \slinkp{s}{\ul{0}}$;
\item
next $\setspot{t}{s}$, $\getatobj{s}$, $\addfield{t}{\nm{up}}$, 
$\addfield{s}{\nm{dn}}$, $\setfield{t}{\nm{up}}{s}$ and 
$\setfield{s}{\nm{dn}}{t}$ in that order, yielding \\
$\slinkp{r}{\ul{0}} \dlcom
 \slinkp{s}{\ul{1}} \dlcom
 \slinkp{t}{\ul{0}} \dlcom
 \flinkp{\ul{0}}{\nm{up}}{\ul{1}} \dlcom
 \flinkp{\ul{1}}{\nm{dn}}{\ul{0}}$;
\item
then again $\setspot{t}{s}$, $\getatobj{s}$, $\addfield{t}{\nm{up}}$, 
$\addfield{s}{\nm{dn}}$, $\setfield{t}{\nm{up}}{s}$ and 
$\setfield{s}{\nm{dn}}{t}$ in that order, yielding \\
$\slinkp{r}{\ul{0}} \dlcom
 \slinkp{s}{\ul{2}} \dlcom
 \slinkp{t}{\ul{1}} \dlcom
 \flinkp{\ul{0}}{\nm{up}}{\ul{1}} \dlcom
 \flinkp{\ul{1}}{\nm{dn}}{\ul{0}} \dlcom
 \flinkp{\ul{1}}{\nm{up}}{\ul{2}} \dlcom
 \flinkp{\ul{2}}{\nm{dn}}{\ul{1}}$;
\item
finally $\clrspot{t}$.
\end{itemize}
\end{example}

The priority rewrite system for \DLD-K given below is a many-sorted 
priority rewrite system.
In addition to the sort $\DaLi$ of data linkages, it has the sort
$\Reply$ of replies.
Because this priority rewrite system is used to describe the state 
changes and replies that result from performing the basic actions of 
\DLD-K, it has for each basic action $\alpha \in \Act_\DLDK$, the unary 
\emph{effect} operator $\funct{\effop{\alpha}}{\DaLi}{\DaLi}$ and the 
unary \emph{yield} operator $\funct{\yldop{\alpha}}{\DaLi}{\Reply}$.
The intuition is that these operators stand for operations that give,
for each state $L$, the state and reply, respectively, that result from
performing basic action $\alpha$ in state $L$.
Moreover, the priority rewrite system has the following two constants of 
sort $\Reply$: $\True$ and $\False$.

In the priority rewrite system for \DLD-K given below, the function
$\atobj$ is used to restrict the basic terms over \DLA\ for which the
syntactical variable $L$ stands.
This function gives, for each basic term $L$ over \DLA, the set of
atomic objects occurring in $L$.
It is defined as follows:
\begin{ldispl}
\begin{geqns}
\atobj(\slink{s}{a}) = \set{a}\;, \\
\smash{\atobj(\pflink{a}{f}) = \set{a}\;,} \\
\smash{\atobj(\flink{a}{f}{b}) = \set{a,b}\;,}
\end{geqns}
\qquad
\begin{geqns}
\atobj(\valass{a}{n}) = \set{a}\;, \\
\atobj(\emptydl) = \emptyset\;, \\
\atobj(L \dlcom L') = \atobj(L) \union \atobj(L')\;.
\end{geqns}
\end{ldispl}%

The priority rewrite system for \DLD-K consists of the axioms of \DLA,
with the exception of the associativity, commutativity and identity
axioms for $\dlcom$, taken as rewrite rules and the rewrite rules for
the effect and yield operators given in Table~\ref{rewrite-eff-yld-K}.%
\begin{table}[p]
\caption{Rewrite rules for effect and yield operators}
\label{rewrite-eff-yld-K}
\begin{eqntbl}
\begin{sprucol}
\bprio{1} &
\effop{\getatobj{s}}(X \dlcom \slinkp{s}{a} \dlcom \slinkp{s}{b}) =
X \dlcom \slinkp{s}{a} \dlcom \slinkp{s}{b}
 & \mif a \neq b
\\
\prio{2} &
\multicolumn{2}{@{}l@{}}
{\effop{\getatobj{s}}(L) = L \dlori \slinkp{s}{a}
 \;\,\mathsf{where}\; a = \cf(\AtObj \diff \atobj(L))
  \hfill \mif \atobj(L) \subset \AtObj}
\\
\prio{2} &
\multicolumn{2}{@{}l@{}}
{\effop{\getatobj{s}}(L) = L
  \hfill \mif \atobj(L) = \AtObj}
\\
\bprio{1} &
\effop{\setspot{s}{t}}(X \dlcom \slinkp{s}{a} \dlcom \slinkp{s}{b}) =
X \dlcom \slinkp{s}{a} \dlcom \slinkp{s}{b}
 & \mif a \neq b
\\
\prio{1} &
\effop{\setspot{s}{t}}(X \dlcom \slinkp{t}{a} \dlcom \slinkp{t}{b}) =
X \dlcom \slinkp{t}{a} \dlcom \slinkp{t}{b}
 & \mif a \neq b
\\
\prio{2} &
\effop{\setspot{s}{t}}(X \dlcom \slinkp{t}{a}) =
(X \dlcom \slinkp{t}{a}) \dlori \slinkp{s}{a}
\\
\prio{3} &
\effop{\setspot{s}{t}}(X \dlcom \slinkp{s}{a}) = \effop{\setspot{s}{t}}(X)
\\
\prio{4} &
\effop{\setspot{s}{t}}(X) = X
\\
\bprio{1} &
\effop{\clrspot{s}}(X \dlcom \slinkp{s}{a} \dlcom \slinkp{s}{b}) =
X \dlcom \slinkp{s}{a} \dlcom \slinkp{s}{b}
 & \mif a \neq b
\\
\prio{2} &
\effop{\clrspot{s}}(X \dlcom \slinkp{s}{a}) = \effop{\clrspot{s}}(X)
\\
\prio{3} &
\effop{\clrspot{s}}(X) = X
\\
\bprio{1} &
\effop{\equaltst{s}{t}}(X) = X
\\
\bprio{1} &
\effop{\undeftst{s}}(X) = X
\\
\bprio{1} &
\effop{\addfield{s}{f}}
 (X \dlcom \slinkp{s}{a} \dlcom \slinkp{s}{b}) =
X \dlcom \slinkp{s}{a} \dlcom \slinkp{s}{b}
 & \mif a \neq b
\\
\prio{2} &
\effop{\addfield{s}{f}}
 (X \dlcom \slinkp{s}{a} \dlcom \flinkp{a}{f}{b}) =
X \dlcom \slinkp{s}{a} \dlcom \flinkp{a}{f}{b}
\\
\prio{2} &
\effop{\addfield{s}{f}}
 (X \dlcom \slinkp{s}{a} \dlcom \pflinkp{a}{f}) =
X \dlcom \slinkp{s}{a} \dlcom \pflinkp{a}{f}
\\
\prio{3} &
\effop{\addfield{s}{f}}(X \dlcom \slinkp{s}{a}) =
(X \dlcom \slinkp{s}{a}) \dlori \pflinkp{a}{f}
\\
\prio{4} &
\effop{\addfield{s}{f}}(X) = X
\\
\bprio{1} &
\effop{\rmvfield{s}{f}}
 (X \dlcom \slinkp{s}{a} \dlcom \slinkp{s}{b}) =
X \dlcom \slinkp{s}{a} \dlcom \slinkp{s}{b}
 & \mif a \neq b
\\
\prio{1} &
\effop{\rmvfield{s}{f}}
 (X \dlcom \slinkp{s}{a} \dlcom \flinkp{a}{f}{b} \dlcom \flinkp{a}{f}{c}) =
X \dlcom \slinkp{s}{a} \dlcom \flinkp{a}{f}{b} \dlcom \flinkp{a}{f}{c}
 & \mif b \neq c
\\
\prio{1} &
\effop{\rmvfield{s}{f}}
 (X \dlcom \slinkp{s}{a} \dlcom \flinkp{a}{f}{b} \dlcom \pflinkp{a}{f}) =
X \dlcom \slinkp{s}{a} \dlcom \flinkp{a}{f}{b} \dlcom \pflinkp{a}{f}
\\
\prio{2} &
\effop{\rmvfield{s}{f}}
 (X \dlcom \slinkp{s}{a} \dlcom \flinkp{a}{f}{b}) =
\effop{\rmvfield{s}{f}}(X \dlcom \slinkp{s}{a})
\\
\prio{2} &
\effop{\rmvfield{s}{f}}
 (X \dlcom \slinkp{s}{a} \dlcom \pflinkp{a}{f}) =
\effop{\rmvfield{s}{f}}(X \dlcom \slinkp{s}{a})
\\
\prio{3} &
\effop{\rmvfield{s}{f}}(X) = X
\\
\bprio{1} &
\effop{\hasfield{s}{f}}(X) = X
\\
\bprio{1} &
\effop{\setfield{s}{f}{t}}
 (X \dlcom \slinkp{s}{a} \dlcom \slinkp{s}{b}) =
X \dlcom \slinkp{s}{a} \dlcom \slinkp{s}{b}
 & \mif a \neq b
\\
\prio{1} &
\effop{\setfield{s}{f}{t}}
 (X \dlcom \slinkp{s}{a} \dlcom \flinkp{a}{f}{b} \dlcom \flinkp{a}{f}{c}) =
X \dlcom \slinkp{s}{a} \dlcom \flinkp{a}{f}{b} \dlcom \flinkp{a}{f}{c}
 & \mif b \neq c
\\
\prio{1} &
\effop{\setfield{s}{f}{t}}
 (X \dlcom \slinkp{s}{a} \dlcom \flinkp{a}{f}{b} \dlcom \pflinkp{a}{f}) =
X \dlcom \slinkp{s}{a} \dlcom \flinkp{a}{f}{b} \dlcom \pflinkp{a}{f}
\\
\prio{1} &
\effop{\setfield{s}{f}{t}}
 (X \dlcom \slinkp{t}{a} \dlcom \slinkp{t}{b}) =
X \dlcom \slinkp{t}{a} \dlcom \slinkp{t}{b}
 & \mif a \neq b
\\
\prio{2} &
\effop{\setfield{s}{f}{t}}
 (X \dlcom \slinkp{s}{a} \dlcom \slinkp{t}{b} \dlcom \flinkp{a}{f}{c}) =
(X \dlcom \slinkp{s}{a} \dlcom \slinkp{t}{b}) \dlori \flinkp{a}{f}{b}
\\
\prio{2} & 
\effop{\setfield{s}{f}{t}}
 (X \dlcom \slinkp{s}{a} \dlcom \slinkp{t}{b} \dlcom \pflinkp{a}{f}) =
(X \dlcom \slinkp{s}{a} \dlcom \slinkp{t}{b}) \dlori \flinkp{a}{f}{b}
\\
\prio{3} &
\effop{\setfield{s}{f}{t}}
 (X \dlcom \slinkp{s}{a} \dlcom \flinkp{a}{f}{b}) =
(X \dlcom \slinkp{s}{a})  \dlori \pflinkp{a}{f}
\\
\prio{4} &
\effop{\setfield{s}{f}{t}}(X) = X
\\
\bprio{1} &
\effop{\clrfield{s}{f}}
 (X \dlcom \slinkp{s}{a} \dlcom \slinkp{s}{b}) =
X \dlcom \slinkp{s}{a} \dlcom \slinkp{s}{b}
 & \mif a \neq b
\\
\prio{1} &
\effop{\clrfield{s}{f}}
 (X \dlcom \slinkp{s}{a} \dlcom \flinkp{a}{f}{b} \dlcom \flinkp{a}{f}{c}) =
X \dlcom \slinkp{s}{a} \dlcom \flinkp{a}{f}{b} \dlcom \flinkp{a}{f}{c}
 & \mif b \neq c
\\
\prio{1} &
\effop{\clrfield{s}{f}}
 (X \dlcom \slinkp{s}{a} \dlcom \flinkp{a}{f}{b} \dlcom \pflinkp{a}{f}) =
X \dlcom \slinkp{s}{a} \dlcom \flinkp{a}{f}{b} \dlcom \pflinkp{a}{f}
\\
\prio{2} &
\effop{\clrfield{s}{f}}
 (X \dlcom \slinkp{s}{a} \dlcom \flinkp{a}{f}{b}) =
(X \dlcom \slinkp{s}{a})  \dlori \pflinkp{a}{f}
\\
\eprio{3} &
\effop{\clrfield{s}{f}}(X) = X
\end{sprucol}
\end{eqntbl}
\end{table}%
\addtocounter{table}{-1}%
\begin{table}[p]
\caption{(Continued)}
\begin{eqntbl}
\begin{sprucol}
\bprio{1} &
\effop{\getfield{s}{t}{f}}
 (X \dlcom \slinkp{s}{a} \dlcom \slinkp{s}{b}) =
X \dlcom \slinkp{s}{a} \dlcom \slinkp{s}{b}
 & \mif a \neq b
\\
\prio{1} &
\effop{\getfield{s}{t}{f}}
 (X \dlcom \slinkp{t}{a} \dlcom \slinkp{t}{b}) =
X \dlcom \slinkp{t}{a} \dlcom \slinkp{t}{b}
 & \mif a \neq b
\\
\prio{1} &
\effop{\getfield{s}{t}{f}}
 (X \dlcom \slinkp{t}{a} \dlcom \flinkp{a}{f}{b} \dlcom \flinkp{a}{f}{c}) =
X \dlcom \slinkp{t}{a} \dlcom \flinkp{a}{f}{b} \dlcom \flinkp{a}{f}{c} 
 & \mif b \neq c 
\\
\prio{1} &
\effop{\getfield{s}{t}{f}}
 (X \dlcom \slinkp{t}{a} \dlcom \flinkp{a}{f}{b} \dlcom \pflinkp{a}{f}) =
X \dlcom \slinkp{t}{a} \dlcom \flinkp{a}{f}{b} \dlcom \pflinkp{a}{f}
\\
\prio{2} &
\effop{\getfield{s}{t}{f}}
 (X \dlcom \slinkp{t}{a} \dlcom \flinkp{a}{f}{b}) =
(X \dlcom \slinkp{t}{a} \dlcom \flinkp{a}{f}{b}) \dlori \slinkp{s}{b}
\\
\prio{2} &
\effop{\getfield{s}{t}{f}}
 (X \dlcom \slinkp{t}{a} \dlcom \pflinkp{a}{f} \dlcom \slinkp{s}{c}) =
\effop{\getfield{s}{t}{f}}(X \dlcom \slinkp{t}{a} \dlcom \pflinkp{a}{f})
\\
\prio{3} &
\effop{\getfield{s}{t}{f}}(X) = X
\\
\bprio{1} &
\yldop{\getatobj{s}}(X \dlcom \slinkp{s}{a} \dlcom \slinkp{s}{b}) =
\False
 & \mif a \neq b
\\
\prio{2} &
\multicolumn{2}{@{}l@{}}
{\yldop{\getatobj{s}}(L) = \True  \hfill \mif \atobj(L) \subset \AtObj}
\\
\prio{2} &
\multicolumn{2}{@{}l@{}}
{\yldop{\getatobj{s}}(L) = \False \hfill \mif \atobj(L) = \AtObj}
\\
\bprio{1} &
\yldop{\setspot{s}{t}}(X \dlcom \slinkp{s}{a} \dlcom \slinkp{s}{b}) =
\False
 & \mif a \neq b
\\
\prio{1} &
\yldop{\setspot{s}{t}}(X \dlcom \slinkp{t}{a} \dlcom \slinkp{t}{b}) =
\False
 & \mif a \neq b
\\
\prio{2} &
\yldop{\setspot{s}{t}}(X) = \True
\\
\bprio{1} &
\yldop{\clrspot{s}}(X \dlcom \slinkp{s}{a} \dlcom \slinkp{s}{b}) =
\False
 & \mif a \neq b
\\
\prio{2} &
\yldop{\clrspot{s}}(X) = \True
\\
\bprio{1} &
\yldop{\equaltst{s}{t}}(X \dlcom \slinkp{s}{a} \dlcom \slinkp{s}{b}) =
\False
 & \mif a \neq b
\\
\prio{1} &
\yldop{\equaltst{s}{t}}(X \dlcom \slinkp{t}{a} \dlcom \slinkp{t}{b}) =
\False
 & \mif a \neq b
\\
\prio{2} &
\yldop{\equaltst{s}{t}}(X \dlcom \slinkp{s}{a} \dlcom \slinkp{t}{a}) =
\True
\\
\prio{3} &
\yldop{\equaltst{s}{t}}(X \dlcom \slinkp{s}{a}) = \False
\\
\prio{3} &
\yldop{\equaltst{s}{t}}(X \dlcom \slinkp{t}{a}) = \False
\\
\prio{4} &
\yldop{\equaltst{s}{t}}(X) = \True
\\
\bprio{1} &
\yldop{\undeftst{s}}(X \dlcom \slinkp{s}{a}) = \False
\\
\prio{2} &
\yldop{\undeftst{s}}(X) = \True
\\
\bprio{1} &
\yldop{\addfield{s}{f}}
 (X \dlcom \slinkp{s}{a} \dlcom \slinkp{s}{b}) =
\False
 & \mif a \neq b
\\
\prio{2} &
\yldop{\addfield{s}{f}}
 (X \dlcom \slinkp{s}{a} \dlcom \flinkp{a}{f}{b}) = \False
\\
\prio{2} &
\yldop{\addfield{s}{f}}
 (X \dlcom \slinkp{s}{a} \dlcom \pflinkp{a}{f}) = \False
\\
\prio{3} &
\yldop{\addfield{s}{f}}(X \dlcom \slinkp{s}{a}) = \True
\\
\prio{4} &
\yldop{\addfield{s}{f}}(X) = \False
\\
\bprio{1} &
\yldop{\rmvfield{s}{f}}
 (X \dlcom \slinkp{s}{a} \dlcom \slinkp{s}{b}) =
\False
 & \mif a \neq b
\\
\prio{1} &
\yldop{\rmvfield{s}{f}}
 (X \dlcom \slinkp{s}{a} \dlcom \flinkp{a}{f}{b} \dlcom \flinkp{a}{f}{c}) =
\False
 & \mif b \neq c
\\
\prio{1} &
\yldop{\rmvfield{s}{f}}
 (X \dlcom \slinkp{s}{a} \dlcom \flinkp{a}{f}{b} \dlcom \pflinkp{a}{f}) =
\False
\\
\prio{2} &
\yldop{\rmvfield{s}{f}}
 (X \dlcom \slinkp{s}{a} \dlcom \flinkp{a}{f}{b}) = \True
\\
\prio{2} &
\yldop{\rmvfield{s}{f}}
 (X \dlcom \slinkp{s}{a} \dlcom \pflinkp{a}{f}) = \True
\\
\prio{3} &
\yldop{\rmvfield{s}{f}}(X) = \False
\\
\bprio{1} &
\yldop{\hasfield{s}{f}}
 (X \dlcom \slinkp{s}{a} \dlcom \slinkp{s}{b}) =
\False
 & \mif a \neq b
\\
\prio{1} &
\yldop{\hasfield{s}{f}}
 (X \dlcom \slinkp{s}{a} \dlcom \flinkp{a}{f}{b} \dlcom \flinkp{a}{f}{c}) =
\False
 & \mif b \neq c
\\
\prio{1} &
\yldop{\hasfield{s}{f}}
 (X \dlcom \slinkp{s}{a} \dlcom \flinkp{a}{f}{b} \dlcom \pflinkp{a}{f}) =
\False
\\
\prio{2} &
\yldop{\hasfield{s}{f}}
 (X \dlcom \slinkp{s}{a} \dlcom \flinkp{a}{f}{b}) = \True
\\
\prio{2} &
\yldop{\hasfield{s}{f}}
 (X \dlcom \slinkp{s}{a} \dlcom \pflinkp{a}{f}) = \True
\\
\eprio{3} &
\yldop{\hasfield{s}{f}}(X) = \False
\end{sprucol}
\end{eqntbl}
\end{table}%
\addtocounter{table}{-1}%
\begin{table}[!t]
\caption{(Continued)}
\begin{eqntbl}
\begin{sprucol}
\bprio{1} &
\yldop{\setfield{s}{f}{t}}
 (X \dlcom \slinkp{s}{a} \dlcom \slinkp{s}{b}) =
\False
 & \mif a \neq b
\\
\prio{1} &
\yldop{\setfield{s}{f}{t}}
 (X \dlcom \slinkp{s}{a} \dlcom \flinkp{a}{f}{b} \dlcom \flinkp{a}{f}{c}) =
\False \hsp{12.75}
 & \mif b \neq c
\\
\prio{1} &
\yldop{\setfield{s}{f}{t}}
 (X \dlcom \slinkp{s}{a} \dlcom \flinkp{a}{f}{b} \dlcom \pflinkp{a}{f}) =
\False
\\
\prio{1} &
\yldop{\setfield{s}{f}{t}}
 (X \dlcom \slinkp{t}{a} \dlcom \slinkp{t}{b}) =
\False
 & \mif a \neq b
\\
\prio{2} &
\yldop{\setfield{s}{f}{t}}
 (X \dlcom \slinkp{s}{a} \dlcom \flinkp{a}{f}{b}) = \True
\\
\prio{2} &
\yldop{\setfield{s}{f}{t}}
 (X \dlcom \slinkp{s}{a} \dlcom \pflinkp{a}{f}) = \True
\\
\prio{3} &
\yldop{\setfield{s}{f}{t}}(X) = \False
\\
\bprio{1} &
\yldop{\clrfield{s}{f}}
 (X \dlcom \slinkp{s}{a} \dlcom \slinkp{s}{b}) =
\False
 & \mif a \neq b
\\
\prio{1} &
\yldop{\clrfield{s}{f}}
 (X \dlcom \slinkp{s}{a} \dlcom \flinkp{a}{f}{b} \dlcom \flinkp{a}{f}{c}) =
\False
 & \mif b \neq c
\\
\prio{1} &
\yldop{\clrfield{s}{f}}
 (X \dlcom \slinkp{s}{a} \dlcom \flinkp{a}{f}{b} \dlcom \pflinkp{a}{f}) =
\False
\\
\prio{2} &
\yldop{\clrfield{s}{f}}
 (X \dlcom \slinkp{s}{a} \dlcom \flinkp{a}{f}{b}) = \True
\\
\prio{2} &
\yldop{\clrfield{s}{f}}
 (X \dlcom \slinkp{s}{a} \dlcom \pflinkp{a}{f}) = \True
\\
\prio{3} &
\yldop{\clrfield{s}{f}}(X) = \False
\\
\bprio{1} &
\yldop{\getfield{s}{t}{f}}
 (X \dlcom \slinkp{s}{a} \dlcom \slinkp{s}{b}) =
\False
 & \mif a \neq b
\\
\prio{1} &
\yldop{\getfield{s}{t}{f}}
 (X \dlcom \slinkp{t}{a} \dlcom \slinkp{t}{b}) =
\False
 & \mif a \neq b
\\
\prio{1} &
\yldop{\getfield{s}{t}{f}}
 (X \dlcom \slinkp{t}{a} \dlcom \flinkp{a}{f}{b} \dlcom \flinkp{a}{f}{c}) =
\False
 & \mif b \neq c
\\
\prio{1} &
\yldop{\getfield{s}{t}{f}}
 (X \dlcom \slinkp{t}{a} \dlcom \flinkp{a}{f}{b} \dlcom \pflinkp{a}{f}) =
\False
\\
\prio{2} &
\yldop{\getfield{s}{t}{f}}
 (X \dlcom \slinkp{t}{a} \dlcom \flinkp{a}{f}{b}) = \True
\\
\prio{2} &
\yldop{\getfield{s}{t}{f}}
 (X \dlcom \slinkp{t}{a} \dlcom \pflinkp{a}{f}) = \True
\\
\eprio{3} &
\yldop{\getfield{s}{t}{f}}(X) = \False
\end{sprucol}
\end{eqntbl}
\end{table}%
\pagebreak[2]
In this table, $L$ stands for an arbitrary basic term over \DLA,
$s$ and $t$ stand for arbitrary spots from $\Spot$,
$f$ stands for an arbitrary field from $\Field$, and
$a$, $b$ and $c$ stand for arbitrary atomic objects from $\AtObj$.
Each of the rewrite rules in Table~\ref{rewrite-eff-yld-K} is 
incomparable with each of the axioms of \DLA\ that are taken as rewrite 
rules.
Moreover, the axioms of \DLA\ that are taken as rewrite rules are 
mutually incomparable.

In Section~\ref{sect-properties}, we will state some properties of the 
priority rewrite system for \DLD.
It is obvious from the proofs that the properties concerned are 
properties of the priority rewrite system for \DLD-K as well. 
Among the latter properties is the well-definedness of the priority 
rewrite system for \DLD-K.
If it would not have this property, the priority rewrite system for 
\DLD-K would not determine a one-step reduction relation.

Henceforth, we will write \ACI\ for the set of equations that consists
of the associativity, commutativity and identity axioms for $\dlcom$.%
\footnote
{The mnemonic name \ACI\ for the associativity, commutativity and 
identity axioms for some operator is taken from~\cite{JM92a}.}
Because there are equal \DLD-K terms that cannot be rewritten to the 
same term once the equations in \ACI\ are only used in one direction, 
reduction modulo \ACI\ is of importance to \DLD-K.
Thus, the one-step reduction relation of interest for \DLD-K is the 
one-step reduction relation modulo \ACI\ determined by the priority 
rewrite system for \DLD-K.
We will write $\msredaci$ for the closure of this reduction relation 
under transitivity and reflexivity.
Notice that \ACI\ does not contain the idempotency axiom for $\dlcom$.
This axiom is taken as rewrite rule because it is only needed in the 
direction from left to right.

\begin{example}
\label{example-DLD-K-rewrite-rules}
The statement that the data linkage denoted by 
\begin{ldispl}
\slinkp{r}{\ul{0}} \dlcom \flinkp{\ul{0}}{\nm{up}}{\ul{1}} \dlcom
\flinkp{\ul{1}}{\nm{dn}}{\ul{0}}
\end{ldispl}%
can be obtained from the empty data linkage by performing 
$\getatobj{r}$, $\getatobj{t}$, $\addfield{r}{\nm{up}}$, 
$\addfield{t}{\nm{dn}}$, $\setfield{r}{\nm{up}}{t}$, 
$\setfield{t}{\nm{dn}}{r}$ and $\clrspot{t}$ in that order
is substantiated by the priority rewrite system for \DLD-K, where it is 
provable that
\begin{ldispl}
\effop{\clrspot{t}}(\effop{\setfield{t}{\nm{dn}}{r}}
 (\effop{\setfield{r}{\nm{up}}{t}}(\effop{\addfield{t}{\nm{dn}}}
  (\effop{\addfield{r}{\nm{up}}}(\effop{\getatobj{t}}
   (\effop{\getatobj{r}}(\emptydl))))))) \msredaci 
\\ \quad
\slinkp{r}{\ul{0}} \dlcom \flinkp{\ul{0}}{\nm{up}}{\ul{1}} \dlcom
\flinkp{\ul{1}}{\nm{dn}}{\ul{0}}\;.
\end{ldispl}%
\end{example}

\section{Data Linkage Dynamics}
\label{sect-DLD}

In this section, we extend \DLD-K with features to deal with values 
found in dynamically structured data, resulting in \DLD.
That is, we add basic actions by means of which calculations can be done 
with values that are associated with atomic objects to the basic actions 
of \DLD-K.

Unlike in \DLD-K, it is assumed that a fixed but arbitrary finite meadow 
has been given and that $\Value$ consists of the elements of that meadow.
A meadow is a commutative ring with a multiplicative identity element 
$1$ and a total multiplicative inverse operation \mbox{${}^{-1}$} 
satisfying the reflexivity equation $(u^{-1})^{-1} = u $ and the 
restricted inverse equation $u \mul (u \mul u^{-1}) = u$.
Thus, a meadow has an additive identity element $0$, a multiplicative 
identity element $1$, an addition operation \mbox{${} + {}$}, a 
multiplication operation \mbox{${} \cdot {}$}, an additive inverse 
operation \mbox{$- {}$}, and a multiplicative inverse operation 
\mbox{${}^{-1}$} that satisfies $0^{-1} = 0$.
Meadows were defined for the first time in~\cite{BT07a} and elaborated
in several subsequent papers (see e.g.~\cite{BBP13a,BR07a}).
The prime examples of finite meadows are finite fields with the 
multiplicative inverse operation made total by imposing that the 
multiplicative inverse of zero is zero.

\DLD\ has the basic actions of \DLD-K and in addition the following 
basic actions:
\begin{itemize}
\item
for each $s \in \Spot$, an \emph{assign zero action} $\asszero{s}$;
\item
for each $s \in \Spot$, an \emph{assign one action} $\assone{s}$;
\item
for each $s,t,u \in \Spot$, an \emph{assign sum action}
$\assadd{s}{t}{u}$;
\item
for each $s,t,u \in \Spot$, an \emph{assign product action}
$\assmul{s}{t}{u}$;
\item
for each $s,t \in \Spot$, an \emph{assign additive inverse action}
$\assneg{s}{t}$;
\item
for each $s,t \in \Spot$, an \emph{assign multiplicative inverse action}
$\assinv{s}{t}$;
\item
for each $s,t \in \Spot$, a \emph{value equality test action}
$\eqvaltst{s}{t}$;
\item
for each $s \in \Spot$, a \emph{value undefinedness test action}
$\undefvtst{s}$.
\end{itemize}
We write $\Act_\DLD$ for the set of all basic actions of \DLD.

When speaking informally about a state $L$ of \DLD, we also say:
\begin{itemize}
\item
if $L$ is locally deterministic with regard to the value assignment for
$a$, \emph{the value assigned to atomic object $a$} instead of the 
unique value $n$ for which $\valass{a}{n}$ is contained in $L$;
\item
\emph{atomic object $a$ has a value assigned} instead of $L$ is locally 
deterministic with regard to the value assignment for $a$.
\end{itemize}

The value-related basic actions of \DLD\ can be explained as follows
if all spots and value assignments involved in performing them are spots 
and value assignments with regard to which the current state is locally 
deterministic:
\begin{itemize}
\item
$\asszero{s}$:
if the content of spot $s$ is an atomic object, then the value assigned
to that atomic object becomes $0$ and the reply is $\True$; otherwise,
nothing changes and the reply is $\False$;
\item
$\assone{s}$:
if the content of spot $s$ is an atomic object, then the value assigned
to that atomic object becomes $1$ and the reply is $\True$; otherwise,
nothing changes and the reply is $\False$;
\item
$\assadd{s}{t}{u}$:
if the content of spot $s$ is an atomic object and the contents of spots
$t$ and $u$ are atomic objects that have values assigned, then the value
assigned to the content of spot $s$ becomes the sum of the values
assigned to the contents of spots $t$ and $u$ and the reply is $\True$;
otherwise, nothing changes and the reply is $\False$;
\item
$\assmul{s}{t}{u}$:
if the content of spot $s$ is an atomic object and the contents of spots
$t$ and $u$ are atomic objects that have values assigned, then the value
assigned to the content of spot $s$ becomes the product of the values
assigned to the contents of spots $t$ and $u$ and the reply is $\True$;
otherwise, nothing changes and the reply is $\False$;
\item
$\assneg{s}{t}$:
if the content of spot $s$ is an atomic object and the content of spot
$t$ is an atomic object that has a value assigned, then the value
assigned to the content of spot $s$ becomes the additive inverse of the
value assigned to the content of spot $t$ and the reply is $\True$;
otherwise, nothing changes and the reply is $\False$;
\item
$\assinv{s}{t}$:
if the content of spot $s$ is an atomic object and the content of spot
$t$ is an atomic object that has a value assigned, then the value
assigned to the content of spot $s$ becomes the multiplicative inverse
of the value assigned to the content of spot $t$ and the reply is
$\True$; otherwise, nothing changes and the reply is $\False$;
\item
$\eqvaltst{s}{t}$:
if the contents of spots $s$ and $t$ are atomic objects that have values
assigned and the value assigned to the content of spot $s$ equals the
value assigned to the content of spot $t$, then nothing changes and the
reply is $\True$; otherwise, nothing changes and the reply is $\False$;
\item
$\undefvtst{s}$:
if the content of spot $s$ is an atomic object that has no value
assigned, then nothing changes and the reply is $\True$; otherwise,
nothing changes and the reply is $\False$.
\end{itemize}
If not all spots and value assignments involved in performing a 
value-related basic action are spots and value assignments with regard 
to which the current state is locally deterministic, there is no state 
change and the reply is $\False$.

Notice that copying, subtraction, and division can be done with the
value-related basic actions available in \DLD.
If the content of spot $s$ is an atomic object and the content of spot
$t$ is an atomic object that has a value assigned, then that value can
be assigned to the content of spot $s$ by first performing $\asszero{s}$
and then performing $\assadd{s}{s}{t}$.
If the content of spot $s$ is an atomic object and the contents of spots
$t$ and $u$ are atomic objects that have values assigned, then the
difference of those values can be assigned to the content of spot $s$
by first performing $\assneg{u}{u}$, next performing $\assadd{s}{t}{u}$
and then performing $\assneg{u}{u}$ once again.
Division can be done like subtraction.

\begin{example}
\label{example-DLD-actions-values}
We consider the data linkage denoted by the following closed \DLA\ term:
\begin{ldispl}
\slinkp{s}{\ul{0}} \dlcom
       \slinkp{t}{\ul{1}} \dlcom \valass{\ul{1}}{7} \dlcom
       \slinkp{u}{\ul{2}} \dlcom \valass{\ul{2}}{3}\;.
\end{ldispl}%
\sloppy
The data linkage obtained from this data linkage by first performing 
$\assneg{u}{u}$, next performing $\assadd{s}{t}{u}$ and then performing 
$\assneg{u}{u}$ once again is the data linkage denoted by the term
\begin{ldispl}
\slinkp{s}{\ul{0}} \dlcom \valass{\ul{0}}{4} \dlcom
       \slinkp{t}{\ul{1}} \dlcom \valass{\ul{1}}{7} \dlcom
       \slinkp{u}{\ul{2}} \dlcom \valass{\ul{2}}{3}\;.
\end{ldispl}%
\end{example}

In \DLD, finite meadows are taken as the basis for the features to deal
with values.
This allows for calculations in finite fields.
The approach followed is generic: take the algebras that are the models
of some set of equational axioms and introduce value-related basic
actions for the operations of those algebras.

The priority rewrite system for \DLD\ consists of the axioms of \DLA,
with the exception of the associativity, commutativity and identity
axioms for $\dlcom$, taken as rewrite rules, the rewrite rules from
the priority rewrite system for \DLD-K, and the rewrite rules given in 
Table~\ref{rewrite-eff-yld}.%
\begin{table}[p]
\caption{Rewrite rules for additional effect and yield operators}
\label{rewrite-eff-yld}
\begin{eqntbl}
\begin{sprucol}
\bprio{1} &
\effop{\asszero{s}}
 (X \dlcom \slinkp{s}{a} \dlcom \slinkp{s}{b}) =
X \dlcom \slinkp{s}{a} \dlcom \slinkp{s}{a}
 & \mif a \neq b
\\
\prio{1} &
\effop{\asszero{s}}
 (X \dlcom \slinkp{s}{a} \dlcom \valass{a}{n} \dlcom \valass{a}{m}) =
X \dlcom \slinkp{s}{a} \dlcom \valass{a}{n} \dlcom \valass{a}{m}
 & \mif n \neq m
\\
\prio{2} &
\effop{\asszero{s}}(X \dlcom \slinkp{s}{a}) =
(X \dlcom \slinkp{s}{a}) \dlori \valass{a}{0}
\\
\prio{3} &
\effop{\asszero{s}}(X) = X
\\
\bprio{1} &
\effop{\assone{s}}
 (X \dlcom \slinkp{s}{a} \dlcom \slinkp{s}{b}) =
X \dlcom \slinkp{s}{a} \dlcom \slinkp{s}{a}
 & \mif a \neq b
\\
\prio{1} &
\effop{\assone{s}}
 (X \dlcom \slinkp{s}{a} \dlcom \valass{a}{n} \dlcom \valass{a}{m}) =
X \dlcom \slinkp{s}{a} \dlcom \valass{a}{n} \dlcom \valass{a}{m}
 & \mif n \neq m
\\
\prio{2} &
\effop{\assone{s}}(X \dlcom \slinkp{s}{a}) =
(X \dlcom \slinkp{s}{a}) \dlori \valass{a}{1}
\\
\prio{3} &
\effop{\assone{s}}(X) = X
\\
\bprio{1} &
\effop{\assadd{s}{t}{u}}
 (X \dlcom \slinkp{s}{a} \dlcom \slinkp{s}{b}) =
X \dlcom \slinkp{s}{a} \dlcom \slinkp{s}{b}
 & \mif a \neq b
\\
\prio{1} &
\effop{\assadd{s}{t}{u}}
 (X \dlcom \slinkp{s}{a} \dlcom \valass{a}{n} \dlcom \valass{a}{m}) =
X \dlcom \slinkp{s}{a} \dlcom \valass{a}{n} \dlcom \valass{a}{m}
 & \mif n \neq m
\\
\prio{1} &
\effop{\assadd{s}{t}{u}}
 (X \dlcom \slinkp{t}{a} \dlcom \slinkp{t}{b}) =
X \dlcom \slinkp{t}{a} \dlcom \slinkp{t}{b}
 & \mif a \neq b
\\
\prio{1} &
\effop{\assadd{s}{t}{u}}
 (X \dlcom \slinkp{t}{a} \dlcom \valass{a}{n} \dlcom \valass{a}{m}) =
X \dlcom \slinkp{t}{a} \dlcom \valass{a}{n} \dlcom \valass{a}{m}
 & \mif n \neq m
\\
\prio{1} &
\effop{\assadd{s}{t}{u}}
 (X \dlcom \slinkp{u}{a} \dlcom \slinkp{u}{b}) =
X \dlcom \slinkp{u}{a} \dlcom \slinkp{u}{b}
 & \mif a \neq b
\\
\prio{1} &
\effop{\assadd{s}{t}{u}}
 (X \dlcom \slinkp{u}{a} \dlcom \valass{a}{n} \dlcom \valass{a}{m}) =
X \dlcom \slinkp{u}{a} \dlcom \valass{a}{n} \dlcom \valass{a}{m}
 & \mif n \neq m
\\
\prio{2} &
\effop{\assadd{s}{t}{u}}
 (X \dlcom \slinkp{s}{a} \dlcom
  \slinkp{t}{b} \dlcom \valass{b}{n} \dlcom
  \slinkp{u}{c} \dlcom \valass{c}{m}) = {} \\ & \qquad
(X \dlcom \slinkp{s}{a} \dlcom
 \slinkp{t}{b} \dlcom \valass{b}{n} \dlcom
 \slinkp{u}{c} \dlcom \valass{c}{m}) \dlori \valass{a}{n + m}
\\
\prio{3} &
\effop{\assadd{s}{t}{u}}(X) = X
\\
\bprio{1} &
\effop{\assmul{s}{t}{u}}
 (X \dlcom \slinkp{s}{a} \dlcom \slinkp{s}{b}) =
X \dlcom \slinkp{s}{a} \dlcom \slinkp{s}{b}
 & \mif a \neq b
\\
\prio{1} &
\effop{\assmul{s}{t}{u}}
 (X \dlcom \slinkp{s}{a} \dlcom \valass{a}{n} \dlcom \valass{a}{m}) =
X \dlcom \slinkp{s}{a} \dlcom \valass{a}{n} \dlcom \valass{a}{m}
 & \mif n \neq m
\\
\prio{1} &
\effop{\assmul{s}{t}{u}}
 (X \dlcom \slinkp{t}{a} \dlcom \slinkp{t}{b}) =
X \dlcom \slinkp{t}{a} \dlcom \slinkp{t}{b}
 & \mif a \neq b
\\
\prio{1} &
\effop{\assmul{s}{t}{u}}
 (X \dlcom \slinkp{t}{a} \dlcom \valass{a}{n} \dlcom \valass{a}{m}) =
X \dlcom \slinkp{t}{a} \dlcom \valass{a}{n} \dlcom \valass{a}{m}
 & \mif n \neq m
\\
\prio{1} &
\effop{\assmul{s}{t}{u}}
 (X \dlcom \slinkp{u}{a} \dlcom \slinkp{u}{b}) =
X \dlcom \slinkp{u}{a} \dlcom \slinkp{u}{b}
 & \mif a \neq b
\\
\prio{1} &
\effop{\assmul{s}{t}{u}}
 (X \dlcom \slinkp{u}{a} \dlcom \valass{a}{n} \dlcom \valass{a}{m}) =
X \dlcom \slinkp{u}{a} \dlcom \valass{a}{n} \dlcom \valass{a}{m}
 & \mif n \neq m
\\
\prio{2} &
\effop{\assmul{s}{t}{u}}
 (X \dlcom \slinkp{s}{a} \dlcom
  \slinkp{t}{b} \dlcom \valass{b}{n} \dlcom
  \slinkp{u}{c} \dlcom \valass{c}{m}) = {} \\ & \qquad
(X \dlcom \slinkp{s}{a} \dlcom
 \slinkp{t}{b} \dlcom \valass{b}{n} \dlcom
 \slinkp{u}{c} \dlcom \valass{c}{m}) \dlori \valass{a}{n \cdot m}
\\
\prio{3} &
\effop{\assmul{s}{t}{u}}(X) = X
\\
\bprio{1} &
\effop{\assneg{s}{t}}
 (X \dlcom \slinkp{s}{a} \dlcom \slinkp{s}{b}) =
X \dlcom \slinkp{s}{a} \dlcom \slinkp{s}{b}
 & \mif a \neq b
\\
\prio{1} &
\effop{\assneg{s}{t}}
 (X \dlcom \slinkp{s}{a} \dlcom \valass{a}{n} \dlcom \valass{a}{m}) =
X \dlcom \slinkp{s}{a} \dlcom \valass{a}{n} \dlcom \valass{a}{m}
 & \mif n \neq m
\\
\prio{1} &
\effop{\assneg{s}{t}}
 (X \dlcom \slinkp{t}{a} \dlcom \slinkp{t}{b}) =
X \dlcom \slinkp{t}{a} \dlcom \slinkp{t}{b}
 & \mif a \neq b
\\
\prio{1} &
\effop{\assneg{s}{t}}
 (X \dlcom \slinkp{t}{a} \dlcom \valass{a}{n} \dlcom \valass{a}{m}) =
X \dlcom \slinkp{t}{a} \dlcom \valass{a}{n} \dlcom \valass{a}{m}
 & \mif n \neq m
\\
\prio{2} &
\effop{\assneg{s}{t}}
 (X \dlcom \slinkp{s}{a} \dlcom
  \slinkp{t}{b} \dlcom \valass{b}{n}) =
(X \dlcom \slinkp{s}{a} \dlcom
 \slinkp{t}{b} \dlcom \valass{b}{n}) \dlori \valass{a}{-n} \hsp{-1.9}
\\
\prio{3} &
\effop{\assneg{s}{t}}(X) = X
\\
\bprio{1} &
\effop{\assinv{s}{t}}
 (X \dlcom \slinkp{s}{a} \dlcom \slinkp{s}{b}) =
X \dlcom \slinkp{s}{a} \dlcom \slinkp{s}{b}
 & \mif a \neq b
\\
\prio{1} &
\effop{\assinv{s}{t}}
 (X \dlcom \slinkp{s}{a} \dlcom \valass{a}{n} \dlcom \valass{a}{m}) =
X \dlcom \slinkp{s}{a} \dlcom \valass{a}{n} \dlcom \valass{a}{m}
 & \mif n \neq m
\\
\prio{1} &
\effop{\assinv{s}{t}}
 (X \dlcom \slinkp{t}{a} \dlcom \slinkp{t}{b}) =
X \dlcom \slinkp{t}{a} \dlcom \slinkp{t}{b}
 & \mif a \neq b
\\
\prio{1} &
\effop{\assinv{s}{t}}
 (X \dlcom \slinkp{t}{a} \dlcom \valass{a}{n} \dlcom \valass{a}{m}) =
X \dlcom \slinkp{t}{a} \dlcom \valass{a}{n} \dlcom \valass{a}{m}
\hsp{2} & \mif n \neq m
\\
\prio{2} &
\effop{\assinv{s}{t}}
 (X \dlcom \slinkp{s}{a} \dlcom
  \slinkp{t}{b} \dlcom \valass{b}{n}) =
(X \dlcom \slinkp{s}{a} \dlcom
 \slinkp{t}{b} \dlcom \valass{b}{n}) \dlori \valass{a}{n^{-1}} \hsp{-1.9}
\\
\prio{3} &
\effop{\assinv{s}{t}}(X) = X
\\
\bprio{1} &
\effop{\eqvaltst{s}{t}}(X) = X
\\
\beprio{1} &
\effop{\undefvtst{s}}(X) = X
\end{sprucol}
\end{eqntbl}
\end{table}%
\addtocounter{table}{-1}%
\begin{table}[p]
\caption{(Continued)}
\begin{eqntbl}
\begin{sprucol}
\bprio{1} &
\yldop{\asszero{s}}
 (X \dlcom \slinkp{s}{a} \dlcom \slinkp{s}{b}) =
\False
 & \mif a \neq b
\\
\prio{1} &
\yldop{\asszero{s}}
 (X \dlcom \slinkp{s}{a} \dlcom \valass{a}{n} \dlcom \valass{a}{m}) =
\False
 & \mif n \neq m
\\
\prio{2} &
\yldop{\asszero{s}}(X \dlcom \slinkp{s}{a}) = \True
\\
\prio{3} &
\yldop{\asszero{s}}(X) = \False
\\
\bprio{1} &
\yldop{\assone{s}}
 (X \dlcom \slinkp{s}{a} \dlcom \slinkp{s}{b}) =
\False
 & \mif a \neq b
\\
\prio{1} &
\yldop{\assone{s}}
 (X \dlcom \slinkp{s}{a} \dlcom \valass{a}{n} \dlcom \valass{a}{m}) =
\False
 & \mif n \neq m
\\
\prio{2} &
\yldop{\assone{s}}(X \dlcom \slinkp{s}{a}) = \True
\\
\prio{3} &
\yldop{\assone{s}}(X) = \False
\\
\bprio{1} &
\yldop{\assadd{s}{t}{u}}
 (X \dlcom \slinkp{s}{a} \dlcom \slinkp{s}{b}) =
\False
 & \mif a \neq b
\\
\prio{1} &
\yldop{\assadd{s}{t}{u}}
 (X \dlcom \slinkp{s}{a} \dlcom \valass{a}{n} \dlcom \valass{a}{m}) =
\False
 & \mif n \neq m
\\
\prio{1} &
\yldop{\assadd{s}{t}{u}}
 (X \dlcom \slinkp{t}{a} \dlcom \slinkp{t}{b}) =
\False
 & \mif a \neq b
\\
\prio{1} &
\yldop{\assadd{s}{t}{u}}
 (X \dlcom \slinkp{t}{a} \dlcom \valass{a}{n} \dlcom \valass{a}{m}) =
\False
 & \mif n \neq m
\\
\prio{1} &
\yldop{\assadd{s}{t}{u}}
 (X \dlcom \slinkp{u}{a} \dlcom \slinkp{u}{b}) =
\False
 & \mif a \neq b
\\
\prio{1} &
\yldop{\assadd{s}{t}{u}}
 (X \dlcom \slinkp{u}{a} \dlcom \valass{a}{n} \dlcom \valass{a}{m}) =
\False
 & \mif n \neq m
\\
\prio{2} &
\yldop{\assadd{s}{t}{u}}
 (X \dlcom \slinkp{s}{a} \dlcom
  \slinkp{t}{b} \dlcom \valass{b}{n} \dlcom
  \slinkp{u}{c} \dlcom \valass{c}{m}) = \True \hsp{6.4}
\\
\prio{3} &
\yldop{\assadd{s}{t}{u}}(X) = \False
\\
\bprio{1} &
\yldop{\assmul{s}{t}{u}}
 (X \dlcom \slinkp{s}{a} \dlcom \slinkp{s}{b}) =
\False
 & \mif a \neq b
\\
\prio{1} &
\yldop{\assmul{s}{t}{u}}
 (X \dlcom \slinkp{s}{a} \dlcom \valass{a}{n} \dlcom \valass{a}{m}) =
\False
 & \mif n \neq m
\\
\prio{1} &
\yldop{\assmul{s}{t}{u}}
 (X \dlcom \slinkp{t}{a} \dlcom \slinkp{t}{b}) =
\False
 & \mif a \neq b
\\
\prio{1} &
\yldop{\assmul{s}{t}{u}}
 (X \dlcom \slinkp{t}{a} \dlcom \valass{a}{n} \dlcom \valass{a}{m}) =
\False
 & \mif n \neq m
\\
\prio{1} &
\yldop{\assmul{s}{t}{u}}
 (X \dlcom \slinkp{u}{a} \dlcom \slinkp{u}{b}) =
\False
 & \mif a \neq b
\\
\prio{1} &
\yldop{\assmul{s}{t}{u}}
 (X \dlcom \slinkp{u}{a} \dlcom \valass{a}{n} \dlcom \valass{a}{m}) =
\False
 & \mif n \neq m
\\
\prio{2} &
\yldop{\assmul{s}{t}{u}}
 (X \dlcom \slinkp{s}{a} \dlcom
  \slinkp{t}{b} \dlcom \valass{b}{n} \dlcom
  \slinkp{u}{c} \dlcom \valass{c}{m}) = \True
\\
\prio{3} &
\yldop{\assmul{s}{t}{u}}(X) = \False
\\
\bprio{1} &
\yldop{\assneg{s}{t}}
 (X \dlcom \slinkp{s}{a} \dlcom \slinkp{s}{b}) =
\False
 & \mif a \neq b
\\
\prio{1} &
\yldop{\assneg{s}{t}}
 (X \dlcom \slinkp{s}{a} \dlcom \valass{a}{n} \dlcom \valass{a}{m}) =
\False
 & \mif n \neq m
\\
\prio{1} &
\yldop{\assneg{s}{t}}
 (X \dlcom \slinkp{t}{a} \dlcom \slinkp{t}{b}) =
\False
 & \mif a \neq b
\\
\prio{1} &
\yldop{\assneg{s}{t}}
 (X \dlcom \slinkp{t}{a} \dlcom \valass{a}{n} \dlcom \valass{a}{m}) =
\False
 & \mif n \neq m
\\
\prio{2} &
\yldop{\assneg{s}{t}}
 (X \dlcom \slinkp{s}{a} \dlcom \slinkp{t}{b} \dlcom \valass{b}{n}) = \True
\\
\prio{3} &
\yldop{\assneg{s}{t}}(X) = \False
\\
\bprio{1} &
\yldop{\assinv{s}{t}}
 (X \dlcom \slinkp{s}{a} \dlcom \slinkp{s}{b}) =
\False
 & \mif a \neq b
\\
\prio{1} &
\yldop{\assinv{s}{t}}
 (X \dlcom \slinkp{s}{a} \dlcom \valass{a}{n} \dlcom \valass{a}{m}) =
\False
 & \mif n \neq m
\\
\prio{1} &
\yldop{\assinv{s}{t}}
 (X \dlcom \slinkp{t}{a} \dlcom \slinkp{t}{b}) =
\False
 & \mif a \neq b
\\
\prio{1} &
\yldop{\assinv{s}{t}}
 (X \dlcom \slinkp{t}{a} \dlcom \valass{a}{n} \dlcom \valass{a}{m}) =
\False
 & \mif n \neq m
\\
\prio{2} &
\yldop{\assinv{s}{t}}
 (X \dlcom \slinkp{s}{a} \dlcom \slinkp{t}{b} \dlcom \valass{b}{n}) = \True
\\
\prio{3} &
\yldop{\assinv{s}{t}}(X) = \False
\\
\bprio{1} &
\yldop{\eqvaltst{s}{t}}
 (X \dlcom \slinkp{s}{a} \dlcom \slinkp{s}{b}) =
\False
 & \mif a \neq b
\\
\prio{1} &
\yldop{\eqvaltst{s}{t}}
 (X \dlcom \slinkp{s}{a} \dlcom \valass{a}{n} \dlcom \valass{a}{m}) =
\False
 & \mif n \neq m
\\
\prio{1} &
\yldop{\eqvaltst{s}{t}}
 (X \dlcom \slinkp{t}{a} \dlcom \slinkp{t}{b}) =
\False
 & \mif a \neq b
\\
\prio{1} &
\yldop{\eqvaltst{s}{t}}
 (X \dlcom \slinkp{t}{a} \dlcom \valass{a}{n} \dlcom \valass{a}{m}) =
\False
 & \mif n \neq m
\\
\prio{2} &
\yldop{\eqvaltst{s}{t}}
 (X \dlcom \slinkp{s}{a} \dlcom \valass{a}{n} \dlcom
  \slinkp{t}{b} \dlcom \valass{b}{n}) = \True
\\
\eprio{3} &
\yldop{\eqvaltst{s}{t}}(X) = \False
\end{sprucol}
\end{eqntbl}
\end{table}%
\addtocounter{table}{-1}%
\begin{table}[!t]
\caption{(Continued)}
\begin{eqntbl}
\begin{sprucol}
\bprio{1} &
\yldop{\undefvtst{s}}(X \dlcom \slinkp{s}{a} \dlcom \slinkp{s}{b}) = 
\False \hsp{18}
 & \mif a \neq b
\\
\prio{2} &
\yldop{\undefvtst{s}}(X \dlcom \slinkp{s}{a} \dlcom \valass{a}{n}) = 
\False
\\
\prio{3} &
\yldop{\undefvtst{s}}(X \dlcom \slinkp{s}{a}) = \True
\\
\eprio{4} &
\yldop{\undefvtst{s}}(X) = \False
\end{sprucol}
\end{eqntbl}
\end{table}
In this table, $s$, $t$ and $u$ stand for arbitrary spots from $\Spot$,
$a$, $b$ and $c$ stand for arbitrary atomic objects from $\AtObj$,
and $n$ and $m$ stand for arbitrary values from $\Value$.
Each of the rewrite rules in Table~\ref{rewrite-eff-yld} is incomparable
with each of the axioms of \DLA\ that are taken as rewrite rules and each
of the rewrite rules from the priority rewrite system for \DLD-K.

The total number of rewrite rules for \DLD\ is quite large.
This is fully attributable to the fact that \DLD\ has $19$ different 
kinds of basic actions.
The number of rewrite rules for the effect operator $\effop{\alpha}$ for 
a basic action $\alpha$ is on average about $4$, and the number of 
rewrite rules for the yield operator $\yldop{\alpha}$ for a basic action 
$\alpha$ is on average about $5$.
Moreover, the fraction of the rewrite rules that deal with the case in 
which not all spots, fields and/or value assignments involved in 
performing the basic action concerned are ones with regard to which the 
data linkage concerned is locally deterministic is on average somewhat 
greater than $1/2$.
Due to the uniform treatment of this case, the rules in question have 
the forms $\effop{\alpha}(L) = L$ and $\yldop{\alpha}(L) = \False$.
They are labelled by $1$, i.e.\ they have the highest priority.
The remaining rewrite rules reflect the informal explanations of the 
basic actions of \DLD\ given before in a direct way.

\begin{example}
\label{example-DLD-rewrite-rules}
What is stated before about copying and subtraction with the 
value-related basic actions of \DLD\ is substantiated by the priority 
rewrite system for \DLD.
Let
$L =
 M \dlcom \slinkp{s}{a} \dlcom
 \slinkp{t}{b} \dlcom \valass{b}{n}$ be a closed \DLA\ term.
Then there exists a basic term $N$ over \DLA\ such that
\begin{ldispl}
\effop{\assadd{s}{s}{t}}(\effop{\asszero{s}}(L)) \msredaci N\;, \\
\hfill L \dlori \valass{a}{n} \msredaci N\;.
\end{ldispl}%
In other words, by first performing $\asszero{s}$ and then performing
$\assadd{s}{s}{t}$, the value assigned to the content of spot $s$
becomes the same as the value assigned to the content of spot $t$.
Let
$L' =
 M' \dlcom \slinkp{s}{a} \dlcom
 \slinkp{t}{b} \dlcom \valass{b}{n} \dlcom
 \slinkp{u}{c} \dlcom \valass{c}{m}$
be a closed \DLA\ term.
Then there exists a basic term $N'$ over \DLA\ such that
\pagebreak[2]
\begin{ldispl}
\effop{\assneg{u}{u}}
  (\effop{\assadd{s}{t}{u}}(\effop{\assneg{u}{u}}(L'))) \msredaci N'\;,
 \\
\hfill L' \dlori \valass{a}{n - m} \msredaci N'\;.
\end{ldispl}%
\sloppy
In other words, by first performing $\assneg{u}{u}$, next performing
$\assadd{s}{t}{u}$ and then performing $\assneg{u}{u}$ once again, the
value assigned to the content of spot $s$ becomes the difference of the
values assigned to the contents of spots $t$ and $u$.
\end{example}

It is easy to check that for all $\alpha \in \Act_\DLD$, for all closed 
\DLA\ terms $L$ and $L'$ such that $\effop{\alpha}(L) \msredaci L'$, $L$
is deterministic if and only if $L'$ is deterministic.
In other words, both determinism and non-determinism are properties of 
data linkages that are preserved by the basic actions of \DLD.
Because of this, the data linkage obtained from another data linkage by 
performing a number of basic actions of \DLD\ in succession is 
deterministic if and only if the latter data linkage is determistic.

With much effort, a few applications of non-deterministic data linkages
could be thought up, but these applications would require a variant of 
\DLD\ with basic actions of which we have not yet a clear image.
This raises the question why \DLA\ and \DLD\ cover non-deterministic 
data linkages.
The answer is that we have devised several variants of the pair of \DLA\ 
and \DLD\ that are restricted to deterministic data linkages, but they 
turned out to be more complicated than the pair of \DLA\ and \DLD.

The priority rewrite system for \DLD\ is used in
Section~\ref{sect-comb-TA-DLD} in examples concerning computations in
which the basic actions of \DLD\ are involved.

\section{Properties of the Priority Rewrite System for DLD}
\label{sect-properties}

In this section, we state some properties of the priority rewrite system 
for \DLD.
For the purpose of stating the properties in question rigorously, we
introduce the set $\cE$ of \emph{effect terms} and the set $\cY$  of
\emph{yield terms}.
They are inductively defined by the following rules:
\begin{itemize}
\item
$\emptydl \in \cE$;
\item
if $s \in \Spot$ and $a \in \AtObj$, then $\slink{s}{a} \in \cE$;
\item
if $a \in \AtObj$ and $f \in \Field$, then $\pflink{a}{f} \in \cE$;
\item
if $a,b \in \AtObj$ and $f \in \Field$, then $\flink{a}{f}{b} \in \cE$;
\item
if $a \in \AtObj$ and $n \in \Value$, then $\valass{a}{n} \in \cE$;
\item
if $D_1,D_2 \in \cE$, then $D_1 \dlcom D_2 \in \cE$;
\item
if $D_1,D_2 \in \cE$, then $D_1 \dlori D_2 \in \cE$;
\item
if $\alpha \in \Act_\DLD$ and $D \in \cE$, then
$\effop{\alpha}(D) \in \cE$;
\item
if $\alpha \in \Act_\DLD$ and $D \in \cE$, then
$\yldop{\alpha}(D) \in \cY$.
\end{itemize}
Clearly, $\cB$ is a proper subset of $\cE$.
Below we will prove that effect terms have normal forms that are basic 
terms over \DLA.

As a preparation, we state an important property of the underlying term 
rewriting system of the priority rewrite system for \DLD.
\begin{proposition}
[Strongly normalizing underlying term rewriting system]
\label{prop-strongly-normalizing}
The underlying term rewriting system of the priority rewrite system for 
\DLD\ is strongly normalizing modulo \ACI.
\end{proposition}
\begin{proof}
We will write $\cR'_\DLD$ for the underlying term rewriting system of 
the priority rewrite system for \DLD, we will write \AC\ for the set of 
equations that consists of the associativity and commutativity axioms 
for $\dlcom$, and we will write $\osredacu$ and $\osredaciu$ for the 
one-step reduction modulo \AC\ relation of $\cR'_\DLD$ and the one-step 
reduction modulo \ACI\ relation of $\cR'_\DLD$, respectively. 

First, it is proved that $\cR'_\DLD$ is strongly normalizing modulo \AC.
This is easily proved by means of the reduction ordering induced by the 
integer polynomials $\phi(D)$ associated with \DLD\ terms $D$ as 
follows:
\begin{ldispl}
\begin{geqns}
\phi(X) = \ul{X}\;, \\
\phi(\slink{s}{a}) = 3\;, \\
\smash{\phi(\pflink{a}{f}) = 3\;,} \\
\smash{\phi(\flink{a}{f}{b}) = 3\;,} \\
\phi(\valass{a}{n}) = 3\;,
\eqnsep
\phi(\True) = 3\;, \\
\phi(\False) = 3\;,
\end{geqns}
\qquad
\begin{geqns}
\phi(\emptydl) = 2\;, \\
\phi(D_1 \dlcom D_2) = \phi(D_1) + \phi(D_2) + 1\;, \\
\phi(D_1 \dlori D_2) = \phi(D_1) \mul \phi(D_2)\;, \\
\phi(\effop{\alpha}(D)) = 4 \mul \phi(D)\;, \\
\eqnsep
\phi(\yldop{\alpha}(D)) = 4 \mul \phi(D)\;,
\end{geqns}
\end{ldispl}%
where it is assumed that, for each variable $X$ over data linkages,
there is a variable $\ul{X}$ over integers.
Here, it is crucial that $\phi(D_1 \dlcom D_2) = \phi(D_2 \dlcom D_1)$ 
and 
$\phi(D_1 \dlcom (D_2 \dlcom D_3)) = \phi((D_1 \dlcom D_2) \dlcom D_3)$
for all \DLD\ terms $D_1$, $D_2$, $D_3$ (see e.g.\ Section~6.2.3 
from~\cite{Zan03a}).

Next, it is proved by means of a function $\theta$ on \DLD\ terms that 
$\cR'_\DLD$ is strongly normalizing modulo \ACI.
The function $\theta$, which transforms \DLD\ terms to ones for which 
the applicable reduction modulo \ACI\ steps of $\cR'_\DLD$ do not depend 
on the identity axiom for $\dlcom$, is defined by 
$\theta(D) = \theta_1(\theta_0(D))$, where 
\begin{ldispl}
\begin{geqns}
\theta_0(X) = X\;, \\ 
\theta_0(\slink{s}{a}) = \slink{s}{a}\;, \\
\smash{\theta_0(\pflink{a}{f}) = \pflink{a}{f}\;,} \\
\smash{\theta_0(\flink{a}{f}{b}) = \flink{a}{f}{b}\;,} \\
\theta_0(\valass{a}{n}) = \valass{a}{n}\;, \\
\eqnsep
\theta_0(\True) = \True\;, \\
\theta_0(\False) = \False\;,
\rule[-3ex]{0ex}{3ex}
\end{geqns}
\qquad
\begin{geqns}
\theta_0(\emptydl) = \emptydl\;, \\
\theta_0(\emptydl \dlcom D) = \theta_0(D)\;, \\
\theta_0(D \dlcom \emptydl) = \theta_0(D)\;, \\
\theta_0(D_1 \dlcom D_2) = \theta_0(D_1) \dlcom \theta_0(D_2)
\;\; \mif D_1 \not\equiv \emptydl \Land D_2 \not\equiv \emptydl\;, \\
\theta_0(D_1 \dlori D_2) = \theta_0(D_1) \dlori \theta_0(D_2)\;, \\
\theta_0(\effop{\alpha}(D)) = \effop{\alpha}(\theta_0(D))\;,
\eqnsep
\theta_0(\yldop{\alpha}(D)) = \yldop{\alpha}(\theta_0(D))\;,
\end{geqns}
\eqnsep
\begin{geqns}
\theta_1(X) = X\;, \\
\theta_1(\slink{s}{a}) = \slink{s}{a}\;, \\
\smash{\theta_1(\pflink{a}{f}) = \pflink{a}{f}\;,} \\
\smash{\theta_1(\flink{a}{f}{b}) = \flink{a}{f}{b}\;,} \\
\theta_1(\valass{a}{n}) = \valass{a}{n}\;,
\eqnsep
\theta_1(\True) = \True\;, \\
\theta_1(\False) = \False\;.
\end{geqns}
\qquad
\begin{geqns}
\theta_1(\emptydl) = \emptydl\;, \\
\theta_1(D_1 \dlcom D_2) = \theta_1(D_1) \dlcom \theta_1(D_2)\;, 
\phantom{\;\; \mif D_1 \not\equiv \emptydl \Land D_2 \not\equiv \emptydl}
\\
\theta_1(D_1 \dlori D_2) =
(\emptydl \dlcom \theta_1(D_1)) \dlori \theta_1(D_2)\;, \\
\theta_1(\effop{\alpha}(D)) =
\effop{\alpha}(\emptydl \dlcom \theta_1(D))\;, \\ 
\eqnsep
\theta_1(\yldop{\alpha}(D)) =
\yldop{\alpha}(\emptydl \dlcom \theta_1(D))\;,
\end{geqns}
\end{ldispl}%
By checking all rewrite rules, it is easily established that 
$t \osredaciu s$ only if $\theta(t) \osredacu \theta(s)$.
From this it follows that, for each reduction sequence with respect to
$\osredaciu$, the sequence obtained by replacing each term $t$ in the
reduction sequence by $\theta(t)$ is a reduction sequence with respect
to $\osredacu$.
Now assume that $\cR'_\DLD$ is not strongly normalizing modulo \ACI.
Then there exists an infinite reduction sequence with respect to
$\osredaciu$.
Consequently, there exists an infinite reduction sequence with respect
to $\osredacu$ as well.
In other words, $\cR'_\DLD$ is not strongly normalizing modulo \AC.
However, the contrary was proved above.
Hence, $\cR'_\DLD$ is strongly normalizing modulo \ACI.
\qed
\end{proof}

A consequence of Theorem~3.11, Proposition~3.13, and Proposition~3.14 
from~\cite{BBKW89a} is that a priority rewrite system is well-defined if 
its underlying term rewriting system is strongly normalizing.
Moreover, a term rewriting system is strongly normalizing modulo \ACI\
only if it is strongly normalizing as well.
So we have an important corollary of 
Proposition~\ref{prop-strongly-normalizing}.
\begin{corollary}[Well-definedness]
\label{corollary-well-definedness-dld}
\sloppy
The priority rewrite system for \DLD\ is well-defined.
\end{corollary}
If it would not be well-defined, the priority rewrite system for \DLD\ 
would not determine a one-step reduction relation and the following
theorem would not make sense.
\begin{theorem}[Normal forms]
\label{theorem-normal-form-dld}
The priority rewrite system for \DLD\ has the following properties
concerning normal forms with respect to reduction modulo \ACI:
\begin{enumerate}
\item
each element of $\cE$ has a unique normal form modulo \ACI;
\item
the normal forms of the elements of $\cE$ are exactly the basic terms 
over \DLA;
\item
each element of $\cY$ has a unique normal form;
\item
the normal forms of the elements of $\cY$ are exactly the constants of
sort $\Reply$.
\end{enumerate}
\end{theorem}
\begin{proof}
Properties~1 and~3 are proved combined.
It follows immediately from Proposition~\ref{prop-strongly-normalizing}
that the priority rewrite system for \DLD\ is strongly normalizing 
modulo \ACI\ on all closed \DLD\ terms.
Therefore, it remains to be shown that the priority rewrite system for 
\DLD\ is weakly confluent modulo \ACI\ on all closed \DLD\ terms.

In this weak confluence proof, we use the 
\emph{one-step equality relation} $\osaci$.
This relation is defined as the closure of the set of all closed 
instances of the equations in \ACI\ under symmetry and closed contexts.
A critical pair $\tup{t \osred s,t \osred s'}$ in which $t$, $s$, and 
$s'$ are closed terms is called a \emph{closed critical pair}.
\pagebreak[2]

The following holds for the priority rewrite system for \DLD:
\begin{enumerate}
\item
$r' < r$ if and only if the left hand side of $r$ is a substitution
instance of the left hand side of $r'$;
\item
for all closed critical pairs $\tup{t \osred s,t \osred s'}$ that arise 
from overlap modulo \ACI\ of a rewrite rule on an incomparable rewrite 
rule, $s$ and $s'$ have a common reduct modulo \ACI;
\item
for all closed critical pairs $\tup{t \osred s,t \osaci s'}$ that arise 
from overlap modulo \ACI\ of a rewrite rule on an equation from \ACI, 
there exists a one-step reduction $s' \osredaci s''$ that consists of 
the contraction of a redex modulo \ACI\ such that $s$ and $s''$ have a 
common reduct modulo \ACI;
\item
overlaps between comparable rewrite rules are overlaps at the outermost
positions only.
\end{enumerate}
From this, the weak confluence modulo \ACI\ of the priority rewrite 
system for \DLD\ on all closed \DLD\ terms follows straightforwardly, 
following a line of reasoning similar to the one followed in the proof 
of Theorem~4.8 from~\cite{BBKW89a}, using Theorems~5 and~16 
from~\cite{JK86a}.
The proof of Theorem~4.8 from~\cite{BBKW89a} is concerned with showing 
that the one-step reduction relation determined by a priority rewrite 
system does not give rise to critical pairs while the current proof is 
concerned with showing that the one-step reduction relation determined 
by the priority rewrite system of \DLD\ does not give rise to critical 
pairs that are not convergent.
We use Theorems~5 and~16 from~\cite{JK86a} because we consider reduction 
modulo \ACI.
These theorems allow for confluence modulo a set of equations to be 
established by convergence checks on critical pairs. 
In the use of the second theorem, it suffices to let $L = \emptyset$. 
Because we can add to point~3 above that the redex concerned occurs 
`less deep than the overlap', this theorem can be used in spite of the
fact that \ACI-equivalence classes are infinite (see the second remark 
following the proof of the theorem).

Properties~2 and~4 are easily proved combined by structural induction.
\qed
\end{proof}

In Table~\ref{rewrite-eff-yld-K}, $L$ stands for an arbitrary basic term
over \DLA.
This means that each rewrite rule schema in which $L$ occurs represents
an infinite number of rewrite rules.
We have a corollary of Theorem~\ref{theorem-normal-form-dld} which is
relevant to this point because, modulo \ACI, the number of basic terms 
over \DLA\ is finite.
\begin{corollary}[Equivalent priority rewrite system]
\label{corollary-eqv-PRS}
Let $\sim$ be the equivalence relation on $\cB$ induced by $\ACI$, 
let $\funct{\mathit{repr}}{\cB / {\sim}}{\cB}$ be such that
$\mathit{repr}(\cB') \in \cB'$ for all $\cB' \in \cB / {\sim}$,
and let $\cB^*$ be the image of $\cB / {\sim}$ under $\mathit{repr}$.
Take the priority rewrite system obtained from the priority rewrite 
system for \DLD\ by restricting the basic terms over \DLA\ that $L$ 
stands for to the elements of $\cB^*$.
This adapted priority rewrite system determines the same one-step 
reduction relation modulo \ACI\ as the priority
rewrite system of \DLD.
\end{corollary}
Corollary~\ref{corollary-eqv-PRS} is material to 
Corollary~\ref{corollary-decidability}.

Let a priority rewrite system $\tup{\cR,<}$ be given.
Let $r$ be a rewrite rule from $\cR$, and let $t \osred s$ be an
$r$-rewrite.
Then $r$ is \emph{enabled} for $t$ if $t \osred s$ belongs to the
one-step reduction relation determined by $\tup{\cR,<}$.
\begin{proposition}[Enabled rewrite rules]
\label{prop-enabledness}
Let $r$ be a rewrite rule from the priority rewrite system for \DLD, and
let $D$ be a closed $r$-redex.
Then
\begin{enumerate}
\item
if $D \not\equiv \effop{\alpha}(D')$ and
$D \not\equiv \yldop{\alpha}(D')$ for all $\alpha \in \Act_\DLD$ and
$D' \in \cE$, then $r$ is enabled for $D$;
\item
if $D \equiv \effop{\alpha}(D')$ or $D \equiv \yldop{\alpha}(D')$ for
some $\alpha \in \Act_\DLD$ and $D' \in \cE$, then $r$ is enabled for
$D$ if and only if $D$ is not a closed $r'$-redex for some rewrite rule
$r'$ with $r < r'$.
\pagebreak[2]
\end{enumerate}
\end{proposition}
\begin{proof}
This follows immediately from the definition of the one-step reduction
relation determined by a priority rewrite system and the priority
rewrite system for \DLD.
\qed
\end{proof}

We have an important corollary of Theorem~\ref{theorem-normal-form-dld},
Corollary~\ref{corollary-eqv-PRS}, and 
Proposition~\ref{prop-enabledness}.
\begin{corollary}[Decidability of reduction relation]
\label{corollary-decidability}
The one-step reduction modulo \ACI\ relation determined by the priority
rewrite system for \DLD\ is decidable.
\end{corollary}

\section{Reclaiming Garbage in Data Linkage Dynamics}
\label{sect-garb-coll}

Atomic objects that are not reachable via spots and fields can be
reclaimed.
Reclamation of unreachable atomic objects is relevant because the set
$\AtObj$ of atomic objects is finite.
There are various ways to achieve reclamation of unreachable atomic
objects.
In this section, we introduce some of them.

Data linkage dynamics has the following reclamation-related actions:
\begin{itemize}
\item
a \emph{full garbage collection action} $\fgc$;
\item
a \emph{restricted garbage collection action} $\rgc$;
\item
for each $s \in \Spot$,
a \emph{get fresh atomic object action with safe disposal}
$\sdgetatobj{s}$;
\item
for each $s,t \in \Spot$,
a \emph{set spot action with safe disposal} $\sdsetspot{s}{t}$;
\item
for each $s \in \Spot$,
a \emph{clear spot action with safe disposal} $\sdclrspot{s}$;
\item
for each $s,t \in \Spot, f \in \Field$,
a \emph{set field action with safe disposal} $\sdsetfield{s}{f}{t}$;
\item
for each $s \in \Spot, f \in \Field$,
a \emph{clear field action with safe disposal} $\sdclrfield{s}{f}$;
\item
for each $s,t \in \Spot, f \in \Field$,
a \emph{get field action with safe disposal} $\sdgetfield{s}{t}{f}$;
\item
for each $s \in \Spot$,
a \emph{get fresh atomic object action with unsafe disposal}
$\udgetatobj{s}$;
\item
for each $s,t \in \Spot$,
a \emph{set spot action with unsafe disposal} $\udsetspot{s}{t}$;
\item
for each $s \in \Spot$,
a \emph{clear spot action with unsafe disposal} $\udclrspot{s}$;
\item
for each $s,t \in \Spot, f \in \Field$,
a \emph{set field action with unsafe disposal} $\udsetfield{s}{f}{t}$;
\item
for each $s \in \Spot,\hsp{-.07} f \in \Field$,\hsp{-.07}
a \emph{clear field action with unsafe disposal} $\udclrfield{s}{f}$;
\item
for each $s,t \in \Spot, f \in \Field$,
a \emph{get field action with unsafe disposal} $\udgetfield{s}{t}{f}$.
\end{itemize}

These reclamation-related actions of \DLD\ can be explained as follows
if all spots and fields involved in performing them are spots and fields
with regard to which the current state is locally deterministic:
\begin{itemize}
\item
$\fgc$:
all unreachable atomic objects are reclaimed, and the reply is $\True$;
\item
$\rgc$:
all unreachable atomic objects that do not occur in a cycle are
reclaimed, and the reply is $\True$;
\pagebreak[2]
\item
$\sdgetatobj{s}$, $\sdsetspot{s}{t}$, $\sdclrspot{s}$,
$\sdsetfield{s}{f}{t}$, $\sdclrfield{s}{f}$, and $\sdgetfield{s}{t}{f}$:
like $\getatobj{s}$, $\setspot{s}{t}$, $\clrspot{s}$, 
$\setfield{s}{f}{t}$, $\clrfield{s}{f}$, and $\getfield{s}{t}{f}$, but 
followed by the reclamation of the old content of the spot or field 
whose content has been replaced if it has become an unreachable atomic 
object;
\item
$\udgetatobj{s}$, $\udsetspot{s}{t}$, $\udclrspot{s}$,
$\udsetfield{s}{f}{t}$, $\udclrfield{s}{f}$, and $\udgetfield{s}{t}{f}$:
like $\getatobj{s}$, $\setspot{s}{t}$, $\clrspot{s}$, 
$\setfield{s}{f}{t}$, $\clrfield{s}{f}$, and $\getfield{s}{t}{f}$, but 
followed by the reclamation of the old content of the spot or field 
whose content has been replaced after the content of everything 
containing it has been made undefined.
\end{itemize}
If not all spots and fields involved in performing a reclamation-related 
action are spots and fields with regard to which the current state is 
locally deterministic, there is no state change and the reply is 
$\False$.

Full or restricted garbage collection can be made automatic by treating
each $\getatobj{s}$ as if it is $\fgc$ or $\rgc$ followed by
$\getatobj{s}$.

Garbage collection originates from programming languages that support
the use of dynamic data structures.
It is not only found in contemporary object-oriented programming 
languages such as Java~\cite{AGH05a,GJSB05a} and 
C\#~\cite{HTWG11a,HWG03a}, but also in historic programming languages 
such as LISP~\cite{McC60a,MAEHL62a}.
In~\cite{McC60a}, the term reclamation is used instead of garbage
collection.

The rewrite rules for full garbage collection and restricted garbage
collection are given in Table~\ref{rewrite-gc}.%
\begin{table}[!t]
\caption{Rewrite rules for full and restricted garbage collection}
\label{rewrite-gc}
\begin{eqntbl}
\begin{prucol}
\bprio{1} &
\effop{\fgc}(X) = \fgcaux(\emptydl,X)
\\
\bprio{1} &
\effop{\rgc}(X) = \rgcaux(\emptydl,X)
\\
\bprio{1} &
\yldop{\fgc}(X) = \True
\\
\bprio{1} &
\yldop{\rgc}(X) = \True
\\
\bprio{1} &
\fgcaux(X,\slinkp{s}{a} \dlcom Y) = \fgcaux(X \dlcom \slinkp{s}{a},Y)
\\
\prio{1} &
\fgcaux(X \dlcom \slinkp{s}{a},\pflinkp{a}{f} \dlcom Y) =
\fgcaux(X \dlcom \slinkp{s}{a} \dlcom \pflinkp{a}{f},Y)
\\
\prio{1} &
\fgcaux(X \dlcom \flinkp{a}{f}{b},\pflinkp{b}{g} \dlcom Y) =
\fgcaux(X \dlcom \flinkp{a}{f}{b} \dlcom \pflinkp{b}{g},Y)
\\
\prio{1} &
\fgcaux(X \dlcom \slinkp{s}{a},\flinkp{a}{f}{b} \dlcom Y) =
\fgcaux(X \dlcom \slinkp{s}{a} \dlcom \flinkp{a}{f}{b},Y)
\\
\prio{1} &
\fgcaux(X \dlcom \flinkp{a}{f}{b},\flinkp{b}{g}{c} \dlcom Y) =
\fgcaux(X \dlcom \flinkp{a}{f}{b} \dlcom \flinkp{b}{g}{c},Y)
\\
\prio{1} &
\fgcaux(X \dlcom \slinkp{s}{a},\valass{a}{n} \dlcom Y) =
\fgcaux(X \dlcom \slinkp{s}{a} \dlcom \valass{a}{n},Y)
\\
\prio{1} &
\fgcaux(X \dlcom \flinkp{a}{f}{b},\valass{b}{n} \dlcom Y) =
\fgcaux(X \dlcom \flinkp{a}{f}{b} \dlcom \valass{b}{n},Y)
\\
\prio{2} &
\fgcaux(X,Y) = X
\\
\bprio{1} &
\rgcaux(X,\slinkp{s}{a} \dlcom Y) = \rgcaux(X \dlcom \slinkp{s}{a},Y)
\\
\prio{1} &
\rgcaux(X \dlcom \slinkp{s}{a},\pflinkp{a}{f} \dlcom Y) =
\rgcaux(X \dlcom \slinkp{s}{a} \dlcom \pflinkp{a}{f},Y)
\\
\prio{1} &
\rgcaux(X \dlcom \flinkp{a}{f}{b},\pflinkp{b}{g} \dlcom Y) =
\rgcaux(X \dlcom \flinkp{a}{f}{b} \dlcom \pflinkp{b}{g},Y)
\\
\prio{1} &
\rgcaux(X,\flinkp{a}{f}{b} \dlcom \pflinkp{b}{g} \dlcom Y) =
\rgcaux(X \dlcom \pflinkp{b}{g},\flinkp{a}{f}{b} \dlcom Y)
\\
\prio{1} &
\rgcaux(X \dlcom \slinkp{s}{a},\flinkp{a}{f}{b} \dlcom Y) =
\rgcaux(X \dlcom \slinkp{s}{a} \dlcom \flinkp{a}{f}{b},Y)
\\
\prio{1} &
\rgcaux(X \dlcom \flinkp{a}{f}{b},\flinkp{b}{g}{c} \dlcom Y) =
\rgcaux(X \dlcom \flinkp{a}{f}{b} \dlcom \flinkp{b}{g}{c},Y)
\\
\prio{1} &
\rgcaux(X,\flinkp{a}{f}{b} \dlcom \flinkp{b}{g}{c} \dlcom Y) =
\rgcaux(X \dlcom \flinkp{b}{g}{c},\flinkp{a}{f}{b} \dlcom Y)
\\
\prio{1} &
\rgcaux(X \dlcom \slinkp{s}{a},\valass{a}{n} \dlcom Y) =
\rgcaux(X \dlcom \slinkp{s}{a} \dlcom \valass{a}{n},Y)
\\
\prio{1} &
\rgcaux(X \dlcom \flinkp{a}{f}{b},\valass{b}{n} \dlcom Y) =
\rgcaux(X \dlcom \flinkp{a}{f}{b} \dlcom \valass{b}{n},Y)
\\
\prio{1} &
\rgcaux(X,\flinkp{a}{f}{b} \dlcom \valass{b}{n} \dlcom Y) =
\rgcaux(X \dlcom \valass{b}{n},\flinkp{a}{f}{b} \dlcom Y)
\\
\prio{1} &
\rgcaux(X,\emptydl) = X
\\
\eprio{2} &
\rgcaux(X,Y) = \rgcaux(\emptydl,X)
\end{prucol}
\end{eqntbl}
\end{table}
In this table, $s$ stands for an arbitrary spot from $\Spot$,
$f$ and $g$ stand for arbitrary fields from $\Field$,
$a$, $b$, $c$ and $d$ stand for arbitrary atomic objects from $\AtObj$,
and $n$ stand for an arbitrary value from $\Value$.

The operators $\effop{\fgc}$ and $\effop{\rgc}$ are described using the
auxiliary operators $\fgcaux$ and $\rgcaux$, respectively.
We mention that in $\fgcaux(L,L')$:
\begin{itemize}
\item
$L \dlcom L'$ is the data linkage on which full garbage collection is
carried out;
\item
all atomic objects found in $L$ are already known to be reachable;
\item
if all atomic objects found in $L'$ are unreachable, then $L$ is the
result of full garbage collection on $L \dlcom L'$.
\end{itemize}
We mention that in $\rgcaux(L,L')$:
\begin{itemize}
\item
$L \dlcom L'$ is the data linkage on which removal of all links from and
value associations with atomic objects that have a reference count equal
to zero is carried out repeatedly until this is no longer possible (the
reference count of an atomic object is the number of links to that
atomic object);
\item
all atomic objects found in $L$ are already known to have a reference
count greater than zero;
\item
if all atomic objects found in $L'$ have a reference count equal to
zero, then $L$ is the result of removing all links from and value
associations with atomic objects that have a reference count equal to
zero from $L \dlcom L'$;
\item
atomic objects found in $L$ that have a reference count greater than
zero in $L \dlcom L'$ may have a reference count equal to zero in $L$.
\end{itemize}
\sloppy
It is striking that the description of restricted garbage collection by
means of rewrite rules with priorities is more complicated than the
description of full garbage collection by means of rewrite rules with
priorities.

The rewrite rules for safe disposal and unsafe disposal are given in
Table~\ref{rewrite-disp}.%
\begin{table}[p]
\caption{Rewrite rules for safe and unsafe disposal}
\label{rewrite-disp}
\begin{eqntbl}
\begin{sprucol}
\bprio{1} &
\effop{\sdgetatobj{s}}(X \dlcom \slinkp{s}{a}) =
\sdaux{a}(\emptydl,\effop{\getatobj{s}}(X \dlcom \slinkp{s}{a}))
\\
\prio{2} &
\effop{\sdgetatobj{s}}(X) = \effop{\getatobj{s}}(X)
\\
\bprio{1} &
\effop{\sdsetspot{s}{t}}(X \dlcom \slinkp{s}{a}) =
\sdaux{a}(\emptydl,\effop{\setspot{s}{t}}(X \dlcom \slinkp{s}{a}))
\\
\prio{2} &
\effop{\sdsetspot{s}{t}}(X) = \effop{\setspot{s}{t}}(X)
\\
\bprio{1} &
\effop{\sdclrspot{s}}(X \dlcom \slinkp{s}{a}) =
\sdaux{a}(\emptydl,\effop{\clrspot{s}}(X \dlcom \slinkp{s}{a}))
\\
\prio{2} &
\effop{\sdclrspot{s}}(X) = \effop{\clrspot{s}}(X)
\\
\bprio{1} &
\multicolumn{2}{@{}l@{}}
{\effop{\sdsetfield{s}{f}{t}}
  (X \dlcom \slinkp{s}{a} \dlcom \flinkp{a}{f}{b}) =
 \sdaux{b}
  (\emptydl,
   \effop{\setfield{s}{f}{t}}
    (X \dlcom \slinkp{s}{a} \dlcom \flinkp{a}{f}{b}))}
\\
\prio{2} &
\effop{\sdsetfield{s}{f}{t}}(X) = \effop{\setfield{s}{f}{t}}(X)
\\
\bprio{1} &
\multicolumn{2}{@{}l@{}}
{\effop{\sdclrfield{s}{f}}
  (X \dlcom \slinkp{s}{a} \dlcom \flinkp{a}{f}{b}) =
 \sdaux{b}
  (\emptydl,
   \effop{\clrfield{s}{f}}
    (X \dlcom \slinkp{s}{a} \dlcom \flinkp{a}{f}{b}))}
\\
\prio{2} &
\effop{\sdclrfield{s}{f}}(X) = \effop{\clrfield{s}{f}}(X)
\\
\bprio{1} &
\effop{\sdgetfield{s}{t}{f}}(X \dlcom \slinkp{s}{a}) =
\sdaux{a}(\emptydl,\effop{\getfield{s}{t}{f}}(X \dlcom \slinkp{s}{a}))
\\
\prio{2} &
\effop{\sdgetfield{s}{t}{f}}(X) = \effop{\getfield{s}{t}{f}}(X)
\\
\bprio{1} &
\effop{\udgetatobj{s}}(X \dlcom \slinkp{s}{a} \dlcom \slinkp{s}{b}) =
X \dlcom \slinkp{s}{a} \dlcom \slinkp{s}{b}
 & \mif a \neq b
\\
\prio{2} &
\effop{\udgetatobj{s}}(X \dlcom \slinkp{s}{a}) =
\sdaux{a}
 (\emptydl,\clraux{a}(\effop{\getatobj{s}}(X \dlcom \slinkp{s}{a})))
\\
\prio{3} &
\effop{\udgetatobj{s}}(X) = \effop{\getatobj{s}}(X)
\\
\bprio{1} &
\effop{\udsetspot{s}{t}}(X \dlcom \slinkp{s}{a} \dlcom \slinkp{s}{b}) =
X \dlcom \slinkp{s}{a} \dlcom \slinkp{s}{b}
 & \mif a \neq b
\\
\prio{1} &
\effop{\udsetspot{s}{t}}(X \dlcom \slinkp{t}{a} \dlcom \slinkp{t}{b}) =
X \dlcom \slinkp{t}{a} \dlcom \slinkp{t}{b}
 & \mif a \neq b
\\
\prio{2} &
\effop{\udsetspot{s}{t}}(X \dlcom \slinkp{s}{a}) =
\sdaux{a}
 (\emptydl,\clraux{a}(\effop{\setspot{s}{t}}(X \dlcom \slinkp{s}{a})))
\\
\prio{3} &
\effop{\udsetspot{s}{t}}(X) = \effop{\setspot{s}{t}}(X)
\\
\bprio{1} &
\effop{\udclrspot{s}}(X \dlcom \slinkp{s}{a} \dlcom \slinkp{s}{b}) =
X \dlcom \slinkp{s}{a} \dlcom \slinkp{s}{b}
 & \mif a \neq b
\\
\prio{2} &
\effop{\udclrspot{s}}(X \dlcom \slinkp{s}{a}) =
\sdaux{a}
 (\emptydl,\clraux{a}(\effop{\clrspot{s}}(X \dlcom \slinkp{s}{a})))
\qquad\qquad
\\
\prio{3} &
\effop{\udclrspot{s}}(X) = \effop{\clrspot{s}}(X)
\\
\bprio{1} &
\effop{\udsetfield{s}{f}{t}}
 (X \dlcom \slinkp{s}{a} \dlcom \slinkp{s}{b}) =
X \dlcom \slinkp{s}{a} \dlcom \slinkp{s}{b}
 & \mif a \neq b
\\
\prio{1} &
\effop{\udsetfield{s}{f}{t}}
 (X \dlcom \slinkp{s}{a} \dlcom \flinkp{a}{f}{b} \dlcom \flinkp{a}{f}{c}) =
X \dlcom \slinkp{s}{a} \dlcom \flinkp{a}{f}{b} \dlcom \flinkp{a}{f}{c}
 & \mif b \neq c
\\
\prio{1} &
\effop{\udsetfield{s}{f}{t}}
 (X \dlcom \slinkp{s}{a} \dlcom \flinkp{a}{f}{b} \dlcom \pflinkp{a}{f}) =
X \dlcom \slinkp{s}{a} \dlcom \flinkp{a}{f}{b} \dlcom \pflinkp{a}{f}
\\
\prio{1} &
\effop{\udsetfield{s}{f}{t}}
 (X \dlcom \slinkp{t}{a} \dlcom \slinkp{t}{b}) =
X \dlcom \slinkp{t}{a} \dlcom \slinkp{t}{b}
 & \mif a \neq b
\\
\prio{2} &
\effop{\udsetfield{s}{f}{t}}
 (X \dlcom \slinkp{s}{a} \dlcom \flinkp{a}{f}{b}) = 
\sdaux{b}
 (\emptydl,
  \clraux{b}
   (\effop{\setfield{s}{f}{t}}
     (X \dlcom \slinkp{s}{a} \dlcom \flinkp{a}{f}{b}))) \hsp{-3}
\\
\prio{3} &
\effop{\udsetfield{s}{f}{t}}(X) = \effop{\setfield{s}{f}{t}}(X)
\\
\bprio{1} &
\effop{\udclrfield{s}{f}}
 (X \dlcom \slinkp{s}{a} \dlcom \slinkp{s}{b}) =
X \dlcom \slinkp{s}{a} \dlcom \slinkp{s}{b}
 & \mif a \neq b
\\
\prio{1} &
\effop{\udclrfield{s}{f}}
 (X \dlcom \slinkp{s}{a} \dlcom \flinkp{a}{f}{b} \dlcom \flinkp{a}{f}{c}) =
X \dlcom \slinkp{s}{a} \dlcom \flinkp{a}{f}{b} \dlcom \flinkp{a}{f}{c}
 & \mif b \neq c
\\
\prio{1} &
\effop{\udclrfield{s}{f}}
 (X \dlcom \slinkp{s}{a} \dlcom \flinkp{a}{f}{b} \dlcom \pflinkp{a}{f}) =
X \dlcom \slinkp{s}{a} \dlcom \flinkp{a}{f}{b} \dlcom \pflinkp{a}{f}
\\
\prio{2} &
\effop{\udclrfield{s}{f}}
 (X \dlcom \slinkp{s}{a} \dlcom \flinkp{a}{f}{b}) = 
\sdaux{b}
 (\emptydl,
  \clraux{b}
   (\effop{\clrfield{s}{f}}
     (X \dlcom \slinkp{s}{a} \dlcom \flinkp{a}{f}{b}))) \hsp{-3}
\\
\prio{3} &
\effop{\udclrfield{s}{f}}(X) = \effop{\clrfield{s}{f}}(X)
\\
\bprio{1} &
\effop{\udgetfield{s}{t}{f}}
 (X \dlcom \slinkp{s}{a} \dlcom \slinkp{s}{b}) =
X \dlcom \slinkp{s}{a} \dlcom \slinkp{s}{b}
 & \mif a \neq b
\\
\prio{1} &
\effop{\udgetfield{s}{t}{f}}
 (X \dlcom \slinkp{t}{a} \dlcom \slinkp{t}{b}) =
X \dlcom \slinkp{t}{a} \dlcom \slinkp{t}{b}
 & \mif a \neq b
\\
\prio{1} &
\effop{\udgetfield{s}{t}{f}}
 (X \dlcom \slinkp{t}{a} \dlcom \flinkp{a}{f}{b} \dlcom \flinkp{a}{f}{c}) =
X \dlcom \slinkp{t}{a} \dlcom \flinkp{a}{f}{b} \dlcom \flinkp{a}{f}{c}
 & \mif b \neq c
\\
\prio{1} &
\effop{\udgetfield{s}{t}{f}}
 (X \dlcom \slinkp{t}{a} \dlcom \flinkp{a}{f}{b} \dlcom \pflinkp{a}{f}) =
X \dlcom \slinkp{t}{a} \dlcom \flinkp{a}{f}{b} \dlcom \pflinkp{a}{f}
\\
\prio{2} &
\effop{\udgetfield{s}{t}{f}}(X \dlcom \slinkp{s}{a}) =
\sdaux{a}
 (\emptydl,
  \clraux{a}(\effop{\getfield{s}{t}{f}}(X \dlcom \slinkp{s}{a})))
\\
\eprio{3} &
\effop{\udgetfield{s}{t}{f}}(X) = \effop{\getfield{s}{t}{f}}(X)
\end{sprucol}
\end{eqntbl}
\end{table}
\addtocounter{table}{-1}%
\begin{table}[p]
\caption{(Continued)}
\begin{eqntbl}
\begin{sprucol}
\bprio{1} &
\yldop{\sdgetatobj{s}}(X) = \yldop{\getatobj{s}}(X)
\\
\bprio{1} &
\yldop{\sdsetspot{s}{t}}(X) = \yldop{\setspot{s}{t}}(X)
\\
\bprio{1} &
\yldop{\sdclrspot{s}}(X) = \yldop{\clrspot{s}}(X)
\\
\bprio{1} &
\yldop{\sdsetfield{s}{f}{t}}(X) = \yldop{\setfield{s}{f}{t}}(X)
\\
\bprio{1} &
\yldop{\sdclrfield{s}{f}}(X) = \yldop{\clrfield{s}{f}}(X)
\\
\bprio{1} &
\yldop{\sdgetfield{s}{t}{f}}(X) = \yldop{\getfield{s}{t}{f}}(X)
\\
\bprio{1} &
\yldop{\udgetatobj{s}}(X) = \yldop{\getatobj{s}}(X)
\\
\bprio{1} &
\yldop{\udsetspot{s}{t}}(X) = \yldop{\setspot{s}{t}}(X)
\\
\bprio{1} &
\yldop{\udclrspot{s}}(X) = \yldop{\clrspot{s}}(X)
\\
\bprio{1} &
\yldop{\udsetfield{s}{f}{t}}(X) = \yldop{\setfield{s}{f}{t}}(X)
\\
\bprio{1} &
\yldop{\udclrfield{s}{f}}(X) = \yldop{\clrfield{s}{f}}(X)
\\
\bprio{1} &
\yldop{\udgetfield{s}{t}{f}}(X) = \yldop{\getfield{s}{t}{f}}(X)
\\
\bprio{1} &
\sdaux{d}(X,\slinkp{s}{a} \dlcom Y) =
\sdaux{d}(X \dlcom \slinkp{s}{a},Y)
\\
\prio{1} &
\sdaux{d}(X \dlcom \slinkp{s}{a},\pflinkp{a}{f} \dlcom Y) =
\sdaux{d}(X \dlcom \slinkp{s}{a} \dlcom \pflinkp{a}{f},Y)
\\
\prio{1} &
\sdaux{d}(X \dlcom \flinkp{a}{f}{b},\pflinkp{b}{g} \dlcom Y) =
\sdaux{d}(X \dlcom \flinkp{a}{f}{b} \dlcom \pflinkp{b}{g},Y)
\\
\prio{1} &
\sdaux{d}(X \dlcom \slinkp{s}{a},\flinkp{a}{f}{b} \dlcom Y) =
\sdaux{d}(X \dlcom \slinkp{s}{a} \dlcom \flinkp{a}{f}{b},Y)
\\
\prio{1} &
\sdaux{d}(X \dlcom \flinkp{a}{f}{b},\flinkp{b}{g}{c} \dlcom Y) =
\sdaux{d}(X \dlcom \flinkp{a}{f}{b} \dlcom \flinkp{b}{g}{c},Y)
\\
\prio{1} &
\sdaux{d}(X \dlcom \slinkp{s}{a},\valass{a}{n} \dlcom Y) =
\sdaux{d}(X \dlcom \slinkp{s}{a} \dlcom \valass{a}{n},Y)
\\
\prio{1} &
\sdaux{d}(X \dlcom \flinkp{a}{f}{b},\valass{b}{n} \dlcom Y) =
\sdaux{d}(X \dlcom \flinkp{a}{f}{b} \dlcom \valass{b}{n},Y)
\\
\prio{2} &
\sdaux{d}(X \dlcom \slinkp{s}{d},\flinkp{a}{f}{d} \dlcom Y) =
\sdaux{d}(X \dlcom \slinkp{s}{d} \dlcom \flinkp{a}{f}{d},Y)
\\
\prio{2} &
\sdaux{d}(X \dlcom \flinkp{a}{f}{d},\flinkp{b}{f}{d} \dlcom Y) =
\sdaux{d}(X \dlcom \flinkp{a}{f}{d} \dlcom \flinkp{b}{f}{d},Y)
\hsp{.25}
\\
\prio{2} &
\sdaux{d}(X \dlcom \slinkp{s}{d},\valass{d}{n} \dlcom Y) =
\sdaux{d}(X \dlcom \slinkp{s}{d} \dlcom \valass{d}{n},Y)
\\
\prio{2} &
\sdaux{d}(X \dlcom \flinkp{a}{f}{d},\valass{d}{n} \dlcom Y) =
\sdaux{d}(X \dlcom \flinkp{a}{f}{d} \dlcom \valass{d}{n},Y)
\\
\prio{3} &
\sdaux{d}(X,\pflinkp{a}{f} \dlcom Y) =
\sdaux{d}(X \dlcom \pflinkp{a}{f},Y) 
 & \mif a \neq d \phantom{{} \Land b \neq d}
\\
\prio{3} &
\sdaux{d}(X,\flinkp{a}{f}{b} \dlcom Y) =
\sdaux{d}(X \dlcom \flinkp{a}{f}{b},Y)
 & \mif a \neq d \Land b \neq d
\\
\prio{3} &
\sdaux{d}(X,\valass{a}{n} \dlcom Y) =
\sdaux{d}(X \dlcom \valass{a}{n},Y)
 & \mif a \neq d \phantom{{} \Land b \neq d}
\\
\prio{4} &
\sdaux{d}(X,Y) = X
\\
\bprio{1} &
\clraux{d}(X \dlcom \slinkp{s}{d}) = \clraux{d}(X)
\\
\prio{1} &
\clraux{d}(X \dlcom \slinkp{s}{a}) =
\clraux{d}(X) \dlcom \slinkp{s}{a}
      & \mif a \neq d \phantom{{} \Land b \neq d}
\\
\prio{1} &
\clraux{d}(X \dlcom \pflinkp{a}{f}) =
\clraux{d}(X) \dlcom \pflinkp{a}{f}
\\
\prio{1} &
\clraux{d}(X \dlcom \flinkp{a}{f}{d}) =
\clraux{d}(X) \dlcom \pflinkp{a}{f}
\\
\prio{1} &
\clraux{d}(X \dlcom \flinkp{a}{f}{b}) =
\clraux{d}(X) \dlcom \flinkp{a}{f}{b}
 & \mif b \neq d \phantom{{} \Land b \neq d}
\\
\eprio{1} &
\clraux{d}(X \dlcom \valass{a}{n}) =
\clraux{d}(X) \dlcom \valass{a}{n}
\end{sprucol}
\end{eqntbl}
\end{table}%
In this table, $s$ and $t$ stand for arbitrary spots from $\Spot$,
$f$ and $g$ stand for arbitrary fields from $\Field$,
$a$, $b$, $c$ and $d$ stand for arbitrary atomic objects from $\AtObj$,
and $n$ stand for an arbitrary value from $\Value$.
\pagebreak[2]

The operators $\effop{\sdgetatobj{s}}$, $\effop{\sdsetspot{s}{t}}$,
$\effop{\sdclrspot{s}}$, $\effop{\sdsetfield{s}{f}{t}}$,
$\effop{\sdclrfield{s}{f}}$, $\effop{\sdgetfield{s}{t}{f}}$,
$\effop{\udgetatobj{s}}$, $\effop{\udsetspot{s}{t}}$,
$\effop{\udclrspot{s}}$, $\effop{\udsetfield{s}{f}{t}}$,
$\effop{\udclrfield{s}{f}}$ and $\effop{\udgetfield{s}{t}{f}}$ are
described using auxiliary operators $\sdaux{a}$ and $\clraux{a}$
($a \in \AtObj$).
Indeed, $\sdaux{a}$ deals with safe disposal of atomic object $a$ and 
$\clraux{a}$ deals with making the content of everything containing $a$ 
undefined.
Carrying out safe disposal of $a$ on $\clraux{a}(L)$ amounts to the same
thing as carrying out unsafe disposal of $a$ on $L$.
We mention that in $\sdaux{a}(L,L')$:
\begin{itemize}
\item
$L \dlcom L'$ is the data linkage on which safe disposal of $a$ is
carried out;
\item
all atomic objects found in $L$ are already known not to be involved in
the safe disposal of $a$;
\item
links and value associations are moved from $L'$ to $L$ in stages as
follows:
\begin{itemize}
\item
first, all links and value associations that make up the reachable part
of $L \dlcom L'$ are moved from $L'$ to $L$,
\item
next, all links to $a$ and value associations with $a$ are moved from
$L'$ to $L$ if $a$ is found in $L$,
\item
finally, all links and value associations in which $a$ is not involved
are moved from $L'$ to $L$;
\end{itemize}
\item
if all atomic objects found in $L'$ are involved in the safe disposal of
$a$, then $L$ is the result of safe disposal of $a$ on $L \dlcom L'$.
\end{itemize}
We mention further that $\clraux{a}(L)$ removes spot links to $a$ from $L$ and
and replaces field links to $a$ by partial field links.

\DLD-R is \DLD\ extended with the reclamation features introduced above.
The priority rewrite system of \DLD-R consists of the rewrite rules from
the priority rewrite system for \DLD\ and the rewrite rules given in 
Tables~\ref{rewrite-gc} and~\ref{rewrite-disp}.
Each of the rewrite rules of \DLD\ is incomparable with each of the
additional rewrite rules.
Moreover, additional rewrite rules in different tables are incomparable.

For \DLD-R, the set of effect terms and the set of yield terms can be
defined like for \DLD.
Theorem~\ref{theorem-normal-form-dld} and
Proposition~\ref{prop-enabledness} go through for \DLD-R.
Moreover, the enabledness of rewrite rules for the auxiliary operators
can be characterized like for the effect and yield operators.
This means that the one-step reduction relation modulo \ACI\ determined
by the priority rewrite system for \DLD-R is decidable as well.

\begin{example}
\label{example-garb-coll}
We consider a data linkage in which $\ul{2}$ and $\ul{3}$ occur as
unreachable atomic objects:
\begin{ldispl}
\slinkp{r}{\ul{0}} \dlcom
\flinkp{\ul{0}}{\nm{up}}{\ul{1}} \dlcom
\flinkp{\ul{2}}{\nm{up}}{\ul{3}} \dlcom
\flinkp{\ul{3}}{\nm{dn}}{\ul{2}}\;.
\end{ldispl}%
We use the rewrite rules for full garbage collection to obtain the
following picture of its application to this data linkage:
\pagebreak[2]
\begin{ldispl}
\effop{\fgc}
 (\slinkp{r}{\ul{0}} \dlcom
  \flinkp{\ul{0}}{\nm{up}}{\ul{1}} \dlcom
  \flinkp{\ul{2}}{\nm{up}}{\ul{3}} \dlcom
  \flinkp{\ul{3}}{\nm{dn}}{\ul{2}})
\osredaci {}
\\
\fgcaux
 (\emptydl,
  \slinkp{r}{\ul{0}} \dlcom
  \flinkp{\ul{0}}{\nm{up}}{\ul{1}} \dlcom
  \flinkp{\ul{2}}{\nm{up}}{\ul{3}} \dlcom
  \flinkp{\ul{3}}{\nm{dn}}{\ul{2}})
\osredaci {}
\\
\fgcaux
 (\slinkp{r}{\ul{0}},
  \flinkp{\ul{0}}{\nm{up}}{\ul{1}} \dlcom
  \flinkp{\ul{2}}{\nm{up}}{\ul{3}} \dlcom
  \flinkp{\ul{3}}{\nm{dn}}{\ul{2}})
\osredaci {}
\\
\fgcaux
 (\slinkp{r}{\ul{0}} \dlcom
  \flinkp{\ul{0}}{\nm{up}}{\ul{1}},
  \flinkp{\ul{2}}{\nm{up}}{\ul{3}} \dlcom
  \flinkp{\ul{3}}{\nm{dn}}{\ul{2}})
\osredaci {}
\\
\slinkp{r}{\ul{0}} \dlcom
\flinkp{\ul{0}}{\nm{up}}{\ul{1}}\;.
\end{ldispl}%
We use the rewrite rules for restricted garbage collection to obtain the
following picture of its application to the same data linkage:
\begin{ldispl}
\effop{\rgc}
 (\slinkp{r}{\ul{0}} \dlcom
  \flinkp{\ul{0}}{\nm{up}}{\ul{1}} \dlcom
  \flinkp{\ul{2}}{\nm{up}}{\ul{3}} \dlcom
  \flinkp{\ul{3}}{\nm{dn}}{\ul{2}})
\osredaci {}
\\
\rgcaux
 (\emptydl,
  \slinkp{r}{\ul{0}} \dlcom
  \flinkp{\ul{0}}{\nm{up}}{\ul{1}} \dlcom
  \flinkp{\ul{2}}{\nm{up}}{\ul{3}} \dlcom
  \flinkp{\ul{3}}{\nm{dn}}{\ul{2}})
\osredaci {}
\\
\rgcaux
 (\slinkp{r}{\ul{0}},
  \flinkp{\ul{0}}{\nm{up}}{\ul{1}} \dlcom
  \flinkp{\ul{2}}{\nm{up}}{\ul{3}} \dlcom
  \flinkp{\ul{3}}{\nm{dn}}{\ul{2}})
\osredaci {}
\\
\rgcaux
 (\slinkp{r}{\ul{0}} \dlcom
  \flinkp{\ul{0}}{\nm{up}}{\ul{1}},
  \flinkp{\ul{2}}{\nm{up}}{\ul{3}} \dlcom
  \flinkp{\ul{3}}{\nm{dn}}{\ul{2}})
\osredaci {}
\\
\rgcaux
 (\slinkp{r}{\ul{0}} \dlcom
  \flinkp{\ul{0}}{\nm{up}}{\ul{1}} \dlcom
  \flinkp{\ul{3}}{\nm{dn}}{\ul{2}},
  \flinkp{\ul{2}}{\nm{up}}{\ul{3}})
\osredaci {}
\\
\rgcaux
 (\slinkp{r}{\ul{0}} \dlcom
  \flinkp{\ul{0}}{\nm{up}}{\ul{1}} \dlcom
  \flinkp{\ul{3}}{\nm{dn}}{\ul{2}} \dlcom
  \flinkp{\ul{2}}{\nm{up}}{\ul{3}},
  \emptydl)
\osredaci {}
\\
\slinkp{r}{\ul{0}} \dlcom
\flinkp{\ul{0}}{\nm{up}}{\ul{1}} \dlcom
\flinkp{\ul{3}}{\nm{dn}}{\ul{2}} \dlcom
\flinkp{\ul{2}}{\nm{up}}{\ul{3}}\;.
\end{ldispl}%
The effect of restricted garbage collection is different because it does
not reclaim atomic objects that occur in a cycle.
\end{example}
\begin{example}
\label{example-disposal}
We consider a data linkage in which atomic object $\ul{0}$ is reachable
in different ways:
\begin{ldispl}
\slinkp{r}{\ul{0}} \dlcom
\slinkp{s}{\ul{0}} \dlcom
\flinkp{\ul{0}}{\nm{up}}{\ul{1}} \dlcom
\slinkp{t}{\ul{2}} \dlcom
\flinkp{\ul{2}}{\nm{up}}{\ul{3}}\;.
\end{ldispl}%
We use the rewrite rules for set spot with safe disposal to obtain the
following picture of its application to this data linkage:
\begin{ldispl}
\effop{\sdsetspot{s}{t}}
 (\slinkp{r}{\ul{0}} \dlcom
  \slinkp{s}{\ul{0}} \dlcom
  \flinkp{\ul{0}}{\nm{up}}{\ul{1}} \dlcom
  \slinkp{t}{\ul{2}} \dlcom
  \flinkp{\ul{2}}{\nm{up}}{\ul{3}})
\osredaci {}
\\
\sdaux{\ul{0}}
 (\emptydl,
  \effop{\setspot{s}{t}}
   (\slinkp{r}{\ul{0}} \dlcom
    \slinkp{s}{\ul{0}} \dlcom
    \flinkp{\ul{0}}{\nm{up}}{\ul{1}} \dlcom
    \slinkp{t}{\ul{2}} \dlcom
    \flinkp{\ul{2}}{\nm{up}}{\ul{3}}))
\osredaci {}
\\
\sdaux{\ul{0}}
 (\emptydl,
  \slinkp{r}{\ul{0}} \dlcom
  \slinkp{s}{\ul{2}} \dlcom
  \flinkp{\ul{0}}{\nm{up}}{\ul{1}} \dlcom
  \slinkp{t}{\ul{2}} \dlcom
  \flinkp{\ul{2}}{\nm{up}}{\ul{3}})
\osredaci {}
\\
\sdaux{\ul{0}}
 (\slinkp{r}{\ul{0}},
  \slinkp{s}{\ul{2}} \dlcom
  \flinkp{\ul{0}}{\nm{up}}{\ul{1}} \dlcom
  \slinkp{t}{\ul{2}} \dlcom
  \flinkp{\ul{2}}{\nm{up}}{\ul{3}})
\osredaci {}
\\
\sdaux{\ul{0}}
 (\slinkp{r}{\ul{0}} \dlcom
  \slinkp{s}{\ul{2}},
  \flinkp{\ul{0}}{\nm{up}}{\ul{1}} \dlcom
  \slinkp{t}{\ul{2}} \dlcom
  \flinkp{\ul{2}}{\nm{up}}{\ul{3}})
\osredaci {}
\\
\sdaux{\ul{0}}
 (\slinkp{r}{\ul{0}} \dlcom
  \slinkp{s}{\ul{2}} \dlcom
  \flinkp{\ul{0}}{\nm{up}}{\ul{1}},
  \slinkp{t}{\ul{2}} \dlcom
  \flinkp{\ul{2}}{\nm{up}}{\ul{3}})
\osredaci {}
\\
\sdaux{\ul{0}}
 (\slinkp{r}{\ul{0}} \dlcom
  \slinkp{s}{\ul{2}} \dlcom
  \flinkp{\ul{0}}{\nm{up}}{\ul{1}} \dlcom
  \slinkp{t}{\ul{2}},
  \flinkp{\ul{2}}{\nm{up}}{\ul{3}})
\osredaci {}
\\
\sdaux{\ul{0}}
 (\slinkp{r}{\ul{0}} \dlcom
  \slinkp{s}{\ul{2}} \dlcom
  \flinkp{\ul{0}}{\nm{up}}{\ul{1}} \dlcom
  \slinkp{t}{\ul{2}} \dlcom
  \flinkp{\ul{2}}{\nm{up}}{\ul{3}},
  \emptydl)
\osredaci {}
\\
\slinkp{r}{\ul{0}} \dlcom
\slinkp{s}{\ul{2}} \dlcom
\flinkp{\ul{0}}{\nm{up}}{\ul{1}} \dlcom
\slinkp{t}{\ul{2}} \dlcom
\flinkp{\ul{2}}{\nm{up}}{\ul{3}}\;.
\end{ldispl}%
We use the rewrite rules for set spot with unsafe disposal to obtain the
following picture of its application to the same data linkage:
\pagebreak[2]
\begin{ldispl}
\effop{\udsetspot{s}{t}}
 (\slinkp{r}{\ul{0}} \dlcom
  \slinkp{s}{\ul{0}} \dlcom
  \flinkp{\ul{0}}{\nm{up}}{\ul{1}} \dlcom
  \slinkp{t}{\ul{2}} \dlcom
  \flinkp{\ul{2}}{\nm{up}}{\ul{3}})
\osredaci {}
\\
\sdaux{\ul{0}}
 (\emptydl,
  \clraux{\ul{0}}
   (\effop{\setspot{s}{t}}
     (\slinkp{r}{\ul{0}} \dlcom
      \slinkp{s}{\ul{0}} \dlcom
      \flinkp{\ul{0}}{\nm{up}}{\ul{1}} \dlcom
      \slinkp{t}{\ul{2}} \dlcom
      \flinkp{\ul{2}}{\nm{up}}{\ul{3}})))
\osredaci {}
\\
\sdaux{\ul{0}}
 (\emptydl,
  \clraux{\ul{0}}
   (\slinkp{r}{\ul{0}} \dlcom
    \slinkp{s}{\ul{2}} \dlcom
    \flinkp{\ul{0}}{\nm{up}}{\ul{1}} \dlcom
    \slinkp{t}{\ul{2}} \dlcom
    \flinkp{\ul{2}}{\nm{up}}{\ul{3}}))
\osredaci {}
\\
\sdaux{\ul{0}}
 (\emptydl,
  \clraux{\ul{0}}
   (\slinkp{s}{\ul{2}} \dlcom
    \flinkp{\ul{0}}{\nm{up}}{\ul{1}} \dlcom
    \slinkp{t}{\ul{2}} \dlcom
    \flinkp{\ul{2}}{\nm{up}}{\ul{3}}))
\osredaci {}
\\
\sdaux{\ul{0}}
 (\emptydl,
  \clraux{\ul{0}}
   (\slinkp{s}{\ul{2}} \dlcom
    \flinkp{\ul{0}}{\nm{up}}{\ul{1}} \dlcom
    \slinkp{t}{\ul{2}}) \dlcom
  \flinkp{\ul{2}}{\nm{up}}{\ul{3}})
\osredaci {}
\\
\sdaux{\ul{0}}
 (\emptydl,
  \clraux{\ul{0}}
   (\slinkp{s}{\ul{2}} \dlcom
    \flinkp{\ul{0}}{\nm{up}}{\ul{1}}) \dlcom
  \slinkp{t}{\ul{2}} \dlcom
  \flinkp{\ul{2}}{\nm{up}}{\ul{3}})
\osredaci {}
\\
\sdaux{\ul{0}}
 (\emptydl,
  \clraux{\ul{0}}(\slink{s}{\ul{2}}) \dlcom
  \flinkp{\ul{0}}{\nm{up}}{\ul{1}} \dlcom
  \slinkp{t}{\ul{2}} \dlcom
  \flinkp{\ul{2}}{\nm{up}}{\ul{3}})
\osredaci {}
\\
\sdaux{\ul{0}}
 (\emptydl,
  \slinkp{s}{\ul{2}} \dlcom
  \flinkp{\ul{0}}{\nm{up}}{\ul{1}} \dlcom
  \slinkp{t}{\ul{2}} \dlcom
  \flinkp{\ul{2}}{\nm{up}}{\ul{3}})
\osredaci {}
\\
\sdaux{\ul{0}}
 (\slinkp{s}{\ul{2}},
  \flinkp{\ul{0}}{\nm{up}}{\ul{1}} \dlcom
  \slinkp{t}{\ul{2}} \dlcom
  \flinkp{\ul{2}}{\nm{up}}{\ul{3}})
\osredaci {}
\\
\sdaux{\ul{0}}
 (\slinkp{s}{\ul{2}} \dlcom
  \slinkp{t}{\ul{2}},
  \flinkp{\ul{0}}{\nm{up}}{\ul{1}} \dlcom
  \flinkp{\ul{2}}{\nm{up}}{\ul{3}})
\osredaci {}
\\
\sdaux{\ul{0}}
 (\slinkp{s}{\ul{2}} \dlcom
  \slinkp{t}{\ul{2}} \dlcom
  \flinkp{\ul{2}}{\nm{up}}{\ul{3}},
  \flinkp{\ul{0}}{\nm{up}}{\ul{1}})
\osredaci {}
\\
\slinkp{s}{\ul{2}} \dlcom
\slinkp{t}{\ul{2}} \dlcom
\flinkp{\ul{2}}{\nm{up}}{\ul{3}}\;.
\end{ldispl}%
The effect of set spot with unsafe disposal is different because it
reclaims an atomic object irrespective of its reachability.
\end{example}

\section{Basic Thread Algebra}
\label{sect-BTA}

In this section, we review \BTA\ (Basic Thread Algebra), a form of
process algebra which is tailored to the behaviours that are produced by
deterministic sequential programs under execution.
The behaviours concerned are called \emph{threads}.

In \BTA, it is assumed that a fixed but arbitrary finite set $\BAct$ of
\emph{basic actions}, with $\Tau \not\in \BAct$, has been given.
Besides, $\Tau$ is a special basic action.
We write $\BActTau$ for $\BAct \union \set{\Tau}$.

A thread is a behaviour which consists of performing basic actions in a
sequential fashion.
Upon each basic action performed, a reply from an execution environment
determines how the thread proceeds.
The possible replies are the Boolean values $\True$ and $\False$.
Performing $\Tau$, which is considered performing an internal action, 
will always lead to the reply $\True$.

\BTA\ has one sort: the sort $\Thr$ of \emph{threads}.
To build terms of sort $\Thr$, \BTA\ has the following constants and
operators:
\begin{itemize}
\item
the \emph{inaction} constant $\const{\DeadEnd}{\Thr}$;
\item
the \emph{termination} constant $\const{\Stop}{\Thr}$;
\item
for each $\alpha \in \BActTau$,
the binary \emph{postconditional composition} operator
$\funct{\pcc{\ph}{\alpha}{\ph}}{\Thr \x \Thr}{\Thr}$.
\end{itemize}
Terms of sort $\Thr$ are built as usual.
Throughout the paper, we assume that there are infinitely many variables
of sort $\Thr$, including $x,y,z$.

We use infix notation for postconditional composition.
We introduce \emph{basic action prefixing} as an abbreviation:
$\alpha \bapf p$, where $p$ is a term of sort $\Thr$, abbreviates
$\pcc{p}{\alpha}{p}$.
We identify expressions of the form $\alpha \bapf p$ with the \BTA\ term 
they stand for.

The thread denoted by a closed term of the form $\pcc{p}{\alpha}{q}$ 
will first perform $\alpha$, and then proceed as the thread denoted by 
$p$ if the reply from the execution environment is $\True$ and proceed 
as the thread denoted by $q$ if the reply from the execution environment 
is $\False$.
The thread denoted by $\DeadEnd$ will become inactive and the thread
denoted by $\Stop$ will terminate.

\begin{example}
\label{example-BTA-terms}
Some simple examples of closed \BTA\ terms are
\begin{ldispl}
a \bapf (\pcc{\Stop}{b}{\DeadEnd})\;, \qquad
\pcc{(b \bapf \Stop)}{a}{\DeadEnd}\;.
\end{ldispl}%
The first term denotes the thread that first performs basic action $a$, 
next performs basic action $b$, if the reply from the execution 
environment on performing $b$ is $\True$, after that terminates, and if 
the reply from the execution environment on performing $b$ is $\False$, 
after that becomes inactive.
The second term denotes the thread that first performs basic action $a$, 
if the reply from the execution environment on performing $a$ is 
$\True$, next performs the basic action $b$ and after that terminates, 
and if the reply from the execution environment on performing $a$ is 
$\False$, next becomes inactive.
\end{example}

\BTA\ has only one axiom.
This axiom is given in Table~\ref{axioms-BTA}.%
\begin{table}[!t]
\caption{Axiom of \BTA}
\label{axioms-BTA}
\begin{eqntbl}
\begin{axcol}
\pcc{x}{\Tau}{y} = \pcc{x}{\Tau}{x}                      & \axiom{T1}
\end{axcol}
\end{eqntbl}
\end{table}
Using the abbreviation introduced above, axiom T1 can be written as
follows: $\pcc{x}{\Tau}{y} = \Tau \bapf x$.

Each closed \BTA\ term denotes a finite thread, i.e.\ a thread with a
finite upper bound to the number of basic actions that it can perform.
Infinite threads, i.e.\ threads without a finite upper bound to the
number of basic actions that it can perform, can be described by guarded
recursion.

A \emph{guarded recursive specification} over \BTA\ is a set of
recursion equations $\set{X = p_X \where X \in V}$, where $V$ is a
set of variables of sort $\Thr$ and each $p_X$ is a term of the form
$\DeadEnd$, $\Stop$ or $\pcc{p}{\alpha}{q}$ with $p$ and $q$
\linebreak[2] 
\BTA\ terms of sort $\Thr$ that contain only variables from $V$.
We are only interested in models of \BTA\ in which guarded recursive
specifications have unique solutions, such as the projective limit model
of \BTA\ presented in~\cite{BB03a}.

\begin{example}
\label{example-BTA+REC}
A simple example of a guarded recursive specification is the one
consisting of following two equations:
\begin{ldispl}
x = \pcc{x}{a}{y}\;, \qquad
y = \pcc{y}{b}{\Stop}\;.
\end{ldispl}%
The $x$-component of the solution of this guarded recursive 
specification is the thread that first performs basic action $a$ 
repeatedly until the reply from the execution environment on performing
$a$ is $\False$, next performs basic action $b$ repeatedly until the 
reply from the execution environment on performing $b$ is $\False$, and 
after that terminates.
\end{example}

\section{Services and Use Operators}
\label{sect-TSI}

A thread may perform a basic action for the purpose of requesting a 
named service provided by an execution environment to process a method 
and to return a reply to the thread at completion of the processing of 
the method.
In this section, we review the extension of \BTA\ with services and
operators which are concerned with this kind of interaction between 
threads and services.

It is assumed that a fixed but arbitrary finite set $\Foci$ of 
\emph{foci} has been given.
Foci play the role of names of the services provided by an execution 
environment.
It is also assumed that a fixed but arbitrary finite set $\Meth$ of
\emph{methods} has been given.
For the set $\BAct$ of basic actions, we take the set
$\set{f.m \where f \in \Foci, m \in \Meth}$.
Performing a basic action $f.m$ is taken as making a request to the
service named $f$ to process command $m$.

A service is able to process certain methods.
The processing of a method may involve a change of the service.
The reply value produced by the service at completion of the processing 
of a method is either $\True$, $\False$ or $\Blocked$.
The special reply $\Blocked$, standing for blocked, is used to deal with
the situation that a service is requested to process a method that it is 
not able to process.

\begin{example}
\label{example-service}
A simple example of a service is one that is able to process methods for 
pushing a natural number on a stack ($\push{n}$), testing whether the 
top of the stack equals a natural number ($\topeq{n}$), and popping the 
top element from the stack ($\pop$).
Processing of a pushing method or a popping method changes the service,
because it changes the stack with which it deals, and produces the reply
value $\True$ if no stack overflow or stack underflow occurs and
$\False$ otherwise.
Processing of a testing method does not change the service, because it
does not changes the stack with which it deals, and produces the reply
value $\True$ if the test succeeds and $\False$ otherwise.
Attempted processing of a method that the service is not able to process
changes the service into one that is not able to process any method and
produces the reply $\Blocked$.
\end{example}

The following is assumed with respect to services:
\begin{itemize}
\item
a signature $\Sig{\Services}$ has been given that includes the following
sorts:
\begin{itemize}
\item
the sort $\Serv$ of \emph{services};
\item
the sort $\Repl$ of \emph{replies};
\end{itemize}
and the following constants and operators:
\begin{itemize}
\item
the
\emph{empty service} constant $\const{\emptyserv}{\Serv}$;
\item
the \emph{reply} constants $\const{\True,\False,\Blocked}{\Repl}$;
\pagebreak[2]
\item
for each $m \in \Meth$, the
\emph{derived service} operator $\funct{\derive{m}}{\Serv}{\Serv}$;
\item
for each $m \in \Meth$, the
\emph{service reply} operator $\funct{\sreply{m}}{\Serv}{\Repl}$;
\end{itemize}
\item
a minimal $\Sig{\Services}$-algebra $\ServAlg$ has been given in which
$\True$, $\False$, and $\Blocked$ are mutually different, and
\begin{itemize}
\item
$\AND{m \in \Meth}{}
  \derive{m}(z) = z \Land \sreply{m}(z) = \Blocked \Limpl
  z = \emptyserv$
holds;
\item
for each $m \in \Meth$,
$\derive{m}(z) = \emptyserv \Liff \sreply{m}(z) = \Blocked$ holds.
\end{itemize}
\end{itemize}

The intuition concerning $\derive{m}$ and $\sreply{m}$ is that on a
request to service $S$ to process method $m$:
\begin{itemize}
\item
if $\sreply{m}(S) \neq \Blocked$, $S$ processes $m$, produces the reply
$\sreply{m}(S)$, and then proceeds as $\derive{m}(S)$;
\item
if $\sreply{m}(S) = \Blocked$, $S$ is not able to process method $m$ and
proceeds as $\emptyserv$.
\end{itemize}
The empty service $\emptyserv$ itself is unable to process any method.

We introduce the following additional operators:
\begin{itemize}
\item
for each $f \in \Foci$, the binary \emph{use} operator
$\funct{\use{\ph}{f}{\ph}}{\Thr \x \Serv}{\Thr}$.
\end{itemize}
We use infix notation for the use operators.

Intuitively, the thread denoted by a closed term of the form 
$\use{p}{f}{S}$ is the thread that results from processing all basic 
actions performed by thread denoted by $p$ that are of the form $f.m$ by
service $S$.
When a basic action of the form $f.m$ performed by a thread is processed 
by a service, the service changes in accordance with the method concerned
and affects the thread as follows: the basic action is turned into the 
internal action $\Tau$ and the two ways to proceed reduce to one on the 
basis of the reply value produced by the service.

The axioms for the use operators are given in Table~\ref{axioms-use}.%
\begin{table}[!t]
\caption{Axioms for use operators}
\label{axioms-use}
\begin{eqntbl}
\begin{saxcol}
\use{\Stop}{f}{S} = \Stop                            & & \axiom{TSU1} \\
\use{\DeadEnd}{f}{S} = \DeadEnd                      & & \axiom{TSU2} \\
\use{(\Tau \bapf x)}{f}{S} =
                          \Tau \bapf (\use{x}{f}{S}) & & \axiom{TSU3} \\
\use{(\pcc{x}{g.m}{y})}{f}{S} =
\pcc{(\use{x}{f}{S})}{g.m}{(\use{y}{f}{S})}
 & \mif f \neq g                                       & \axiom{TSU4} \\
\use{(\pcc{x}{f.m}{y})}{f}{S} =
\Tau \bapf (\use{x}{f}{\derive{m}(S)})
                       & \mif \sreply{m}(S) = \True    & \axiom{TSU5} \\
\use{(\pcc{x}{f.m}{y})}{f}{S} =
\Tau \bapf (\use{y}{f}{\derive{m}(S)})
                       & \mif \sreply{m}(S) = \False   & \axiom{TSU6} \\
\use{(\pcc{x}{f.m}{y})}{f}{S} = \Tau \bapf \DeadEnd
                       & \mif \sreply{m}(S) = \Blocked & \axiom{TSU7}
\end{saxcol}
\end{eqntbl}
\end{table}
In this table, $f$ and $g$ stand for arbitrary foci from $\Foci$,
$m$ stands for an arbitrary method from $\Meth$, and $S$ stands for an
arbitrary term of sort $\Serv$.
Axioms TSU3 and TSU4 express that the internal action $\Tau$ and basic 
actions of the form $g.m$, where $f \neq g$, are not processed by the 
service.
Axioms TSU5 and TSU6 express that a thread is affected by a service
as described above when a basic action of the form $f.m$ performed by 
the thread is processed by the service.
Axiom TSU7 expresses that inaction takes place when a basic action of 
the form $f.m$ performed by the thread cannot be processed by the 
service.

\begin{example}
\label{example-use}
We consider the stack services described in 
Example~\ref{example-service}.
For each sequence $\sigma$ of natural numbers, we take $\NNS(\sigma)$ 
as a constant for the stack service that deals with a stack whose 
content is represented by $\sigma$.
Provided a precise description of the stack services has been given,
axioms TSU1--TSU7 can be used to prove the following equations for all 
$\sigma$ whose size is less than the maximal stack size:
\begin{ldispl}
\use{(\nns.\push{n} \bapf x)}{\nns}{\NNS(\sigma)} =
\Tau \bapf (\use{x}{\nns}{\NNS(n \sigma)})\;,
\\
\use{(\pcc{x}{\nns.\pop}{\Stop})}{\nns}{\NNS(\epsilon)} =
\Tau \bapf \Stop\;,\footnotemark
\\
\use{(\pcc{x}{\nns.\pop}{\Stop})}{\nns}{\NNS(n \sigma)} =
\Tau \bapf (\use{x}{\nns}{\NNS(\sigma)})\;.
\end{ldispl}%
\footnotetext{We use the notation $\epsilon$ for the empty sequence.}
\end{example}

Henceforth, we write \BTAuse\ for \BTA, taking the set
$\set{f.m \where f \in \Foci, m \in \Meth}$ for $\BAct$, extended with
the use operators and the axioms from Table~\ref{axioms-use}.

\section{Thread Algebra and Data Linkage Dynamics Combined}
\label{sect-comb-TA-DLD}

The state changes and replies that result from performing the basic 
actions of data linkage dynamics can be achieved by means of services.
In this section, we explain how basic thread algebra can be combined
with data linkage dynamics by means of services and use operators such 
that the whole can be used for studying issues concerning the use of 
dynamic data structures in programming.

Recall that we write $\DL$ for the set of elements of the initial model
of $\DLA$.
It is assumed that $\dld \in \Foci$ and 
$\Act_\DLD \subseteq \Meth$.

For $\Sig{\Services}$, we take the signature that consists of the sorts,
constants and operators that are mentioned in the assumptions with 
respect to services made in Section~\ref{sect-TSI} and a constant 
$\DLDS(L)$ of sort $\Serv$ for each $L \in \DL$.

For $\ServAlg$, we take a minimal $\Sig{\Services}$-algebra that 
satisfies the conditions that are mentioned in the assumptions with 
respect to services made in Section~\ref{sect-TSI} and the following 
conditions for each $L \in \DL$:
\begin{ldispl}
\renewcommand{\arraystretch}{1.125}
\begin{gceqns}
\derive{m}(\DLDS(L)) = \DLDS(\effop{m}(L)) & \mif m \in \Act_\DLD\;, \\
\derive{m}(\DLDS(L)) = \emptyserv          & \mif m \not\in \Act_\DLD\;, 
\eqnsep
\sreply{m}(\DLDS(L)) = \yldop{m}(L)        & \mif m \in \Act_\DLD\;, \\
\sreply{m}(\DLDS(L)) = \Blocked            & \mif m \not\in \Act_\DLD\;. \\
\end{gceqns}
\renewcommand{\arraystretch}{1}
\end{ldispl}%
Note that $\ServAlg$ is unique up to isomorphism.
The elements of the interpretation of the sort $\Serv$ in $\ServAlg$ are
called \emph{data linkage dynamics services}.

By means of threads and data linkage dynamics services, we can give a 
precise picture of computations in which dynamic data structures are 
involved.

In order to represent computations, we use the binary relation
$\step{\Tau}$ on closed terms of  \BTAuse\ defined by $p \step{\Tau} q$
if and only if $p = \Tau \bapf q$.
Thus, $p \step{\Tau} q$ indicates that $p$ can perform an internal 
action and then proceed as $q$.
Moreover, for each method $\alpha \in \Act_\DLD$, we write
$(\alpha)$ instead of $\dld.\alpha$.
\begin{example}
\label{example-TA-DLD-non-value}
We consider a simple thread in which non-value-related basic actions of
\DLD\ occur:
\begin{ldispl}
(\getatobj{r}) \bapf (\getatobj{t}) \bapf
(\addfield{r}{\nm{up}}) \bapf (\addfield{t}{\nm{dn}}) \bapf
(\setfield{r}{\nm{up}}{t}) \bapf (\setfield{t}{\nm{dn}}{r}) \bapf
(\clrspot{t}) \bapf \Stop\;.
\end{ldispl}%
We use the rewrite rules of \DLD\ and axiom TSU5 from the axioms for the 
use operators to obtain the following picture of the computation of this 
thread in the case where the initial state is the empty data linkage:
\begin{small}
\begin{ldispl}
\begin{geqns}
\use
{((\getatobj{r}) \bapf (\getatobj{t}) \bapf
  (\addfield{r}{\nm{up}}) \bapf (\addfield{t}{\nm{dn}}) \bapf
  (\setfield{r}{\nm{up}}{t}) \bapf (\setfield{t}{\nm{dn}}{r}) \bapf
  (\clrspot{t}) \bapf \Stop)}{\dld}
{\DLDS(\emptydl)}
\step{\Tau} {}
\\
\use
{((\getatobj{t}) \bapf
  (\addfield{r}{\nm{up}}) \bapf (\addfield{t}{\nm{dn}}) \bapf
  (\setfield{r}{\nm{up}}{t}) \bapf (\setfield{t}{\nm{dn}}{r}) \bapf
  (\clrspot{t}) \bapf \Stop)}{\dld}
{\DLDS(\slink{r}{\ul{0}})}
\step{\Tau} {}
\\
\use
{((\addfield{r}{\nm{up}}) \bapf (\addfield{t}{\nm{dn}}) \bapf
  (\setfield{r}{\nm{up}}{t}) \bapf (\setfield{t}{\nm{dn}}{r}) \bapf
  (\clrspot{t}) \bapf \Stop)}{\dld}
{\DLDS(\slinkp{r}{\ul{0}} \dlcom \slinkp{t}{\ul{1}})}
\step{\Tau} {}
\\
\use
{((\addfield{t}{\nm{dn}}) \bapf
  (\setfield{r}{\nm{up}}{t}) \bapf (\setfield{t}{\nm{dn}}{r}) \bapf
  (\clrspot{t}) \bapf \Stop)}{\dld}
{\DLDS(\slinkp{r}{\ul{0}} \dlcom \slinkp{t}{\ul{1}} \dlcom
       \pflinkp{\ul{0}}{\nm{up}})}
\step{\Tau} {}
\\
\use
{((\setfield{r}{\nm{up}}{t}) \bapf (\setfield{t}{\nm{dn}}{r}) \bapf
  (\clrspot{t}) \bapf \Stop)}{\dld}
{\DLDS(\slinkp{r}{\ul{0}} \dlcom \slinkp{t}{\ul{1}} \dlcom
       \pflinkp{\ul{0}}{\nm{up}} \dlcom \pflinkp{\ul{1}}{\nm{dn}})}
\step{\Tau} {}
\\
\use
{((\setfield{t}{\nm{dn}}{r}) \bapf (\clrspot{t}) \bapf \Stop)}{\dld}
{\DLDS(\slinkp{r}{\ul{0}} \dlcom \slinkp{t}{\ul{1}} \dlcom
       \flinkp{\ul{0}}{\nm{up}}{\ul{1}} \dlcom
       \pflinkp{\ul{1}}{\nm{dn}})}
\step{\Tau} {}
\\
\use
{((\clrspot{t}) \bapf \Stop)}{\dld}
{\DLDS(\slinkp{r}{\ul{0}} \dlcom \slinkp{t}{\ul{1}} \dlcom
       \flinkp{\ul{0}}{\nm{up}}{\ul{1}} \dlcom
       \flinkp{\ul{1}}{\nm{dn}}{\ul{0}})}
\step{\Tau} {}
\\
\use
{\Stop}{\dld}
{\DLDS(\slinkp{r}{\ul{0}} \dlcom
       \flinkp{\ul{0}}{\nm{up}}{\ul{1}} \dlcom
       \flinkp{\ul{1}}{\nm{dn}}{\ul{0}})}\;.
\end{geqns}
\end{ldispl}%
\end{small}
\end{example}
\begin{example}
\label{example-TA-DLD-value}
We also consider a simple thread in which value-related basic actions of
\DLD\ occur:
\begin{ldispl}
(\assneg{u}{u}) \bapf (\assadd{s}{t}{u}) \bapf
(\assneg{u}{u}) \bapf \Stop\;.
\end{ldispl}%
This is a thread for calculating the difference of two values as
described in Section~\ref{sect-DLD}.
We use the rewrite rules of \DLD\ and axiom TSU5 from the axioms for the 
use operators to obtain the following picture of a computation of this 
thread:
\begin{small}
\begin{ldispl}
\hsp{-1.9}
\begin{geqns}
\use
{((\assneg{u}{u}) \bapf (\assadd{s}{t}{u}) \bapf
  (\assneg{u}{u}) \bapf \Stop)}{\dld}
{
 \DLDS(\slinkp{s}{\ul{0}} \dlcom
       \slinkp{t}{\ul{1}} \dlcom \valass{\ul{1}}{7} \dlcom
       \slinkp{u}{\ul{2}} \dlcom \valass{\ul{2}}{3})}
\step{\Tau} {}
\\
\use
{((\assadd{s}{t}{u}) \bapf (\assneg{u}{u}) \bapf \Stop)}{\dld}
{\DLDS(\slinkp{s}{\ul{0}} \dlcom
       \slinkp{t}{\ul{1}} \dlcom \valass{\ul{1}}{7} \dlcom
       \slinkp{u}{\ul{2}} \dlcom \valass{\ul{2}}{-3})}
\step{\Tau} {}
\\
\use
{((\assneg{u}{u}) \bapf \Stop)}{\dld}
{\DLDS(\slinkp{s}{\ul{0}} \dlcom \valass{\ul{0}}{4} \dlcom
       \slinkp{t}{\ul{1}} \dlcom \valass{\ul{1}}{7} \dlcom
       \slinkp{u}{\ul{2}} \dlcom \valass{\ul{2}}{-3})}
\step{\Tau} {}
\\
\use
{\Stop}{\dld}
{\DLDS(\slinkp{s}{\ul{0}} \dlcom \valass{\ul{0}}{4} \dlcom
       \slinkp{t}{\ul{1}} \dlcom \valass{\ul{1}}{7} \dlcom
       \slinkp{u}{\ul{2}} \dlcom \valass{\ul{2}}{3})}\;.
\end{geqns}
\end{ldispl}%
\end{small}
\end{example}
These examples show that DLA\ provides a notation that enables us to get
a clear picture of the successive states of a computation.

In~\cite{BL02a},  \PGA\ (ProGram Algebra) and a hierarchy of program 
notations rooted in \PGA\ are presented.
Included in this hierarchy are very simple program notations which are
close to existing assembly languages up to and including simple program
notations that support structured programming by offering a rendering of
conditional and loop constructs.
In~\cite{BL02a}, threads that are definable by finite guarded recursive
specifications over \BTA\ are taken as the behaviours of programs 
represented by closed \PGA\ terms.
The combination of basic thread algebra and data linkage dynamics by 
means of services and use operators can be used for studying issues 
concerning the use of dynamic data structures in programming at the 
level of program behaviours.
Together with one of the program notations rooted in \PGA, this
combination can be used for studying issues concerning the use of
dynamic data structures in programming at the level of programs.

We mention one such issue.
In general terms, the issue is whether we can do without automatic
garbage collection by program transformation at the price of a linear
increase of the number of available atomic objects.
Below we phrase this issue more precisely for \PGLD, but it can be
studied using any other program notation rooted in \PGA\ as well.
\PGLD\ is close to existing assembly languages and has absolute jump
instructions.

Let $\PGLD_\dld$ be an instance of \PGLD\ in which all basic actions of
\DLD\ are available as basic instructions.
For each program $P$ from $\PGLD_\dld$, we write $\extr{P}$ for the
thread that is the behaviour of $P$ according to~\cite{BL02a}.
Let $\DLD_\mathsf{afgc}$ be the variation of \DLD\ in which all basic
actions of the form $\getatobj{s}$ are treated as if they are preceded
by $\fgc$ and let, for each $L \in \DL$, $\DLDS_\mathsf{afgc}(L)$ be
the corresponding data linkage dynamics service with initial state
$L$.
Data linkage dynamics services have the cardinality of the set $\AtObj$
of atomic objects as parameter.
We write $\DLD^n(L)$ and $\DLD^n_\mathsf{afgc}(L)$ to indicate that
the actual cardinality is $n$.
The above-mentioned issue can now be phrased as follows: for which
natural numbers $c$ and $c'$ does there exist a program transformation
that transforms each program $P$ from $\PGLD_\dld$ to a program $Q$ from
$\PGLD_\dld$ such that, for all natural numbers $n$,
$\use{\extr{P}}{\dld}{\DLDS^n_\mathsf{afgc}(\emptydl)} =
\use{\extr{Q}}{\dld}{\DLDS^{c \mul n + c'}(\emptydl)}$?

\section{Another Description of Data Linkage Dynamics}
\label{sect-DLD-alt-descr}

In this section, we describe the state changes and replies that result
from performing the basic actions of \DLD\ in the world of sets.
This alternative description is a widening of the description of the
state changes and replies that result from performing the basic actions 
of molecular dynamics that was given in~\cite{BM06c}.
In Section~\ref{sect-correctness}, we will demonstrate that the 
alternative description agrees with the description based on \DLA.
Thus, we will show the connection between molecular dynamics and the 
upgrade of it presented in the current paper.

We define sets $\nm{SS}$, $\nm{AS}_1$, $\nm{AS}_2$ and $\DLR$ as
follows:
\begin{ldispl}
\begin{aeqns}
\nm{SS}  & = & \Spot \to (\AtObj \union \set{\bot})\;,
\eqnsep
\nm{AS}_1  & = &
\Union{A \in \setof{(\AtObj)}}
 (A \to
  \Union{F \in \setof{(\Field)}} (F \to (\AtObj \union \set{\bot})))\;,
\eqnsep
\nm{AS}_2  & = &
\Union{A \in \setof{(\AtObj)}} (A \to (\Value \union \set{\bot}))\;,
\eqnsep
\DLR & = &
\set{\tup{\sigma,\zeta,\xi} \in
     \nm{SS} \x \nm{AS}_1 \x \nm{AS}_2 \where {}
\\ & & \phantom{\{{}}
     \dom(\zeta) = \dom(\xi) \Land
     \rng(\sigma) \subseteq \dom(\zeta) \union \set{\bot} \Land {}
\\ & & \phantom{\{{}}
         \Forall{a \in \dom(\zeta)}
          {(\rng(\zeta(a)) \subseteq
           \dom(\zeta) \union \set{\bot})}
    }\;.
\end{aeqns}
\end{ldispl}%

The elements of $\DLR$ can be considered representations of
deterministic data linkages.
Let $\tup{\sigma,\zeta,\xi} \in \DLR$, let $s \in \Spot$, let
$a \in \dom(\zeta)$, and let $f \in \dom(\zeta(a))$.
Then $\sigma(s)$ is the content of spot $s$ if $\sigma(s) \neq \bot$,
$f$ is a field of atomic object $a$, $\zeta(a)(f)$ is the content of
field $f$ of atomic object $a$ if $\zeta(a)(f) \neq \bot$, and
$\xi(a)$ is the value assigned to atomic object $a$ if
$\xi(a) \neq \bot$.
The content of spot $s$ is undefined if $\sigma(s) = \bot$, the content
of field $f$ of atomic object $a$ is undefined if
$\zeta(a)(f) = \bot$, and the value assigned to atomic object $a$ is
undefined if $\xi(a) = \bot$.
Notice that $\dom(\zeta)$ is taken for the set of all atomic objects
that are in use.
Therefore, the content of each spot, i.e.\ each element of
$\rng(\sigma)$, must be in $\dom(\zeta)$ if the content is defined
and, for each atomic object $a$ that is in use, the content of each of
its fields, i.e.\ each element of $\rng(\zeta(a))$, must be in
$\dom(\zeta)$ if the content is defined.

The effect and yield operations on $\DL$ are modelled by the effect and
yield operations on $\DLR$ that are defined in Table~\ref{def-eff-yld}.%
\begin{table}[p]
\caption{Definition of effect and yield operations}
\label{def-eff-yld}
\begin{eqntbl}
\begin{seqncol}
\effopr{\getatobj{s}}(\sigma,\zeta,\xi) = {} 
\\ \quad
\tup{\sigma \owr \maplet{s}{\cf(\AtObj \diff \dom(\zeta))}, 
\\ \quad \phantom{(}
     \zeta \owr \maplet{\cf(\AtObj \diff \dom(\zeta))}{\emptymap}, 
\\ \quad \phantom{(}
     \xi \owr \maplet{\cf(\AtObj \diff \dom(\zeta))}{\bot}}
 & \mif \dom(\zeta) \subset \AtObj
\\
\effopr{\getatobj{s}}(\sigma,\zeta,\xi) =
\tup{\sigma,\zeta,\xi}
 & \mif \dom(\zeta) = \AtObj
\\
\effopr{\setspot{s}{t}}(\sigma,\zeta,\xi) =
\tup{\sigma \owr \maplet{s}{\sigma(t)},\zeta,\xi}
\\
\effopr{\clrspot{s}}(\sigma,\zeta,\xi) =
\tup{\sigma \owr \maplet{s}{\bot},\zeta,\xi}
\\
\effopr{\equaltst{s}{t}}(\sigma,\zeta,\xi) =
\tup{\sigma,\zeta,\xi}
\\
\effopr{\undeftst{s}}(\sigma,\zeta,\xi) =
\tup{\sigma,\zeta,\xi}
\\
\effopr{\addfield{s}{f}}(\sigma,\zeta,\xi) = {} \\ \quad 
\tup{\sigma,
     \zeta \owr
     \maplet{\sigma(s)}{\zeta(\sigma(s)) \owr \maplet{f}{\bot}},
     \xi}
 & \mif \defo{\sigma}(s) \Land f \not\in \dom(\zeta(\sigma(s)))
\\
\effopr{\addfield{s}{f}}(\sigma,\zeta,\xi) =
\tup{\sigma,\zeta,\xi}
 & \mif \Lnot (\defo{\sigma}(s) \Land f \not\in \dom(\zeta(\sigma(s))))
\\
\effopr{\rmvfield{s}{f}}(\sigma,\zeta,\xi) = {} 
\tup{\sigma,
     \zeta \owr
     \maplet{\sigma(s)}{\zeta(\sigma(s)) \dsub \set{f}},
     \xi}
 & \mif \defo{\sigma}(s) \Land f \in \dom(\zeta(\sigma(s)))
\\
\effopr{\rmvfield{s}{f}}(\sigma,\zeta,\xi) =
\tup{\sigma,\zeta,\xi}
 & \mif \Lnot (\defo{\sigma}(s) \Land f \in \dom(\zeta(\sigma(s))))
\\
\effopr{\hasfield{s}{f}}(\sigma,\zeta,\xi) =
\tup{\sigma,\zeta,\xi}
\\
\effopr{\setfield{s}{f}{t}}(\sigma,\zeta,\xi) = {} \\ \quad
\tup{\sigma,
     \zeta \owr
     \maplet
      {\sigma(s)}
      {\zeta(\sigma(s)) \owr \maplet{f}{\sigma(t)}},
     \xi} 
 & \mif \defo{\sigma}(s) \Land f \in \dom(\zeta(\sigma(s)))
\\
\effopr{\setfield{s}{f}{t}}(\sigma,\zeta,\xi) =
\tup{\sigma,\zeta,\xi}
 & \mif \Lnot (\defo{\sigma}(s) \Land f \in \dom(\zeta(\sigma(s))))
\\
\effopr{\clrfield{s}{f}}(\sigma,\zeta,\xi) = {} \\ \quad
\tup{\sigma,
     \zeta \owr
     \maplet
      {\sigma(s)}
      {\zeta(\sigma(s)) \owr \maplet{f}{\bot}},
     \xi}
 & \mif \defo{\sigma}(s) \Land f \in \dom(\zeta(\sigma(s)))
\\
\effopr{\clrfield{s}{f}}(\sigma,\zeta,\xi) =
\tup{\sigma,\zeta,\xi}
 & \mif \Lnot (\defo{\sigma}(s) \Land f \in \dom(\zeta(\sigma(s))))
\\
\effopr{\getfield{s}{t}{f}}(\sigma,\zeta,\xi) = {} 
\tup{\sigma \owr \maplet{s}{\zeta(\sigma(t))(f)},\zeta,\xi}
 & \mif \defo{\sigma}(t) \Land f \in \dom(\zeta(\sigma(t)))
\\
\effopr{\getfield{s}{t}{f}}(\sigma,\zeta,\xi) =
\tup{\sigma,\zeta,\xi}
 & \mif \Lnot (\defo{\sigma}(t) \Land f \in \dom(\zeta(\sigma(t))))
\\
\effopr{\asszero{s}}(\sigma,\zeta,\xi) =
\tup{\sigma,\zeta,\xi \owr \maplet{\sigma(s)}{0}}
 & \mif \defo{\sigma}(s)
\\
\effopr{\asszero{s}}(\sigma,\zeta,\xi) =
\tup{\sigma,\zeta,\xi}
 & \mif \Lnot \defo{\sigma}(s)
\\
\effopr{\assone{s}}(\sigma,\zeta,\xi) =
\tup{\sigma,\zeta,\xi \owr \maplet{\sigma(s)}{1}}
 & \mif \defo{\sigma}(s)
\\
\effopr{\assone{s}}(\sigma,\zeta,\xi) =
\tup{\sigma,\zeta,\xi}
 & \mif \Lnot \defo{\sigma}(s)
\\
\effopr{\assadd{s}{t}{u}}(\sigma,\zeta,\xi) = {} \\ \quad
\tup{\sigma,\zeta,
     \xi \owr
     \maplet{\sigma(s)}{\xi(\sigma(t)) + \xi(\sigma(u))}}
 & \mif \defo{\sigma}(s) \Land \defv{\sigma,\xi}(t) \Land
        \defv{\sigma,\xi}(u)
\\
\effopr{\assadd{s}{t}{u}}(\sigma,\zeta,\xi) =
\tup{\sigma,\zeta,\xi}
 & \mif \Lnot (\defo{\sigma}(s) \Land \defv{\sigma,\xi}(t) \Land
              \defv{\sigma,\xi}(u))
\\
\effopr{\assmul{s}{t}{u}}(\sigma,\zeta,\xi) = {} \\ \quad 
\tup{\sigma,\zeta,
     \xi \owr
     \maplet{\sigma(s)}
            {\xi(\sigma(t)) \cdot \xi(\sigma(u))}}
 & \mif \defo{\sigma}(s) \Land \defv{\sigma,\xi}(t) \Land
        \defv{\sigma,\xi}(u)
\\
\effopr{\assmul{s}{t}{u}}(\sigma,\zeta,\xi) =
\tup{\sigma,\zeta,\xi}
 & \mif \Lnot (\defo{\sigma}(s) \Land \defv{\sigma,\xi}(t) \Land
              \defv{\sigma,\xi}(u))
\\
\effopr{\assneg{s}{t}}(\sigma,\zeta,\xi) = {} 
\tup{\sigma,\zeta,
     \xi \owr \maplet{\sigma(s)}{- \xi(\sigma(t))}}
 & \mif \defo{\sigma}(s) \Land \defv{\sigma,\xi}(t)
\\
\effopr{\assneg{s}{t}}(\sigma,\zeta,\xi) =
\tup{\sigma,\zeta,\xi}
 & \mif \Lnot (\defo{\sigma}(s) \Land \defv{\sigma,\xi}(t))
\\
\effopr{\assinv{s}{t}}(\sigma,\zeta,\xi) = {} 
\tup{\sigma,\zeta,
     \xi \owr \maplet{\sigma(s)}{\xi(\sigma(t))^{-1}}}
 & \mif \defo{\sigma}(s) \Land \defv{\sigma,\xi}(t)
\\
\effopr{\assinv{s}{t}}(\sigma,\zeta,\xi) =
\tup{\sigma,\zeta,\xi}
 & \mif \Lnot (\defo{\sigma}(s) \Land \defv{\sigma,\xi}(t))
\\
\effopr{\eqvaltst{s}{t}}(\sigma,\zeta,\xi) =
\tup{\sigma,\zeta,\xi}
\\
\effopr{\undefvtst{s}}(\sigma,\zeta,\xi) =
\tup{\sigma,\zeta,\xi}
\end{seqncol}
\end{eqntbl}
\end{table}%
\addtocounter{table}{-1}%
\begin{table}[!p]
\caption{(Continued)}
\begin{eqntbl}
\begin{seqncol}
\yldopr{\getatobj{s}}(\sigma,\zeta,\xi) = \True
 & \mif \dom(\zeta) \subset \AtObj
\\
\yldopr{\getatobj{s}}(\sigma,\zeta,\xi) = \False
 & \mif \dom(\zeta) = \AtObj
\\
\yldopr{\setspot{s}{t}}(\sigma,\zeta,\xi) = \True
\\
\yldopr{\clrspot{s}}(\sigma,\zeta,\xi) = \True
\\
\yldopr{\equaltst{s}{t}}(\sigma,\zeta,\xi) = \True
 & \mif \sigma(s) = \sigma(t)
\\
\yldopr{\equaltst{s}{t}}(\sigma,\zeta,\xi) = \False
 & \mif \sigma(s) \neq \sigma(t)
\\
\yldopr{\undeftst{s}}(\sigma,\zeta,\xi) = \True
 & \mif \Lnot \defo{\sigma}(s)
\\
\yldopr{\undeftst{s}}(\sigma,\zeta,\xi) = \False
 & \mif \defo{\sigma}(s)
\\
\yldopr{\addfield{s}{f}}(\sigma,\zeta,\xi) = \True
 & \mif \defo{\sigma}(s) \Land f \not\in \dom(\zeta(\sigma(s)))
\\
\yldopr{\addfield{s}{f}}(\sigma,\zeta,\xi) = \False
 & \mif \Lnot (\defo{\sigma}(s) \Land f \not\in \dom(\zeta(\sigma(s))))
\\
\yldopr{\rmvfield{s}{f}}(\sigma,\zeta,\xi) = \True
 & \mif \defo{\sigma}(s) \Land f \in \dom(\zeta(\sigma(s)))
\\
\yldopr{\rmvfield{s}{f}}(\sigma,\zeta,\xi) = \False
 & \mif \Lnot (\defo{\sigma}(s) \Land f \in \dom(\zeta(\sigma(s))))
\\
\yldopr{\hasfield{s}{f}}(\sigma,\zeta,\xi) = \True
 & \mif \defo{\sigma}(s) \Land f \in \dom(\zeta(\sigma(s)))
\\
\yldopr{\hasfield{s}{f}}(\sigma,\zeta,\xi) = \False
 & \mif \Lnot (\defo{\sigma}(s) \Land f \in \dom(\zeta(\sigma(s))))
\\
\yldopr{\setfield{s}{f}{t}}(\sigma,\zeta,\xi) = \True
 & \mif \defo{\sigma}(s) \Land f \in \dom(\zeta(\sigma(s)))
\\
\yldopr{\setfield{s}{f}{t}}(\sigma,\zeta,\xi) = \False
 & \mif \Lnot (\defo{\sigma}(s) \Land f \in \dom(\zeta(\sigma(s))))
\\
\yldopr{\clrfield{s}{f}}(\sigma,\zeta,\xi) = \True
 & \mif \defo{\sigma}(s) \Land f \in \dom(\zeta(\sigma(s)))
\\
\yldopr{\clrfield{s}{f}}(\sigma,\zeta,\xi) = \False
 & \mif \Lnot (\defo{\sigma}(s) \Land f \in \dom(\zeta(\sigma(s))))
\\
\yldopr{\getfield{s}{t}{f}}(\sigma,\zeta,\xi) = \True
 & \mif \defo{\sigma}(t) \Land f \in \dom(\zeta(\sigma(t)))
\\
\yldopr{\getfield{s}{t}{f}}(\sigma,\zeta,\xi) = \False
 & \mif \Lnot (\defo{\sigma}(t) \Land f \in \dom(\zeta(\sigma(t))))
\\
\yldopr{\asszero{s}}(\sigma,\zeta,\xi) = \True
 & \mif \defo{\sigma}(s)
\\
\yldopr{\asszero{s}}(\sigma,\zeta,\xi) = \False
 & \mif \Lnot \defo{\sigma}(s)
\\
\yldopr{\assone{s}}(\sigma,\zeta,\xi) = \True
 & \mif \defo{\sigma}(s)
\\
\yldopr{\assone{s}}(\sigma,\zeta,\xi) = \False
 & \mif \Lnot \defo{\sigma}(s)
\\
\yldopr{\assadd{s}{t}{u}}(\sigma,\zeta,\xi) = \True \hsp{8}
 & \mif \defo{\sigma}(s) \Land \defv{\sigma,\xi}(t) \Land
        \defv{\sigma,\xi}(u) 
\\
\yldopr{\assadd{s}{t}{u}}(\sigma,\zeta,\xi) = \False
 & \mif \Lnot (\defo{\sigma}(s) \Land \defv{\sigma,\xi}(t) \Land
              \defv{\sigma,\xi}(u))
\\
\yldopr{\assmul{s}{t}{u}}(\sigma,\zeta,\xi) = \True
 & \mif \defo{\sigma}(s) \Land \defv{\sigma,\xi}(t) \Land
        \defv{\sigma,\xi}(u)
\\
\yldopr{\assmul{s}{t}{u}}(\sigma,\zeta,\xi) = \False
 & \mif \Lnot (\defo{\sigma}(s) \Land \defv{\sigma,\xi}(t) \Land
        \defv{\sigma,\xi}(u))
\\
\yldopr{\assneg{s}{t}}(\sigma,\zeta,\xi) = \True
 & \mif \defo{\sigma}(s) \Land \defv{\sigma,\xi}(t)
\\
\yldopr{\assneg{s}{t}}(\sigma,\zeta,\xi) = \False
 & \mif \Lnot (\defo{\sigma}(s) \Land \defv{\sigma,\xi}(t))
\\
\yldopr{\assinv{s}{t}}(\sigma,\zeta,\xi) = \True
 & \mif \defo{\sigma}(s) \Land \defv{\sigma,\xi}(t)
\\
\yldopr{\assinv{s}{t}}(\sigma,\zeta,\xi) = \False
 & \mif \Lnot (\defo{\sigma}(s) \Land \defv{\sigma,\xi}(t))
\\
\yldopr{\eqvaltst{s}{t}}(\sigma,\zeta,\xi) = \True
 & \mif \defo{\sigma}(s) \Land \defo{\sigma}(t) \Land
        \xi(\sigma(s)) = \xi(\sigma(t))
\\
\yldopr{\eqvaltst{s}{t}}(\sigma,\zeta,\xi) = \False
 & \mif \Lnot (\defo{\sigma}(s) \Land \defo{\sigma}(t) \Land
              \xi(\sigma(s)) = \xi(\sigma(t)))
\\
\yldopr{\undefvtst{s}}(\sigma,\zeta,\xi) = \True
 & \mif \defo{\sigma}(s) \Land \Lnot \defv{\sigma,\xi}(s)
\\
\yldopr{\undefvtst{s}}(\sigma,\zeta,\xi) = \False
 & \mif \Lnot (\defo{\sigma}(s) \Land \Lnot \defv{\sigma,\xi}(s))
\end{seqncol}
\end{eqntbl}
\end{table}
In these tables, $\defo{\sigma}(s)$ abbreviates $\sigma(s) \neq \bot$
and $\defv{\sigma,\xi}(s)$ abbreviates 
$\sigma(s) \neq \bot \Land \xi(\sigma(s)) \neq \bot$. 
We use the following notation for functions:
$\emptymap$ for the empty function;
$\maplet{e}{e'}$ for the function $f$ with $\dom(f) = \set{e}$ such that
$f(e) = e'$;
$f \owr g$ for the function $h$ with $\dom(h) = \dom(f) \union \dom(g)$
such that for all $e \in \dom(h)$,\, $h(e) = f(e)$ if
$e \not\in \dom(g)$ and $h(e) = g(e)$ otherwise;
$f \restr S$ for the function $g$ with $\dom(g) = S$
such that for all $e \in \dom(g)$,\, $g(e) = f(e)$;
and $f \dsub S$ for the function $g$ with $\dom(g) = \dom(f) \diff S$
such that for all $e \in \dom(g)$,\, $g(e) = f(e)$.

\section{Correctness of the Alternative Description}
\label{sect-correctness}

In this section, we show that the description of the state changes and
replies that result from performing the basic actions of \DLD\ given in
Section~\ref{sect-DLD-alt-descr} agrees with the one given in
Sections~\ref{sect-DLD-K} and~\ref{sect-DLD}, provided only 
deterministic data linkages are considered.
Recall that the basic actions of \DLD\ preserve determinism of data 
linkages.
This means that non-deterministic data linkages are not relevant to 
\DLD\ if only deterministic data linkages are allowed as initial state.

The step from deterministic data linkages to the concrete states
introduced in Section~\ref{sect-DLD-alt-descr} is an instance of a data
refinement:
\begin{itemize}
\item
the concrete states are considered representations of deterministic
data linkages;
\item
the effect and yield operations on deterministic data linkages are
modelled by the effect and yield operations on the concrete states.
\end{itemize}
This data refinement is correct in the sense that:
\begin{itemize}
\item
there is a representation of each deterministic data linkage;
\item
for each $\alpha \in \Act_\DLD$, for each deterministic data linkage,
the result of applying $\effopr{\alpha}$ to the representation of that
data linkage is the representation of the result of applying
$\effop{\alpha}$ to that data linkage;
\item
for each $\alpha \in \Act_\DLD$, for each deterministic data linkage,
the result of applying $\yldopr{\alpha}$ to the representation of that
data linkage is the result of applying $\yldop{\alpha}$ to that data
linkage.
\end{itemize}
Correctness in this sense agrees with the notion of correctness for data
refinements as used in, for example, the software development method
VDM~\cite{Jon90a}.
Following the terminology of VDM, the three aspects of correctness might
be called representation adequacy, action effect modelling and action
reply modelling.

For the purpose of stating the above-mentioned correctness rigorously,
we introduce a \emph{retrieve function} that relates the concrete states
to deterministic data linkages:
\begin{ldispl}
\retr(\sigma,\zeta,\xi) = {}
\\ \quad
\dlCom{s \in \set{s' \in \dom(\sigma) \where \sigma(s') \neq \bot}}
 \slinkp{s}{\sigma(s)} \dlcom {}
\\ \quad
\dlCom{a \in \dom(\zeta)}
\bigl
(\dlCom{f \in
        \set{f' \in \dom(\zeta(a)) \where \zeta(a)(f') = \bot}}
 \pflinkp{a}{f}\bigr) \dlcom {}
\\ \quad
\dlCom{a \in \dom(\zeta)}
\bigl
(\dlCom{f \in
        \set{f' \in \dom(\zeta(a)) \where \zeta(a)(f') \neq \bot}}
 \flinkp{a}{f}{\zeta(a)(f)}\bigr) \dlcom {}
\\ \quad
\dlCom{a \in \set{a' \in \dom(\xi) \where \xi(a') \neq \bot}}
 \valass{a}{\xi(a)}\;,
\end{ldispl}%
where $\dlCom{i \in \mathcal{I}} D_i$, with $\mathcal{I}$ a finite set 
and $D_i$ a \DLD\ term for each $i \in \mathcal{I}$, stands for a term 
$D_{i_1} \dlcom \ldots \dlcom D_{i_n}$ such that 
$\mathcal{I} = \set{i_1,\ldots,i_n}$ (all such terms are equal by 
associativity and commutativity of $\dlcom$) if $\mathcal{I}$ is not an
empty set, and for the term $\emptydl$ otherwise.
The function $\retr$ can be thought of as regaining the abstract 
deterministic data linkages from their concrete representations.
\pagebreak[2]

The correctness of the step from deterministic data linkages to the
concrete states is stated rigorously in the following theorem.
\begin{theorem}[Correctness]
\label{theorem-correctness}
$\DL$ and the effect and yield operations on $\DL$ are related to
$\DLR$ and the effect and yield operations on $\DLR$ as follows:
\begin{enumerate}
\item
for all $L \in \DL$ that are deterministic, there exists a
$\tup{\sigma,\zeta,\xi} \in \DLR$ such that:
\begin{ldispl}
L = \retr(\sigma,\zeta,\xi)\;;
\end{ldispl}%
\item
for all $\alpha \in \Act_\DLD$, for all $\tup{\sigma,\zeta,\xi} \in \DLR$:
\begin{ldispl}
\begin{aeqns}
\retr(\effopr{\alpha}(\sigma,\zeta,\xi)) & = &
\effop{\alpha}(\retr(\sigma,\zeta,\xi))\;,
\\
\yldopr{\alpha}(\sigma,\zeta,\xi) & = &
\yldop{\alpha}(\retr(\sigma,\zeta,\xi))\;.
\end{aeqns}
\end{ldispl}%
\end{enumerate}
\end{theorem}
\begin{proof}
This is straightforwardly proved by case distinction on the basic action
$\alpha$ using elementary laws for $\owr$ and $\dsub$.
The following facts about the connection between deterministic data
linkages and their representations are useful in the proof:
\begin{ldispl}
\retr(\sigma \owr \maplet{s}{a},\zeta,\xi) =
\retr(\sigma,\zeta,\xi) \dlcom \slinkp{s}{a}\;,
\\
\retr(\sigma,
      \zeta \owr \maplet{a}{\zeta(a) \owr \maplet{f}{\bot}},
      \xi) =
\retr(\sigma,\zeta,\xi) \dlcom \pflinkp{a}{f}\;,
\\
\retr(\sigma,
      \zeta \owr \maplet{a}{\zeta(a) \owr \maplet{f}{b}},
      \xi) =
\retr(\sigma,\zeta,\xi) \dlcom \flinkp{a}{f}{b}\;,
\\
\retr(\sigma,\zeta,
      \xi \owr \maplet{a}{n}) =
\retr(\sigma,\zeta,\xi) \dlcom \valass{a}{n}\;.
\end{ldispl}%
They follow from the definition of $\retr$ and elementary laws for
$\owr$.
\qed
\end{proof}

\section{Another Description of Garbage Reclamation}
\label{sect-garb-coll-maps}

In this section, we describe the state changes and replies that
result from performing the reclamation-related actions of data linkage
dynamics in the world of sets.
Like the effect operations for reclamation on $\DL$, the effect 
operations for reclamation on $\DLR$ are defined using auxiliary functions.
The auxiliary function 
$\funct{\reach}{\nm{SS} \x \nm{AS}_1}{\setof{(\AtObj)}}$ is used to define
both $\effopr{\fgc}$ and $\effopr{\rgc}$, and the auxiliary function 
$\funct{\incycle}{\nm{AS}_1}{\setof{(\AtObj)}}$ is used to define 
$\effopr{\rgc}$.
These auxiliary functions are defined as follows:
\begin{ldispl}
\begin{aeqns}
\reach(\sigma,\zeta) & = &
\Union{a \in \rng(\sigma)} \reach(a,\zeta)\;,
\eqnsep
\incycle(\zeta) & = &
\set{a \in \dom(\zeta) \where
     \Exists{a' \in \rng(\zeta(a))}{(a \in \reach(a',\zeta))}}\;,
\end{aeqns}
\end{ldispl}%
where $\reach(a,\zeta) \subseteq \AtObj$ is inductively defined by
the following rules:
\pagebreak[2]
\begin{itemize}
\item
$a \in \reach(a,\zeta)$;
\item
if $a' \in \reach(a,\zeta)$ and $a'' \in \rng(\zeta(a'))$, then
$a'' \in \reach(a,\zeta)$.
\end{itemize}

\begin{table}[!t]
\caption{Definition of effect and yield operations for reclamation}
\label{def-eff-yld-reclaim}
\begin{eqntbl}
\begin{eqncol}
\effopr{\fgc}(\sigma,\zeta,\xi) =
\tup{\sigma,
     \zeta \restr \reach(\sigma,\zeta),
     \xi \restr \reach(\sigma,\zeta)}
\\
\effopr{\rgc}(\sigma,\zeta,\xi) = {} \\ \qquad
\tup{\sigma,
     \zeta \restr (\reach(\sigma,\zeta) \union \incycle(\zeta)),
     \xi \restr (\reach(\sigma,\zeta) \union \incycle(\zeta))}
\\
\effopr{\sdgetatobj{s}}(\sigma,\zeta,\xi) =
\sd(\sigma(s),\effopr{\getatobj{s}}(\sigma,\zeta,\xi))
\\
\effopr{\sdsetspot{s}{t}}(\sigma,\zeta,\xi) =
\sd(\sigma(s),\effopr{\setspot{s}{t}}(\sigma,\zeta,\xi))
\\
\effopr{\sdclrspot{s}}(\sigma,\zeta,\xi) =
\sd(\sigma(s),\effopr{\clrspot{s}}(\sigma,\zeta,\xi))
\\
\effopr{\sdsetfield{s}{f}{t}}(\sigma,\zeta,\xi) =
\sd(\zeta(\sigma(s))(f),
     \effopr{\setfield{s}{f}{t}}(\sigma,\zeta,\xi))
\\
\effopr{\sdclrfield{s}{f}}(\sigma,\zeta,\xi) =
\sd(\zeta(\sigma(s))(f),
     \effopr{\clrfield{s}{f}}(\sigma,\zeta,\xi))
\\
\effopr{\sdgetfield{s}{t}{f}}(\sigma,\zeta,\xi) =
\sd(\sigma(s),
     \effopr{\getfield{s}{t}{f}}(\sigma,\zeta,\xi))
\\
\effopr{\udgetatobj{s}}(\sigma,\zeta,\xi) =
\ud(\sigma(s),\effopr{\getatobj{s}}(\sigma,\zeta,\xi))
\\
\effopr{\udsetspot{s}{t}}(\sigma,\zeta,\xi) =
\ud(\sigma(s),\effopr{\setspot{s}{t}}(\sigma,\zeta,\xi))
\\
\effopr{\udclrspot{s}}(\sigma,\zeta,\xi) =
\ud(\sigma(s),\effopr{\clrspot{s}}(\sigma,\zeta,\xi))
\\
\effopr{\udsetfield{s}{f}{t}}(\sigma,\zeta,\xi) =
\ud(\zeta(\sigma(s))(f),
     \effopr{\setfield{s}{f}{t}}(\sigma,\zeta,\xi))
\\
\effopr{\udclrfield{s}{f}}(\sigma,\zeta,\xi) =
\ud(\zeta(\sigma(s))(f),
     \effopr{\clrfield{s}{f}}(\sigma,\zeta,\xi))
\\
\effopr{\udgetfield{s}{t}{f}}(\sigma,\zeta,\xi) =
\ud(\sigma(s),
     \effopr{\getfield{s}{t}{f}}(\sigma,\zeta,\xi))
\eqnsep
\yldopr{\fgc}(\sigma,\zeta,\xi) = \True
\\
\yldopr{\rgc}(\sigma,\zeta,\xi) = \True
\\ 
\yldopr{\sdgetatobj{s}}(\sigma,\zeta,\xi) =
\yldopr{\getatobj{s}}(\sigma,\zeta,\xi)
\\
\yldopr{\sdsetspot{s}{t}}(\sigma,\zeta,\xi) =
\yldopr{\setspot{s}{t}}(\sigma,\zeta,\xi)
\\
\yldopr{\sdclrspot{s}}(\sigma,\zeta,\xi) =
\yldopr{\clrspot{s}}(\sigma,\zeta,\xi)
\\
\yldopr{\sdsetfield{s}{f}{t}}(\sigma,\zeta,\xi) =
\yldopr{\setfield{s}{f}{t}}(\sigma,\zeta,\xi)
\\
\yldopr{\sdclrfield{s}{f}}(\sigma,\zeta,\xi) =
\yldopr{\clrfield{s}{f}}(\sigma,\zeta,\xi)
\\
\yldopr{\sdgetfield{s}{t}{f}}(\sigma,\zeta,\xi) =
\yldopr{\getfield{s}{t}{f}}(\sigma,\zeta,\xi)
\\
\yldopr{\udgetatobj{s}}(\sigma,\zeta,\xi) =
\yldopr{\getatobj{s}}(\sigma,\zeta,\xi)
\\
\yldopr{\udsetspot{s}{t}}(\sigma,\zeta,\xi) =
\yldopr{\setspot{s}{t}}(\sigma,\zeta,\xi)
\\
\yldopr{\udclrspot{s}}(\sigma,\zeta,\xi) =
\yldopr{\clrspot{s}}(\sigma,\zeta,\xi)
\\
\yldopr{\udsetfield{s}{f}{t}}(\sigma,\zeta,\xi) =
\yldopr{\setfield{s}{f}{t}}(\sigma,\zeta,\xi)
\\
\yldopr{\udclrfield{s}{f}}(\sigma,\zeta,\xi) =
\yldopr{\clrfield{s}{f}}(\sigma,\zeta,\xi)
\\
\yldopr{\udgetfield{s}{t}{f}}(\sigma,\zeta,\xi) =
\yldopr{\getfield{s}{t}{f}}(\sigma,\zeta,\xi)
\end{eqncol}
\end{eqntbl}
\end{table}

The auxiliary function $\funct{\sd}{\AtObj \x \DLR}{\DLR}$ is used to
define $\effopr{\alpha}$ for the actions $\alpha$ in which a basic 
action of \DLD\ is combined with safe disposal, and 
the auxiliary function $\funct{\ud}{\AtObj \x \DLR}{\DLR}$ is used to
define $\effopr{\alpha}$ for the actions $\alpha$ in which a basic 
action of \DLD\ is combined with unsafe disposal. 
These auxiliary functions are defined as follows:
\begin{ldispl}
\begin{aceqns}
\sd(a,\tup{\sigma,\zeta,\xi}) & = &
\tup{\sigma,\zeta \dsub \set{a},\xi \dsub \set{a}}
 & \mif a \not\in \reach(\sigma,\zeta)\;,
\\
\sd(a,\tup{\sigma,\zeta,\xi}) & = &
\tup{\sigma,\zeta,\xi}
 & \mif a \in \reach(\sigma,\zeta)\;,
\eqnsep
\ud(a,\tup{\sigma,\zeta,\xi}) & = &
\multicolumn{2}{@{}l@{}}
 {\sd(a,\tup{\clrs(a,\sigma),\clrf(a,\zeta),\xi})\;,}
\end{aceqns}
\end{ldispl}%
where $\funct{\clrs}{\AtObj \x \nm{SS}}{\nm{SS}}$ and
$\funct{\clrf}{\AtObj \x \nm{AS}_1}{\nm{AS}_1}$ are defined as follows:
\pagebreak[2]
\begin{ldispl}
\begin{aceqns}
\clrs(a,\sigma)(s) & = & \sigma(s) & \mif \sigma(s) \neq a\;, \\
\clrs(a,\sigma)(s) & = & \bot      & \mif \sigma(s) = a\;,
\eqnsep
\clrf(a,\zeta)(a')(f) & = & \zeta(a')(f)
 & \mif \zeta(a')(f) \neq a\;, \\
\clrf(a,\zeta)(a')(f) & = & \bot & \mif \zeta(a')(f) = a\;.
\end{aceqns}
\end{ldispl}%

The effect and yield operations for reclamation on $\DLR$ are
defined in Table~\ref{def-eff-yld-reclaim}.

Theorem~\ref{theorem-correctness} goes through for \DLD-R.
The additional cases to be considered involve proofs by induction over
the definition of $\reach(a,\zeta)$ for appropriate $a$ and $\zeta$.

\section{Conclusions}
\label{sect-concl}

We have presented an algebra of which the elements are intended for
modelling the states of computations in which dynamic data structures
are involved.
We have also presented a simple model of computation in which states of
computations are modelled as elements of this algebra and state changes
take place by means of certain actions. 
We have described the state changes and replies that result from
performing those actions by means of a term rewriting system with rule
priorities.

We followed a rather fundamental approach.
Instead of developing the model of computation on top of an existing
theory or model, we started from first principles by giving an
elementary algebraic specification~\cite{BT06a} of the states of
computations in which dynamic data structures are involved.
We found out that term rewriting with priorities is a convenient
technique to describe the dynamic aspects of the model in an appealing
mechanical way.
In particular, we managed to give a clear idea of the features related
to reclamation of garbage.

It stands out that the description of the dynamic aspects of the 
presented model of computation by means of a term rewriting system with 
priorities is rather sizable.
However, an alternative description in the world of sets that, unlike
the one that we have given in this paper, covers both deterministic and
non-deterministic data linkages, would be rather sizable as well.
Moreover, we believe that the use of conditional term 
rewriting~\cite{BK86d} instead of priority rewriting would give rise to 
a less compact and more complicated description.

The presented model of computation and its description are well-thought
out, but the choices made can only be justified by applications.
Applications to that effect should not only include applications as a 
setting in which theoretical issues concerning dynamic data structures 
are studied.
They should also include applications as a setting in which practical
programming problems that involve dynamic data structures are studied.
In~\cite{BM08a}, we have studied the programming of an interpreter for a 
program notation that is close to existing assembly languages in this
setting.

Together with thread algebra and program algebra~\cite{BL02a}, we hold 
the model of computation as described in this paper to be a suitable 
starting-point for investigations into theoretical issues concerning the 
interplay between programs and dynamic data structures, including issues 
concerning reclamation of garbage.
We have studied one such issue in~\cite{BM08e}, namely the feasibility 
of automatically making everything garbage as soon as it can be viewed 
as garbage.
For the study in question, the abstraction from the representation of 
dynamic data structures by means of pointers turned out to be really
useful.

\subsection*{Acknowledgements}

This research has been partly carried out in the framework of the
Jacquard-project Symbiosis, which is funded by the Netherlands
Organisation for Scientific Research (NWO).
We thank two anonymous referees for carefully reading a preliminary 
version of this paper and for suggesting improvements of the 
presentation of the paper.

\bibliographystyle{splncs03}
\bibliography{IS}

\appendix

\section{Term Rewriting Systems}
\label{appendix}

In this appendix, the basic definitions and results regarding term 
rewriting systems are collected.

We assume that a set of constants, a set of operators with fixed 
arities, and a set of variables have been given; and we consider term 
rewriting systems for terms that can be built from the constants, 
operators, and variables in these sets.

A \emph{rewrite rule} is a pair of terms $t \osred s$, where $t$ is not 
a variable and each variable occurring in $s$ occurs in $t$ as well.
A \emph{term rewriting system} is a set of rewrite rules.

Let $\cR$ be a term rewriting system.
Then a \emph{reduction step} of $\cR$ is a pair $t \osred s$ such that 
for some substitution instance $t' \osred s'$ of a rewrite rule of 
$\cR$, $t'$ is a subterm of $t$, and $s$ is $t$ with $t'$ replaced by 
$s'$.
Here, $t'$ is called the \emph{redex} of the reduction step, $s'$ is 
called the \emph{contractum} of the reduction step, and the replacement 
of the redex is called the \emph{contraction} of the redex.
The \emph{one-step reduction relation} $\osred$ of $\cR$ is the set of 
all reduction steps of $\cR$.
This means that the one-step reduction relation $\osred$ is the closure 
of $\cR$ under substitutions and contexts.

Let $\cR$ be a term rewriting system.
Then a \emph{reduction sequence} of $\cR$ is a finite sequence 
$t_1 \osred t_2$, \ldots, $t_n \osred t_{n+1}$ of consecutive reduction 
steps of $\cR$ or an infinite sequence $t_1 \osred t_2$, 
$t_2 \osred t_3$, \ldots\ of consecutive reduction steps of $\cR$.
A \emph{reduction} of $\cR$ is a pair $t \msred s$ such that either 
$t$ is syntactically equal to $s$ or there exists a finite reduction 
sequence $t_1 \osred t_2$, \ldots, $t_n \osred t_{n+1}$ of $\cR$ such 
that $t$ and $s$ are syntactically equal to $t_1$ and $t_{n+1}$, 
respectively.
Here, $s$ is called a \emph{reduct} of $t$.
The \emph{reduction relation} $\msred$ of $\cR$ is the set of all 
reductions of $\cR$.
This means that the reduction relation $\msred$ is the closure of the 
one-step reduction step relation $\osred$ under transitivity and 
reflexivity.

Let $\cR$ be a term rewriting system.
Then a term $t$ \emph{is a normal form of $\cR$} if there does not exist 
a term $s$ such that $t \osred s$ is a reduction step of $\cR$.
A term $t$ \emph{has a normal form in $\cR$} if there exists a term $s$ 
such that $t \msred s$ is a reduction of $\cR$ and $s$ is a normal form 
of $\cR$.
$\cR$ \emph{is strongly normalizing} on term $t$ if there does not exist 
an infinite reduction sequence $t \osred t_1$, $t_1 \osred t_2$, 
$t_2 \osred t_3$, \ldots\ of $\cR$.
$\cR$ \emph{is strongly normalizing} if $\cR$ is strongly normalizing on 
all terms.
$\cR$ \emph{is weakly confluent} if for each two reduction steps 
$t \osred s_1$ and $t \osred s_2$ of $\cR$ there exist two reductions 
$s_1 \msred s$ and $s_2 \msred s$ of $\cR$.
If $\cR$ is strongly normalizing and weakly confluent, all terms
have a unique normal form in $\cR$.

Let $\cR$ be a term rewriting system.
Then a \emph{reduction ordering} for $\cR$ is a well-founded ordering on 
terms that is closed under substitutions and contexts.
$\cR$ is strongly normalizing if and only if there exists a reduction
ordering $>$ for $\cR$ such that $t > s$ for each rewrite rule 
$t \osred s$ of $\cR$.

Let $\cR$ be a term rewriting system.
Then a \emph{critical pair} of $\cR$ is a pair 
$(t' \osred s'_1, t' \osred s'_2)$ of different reduction steps of $\cR$ 
such that $t' \osred s'_2$ is a substitution instance of a rewrite rule 
$t \osred s$ of $\cR$ and the redex of $t' \osred s'_1$ is a 
substitution instance of a non-variable subterm of $t$.
A critical pair $(t' \osred s'_1, t' \osred s'_2)$ of $\cR$ 
\emph{is convergent} if $s'_1$ and $s'_2$ have a common reduct.
$\cR$ is weakly confluent if and only if all critical pairs of $\cR$ are
convergent.

Let $\cR$ be a term rewriting system and 
let $E$ be a set of equations between terms.
Then a \emph{reduction modulo $E$ step} of $\cR$ is a pair 
$t \osredm{E} s$ such that there exist a reduction step $t' \osred s'$ 
of $\cR$ such that $t = t'$ and $s = s'$ are derivable from $E$.
The \emph{one-step reduction modulo $E$ relation} $\osredm{E}$ of $\cR$ 
is the set of all reduction modulo $E$ steps of $\cR$.
A \emph{reduction modulo $E$ sequence} of $\cR$ is a finite sequence 
$t_1 \osredm{E} t_2$, \ldots, $t_n \osredm{E} t_{n+1}$ of consecutive 
reduction modulo $E$ steps of $\cR$ or an infinite sequence 
$t_1 \osredm{E} t_2$, $t_2 \osredm{E} t_3$, \ldots\ of consecutive 
reduction modulo $E$ steps of $\cR$.
A \emph{reduction modulo $E$} of $\cR$ is pair $t \msredm{E} s$ such 
that either $t = s$ is derivable from $E$ or there exists a finite 
reduction modulo $E$ sequence $t_1 \osredm{E} t_2$, \ldots, 
$t_n \osredm{E} t_{n+1}$ of $\cR$ such that $t$ and $s$ are 
syntactically equal to $t_1$ and $t_{n+1}$, respectively.
The \emph{reduction modulo $E$ relation} $\msredm{E}$ of $\cR$ is the 
set of all reductions modulo $E$ of $\cR$.

Let $\cR$ be a term rewriting system and 
let $E$ be a set of equations between terms.
Then a term $t$ \emph{is a normal form of $\cR$ with respect to 
reduction modulo $E$} if there does not exist a term $s$ such that 
$t \osredm{E} s$ is a reduction modulo $E$ step of $\cR$.
A term $t$ \emph{has a normal form in $\cR$ with respect to reduction 
modulo $E$} if there exists a term $s$ such that $t \msredm{E} s$ is a 
reduction modulo $E$ of $\cR$ and $s$ is a normal form of $\cR$ with
respect to reduction modulo $E$.
$\cR$ \emph{is strongly normalizing modulo} $E$ on term $t$ if there 
does not exist an infinite sequence $t \osredm{E} t_1$, 
$t_1 \osredm{E} t_2$, $t_2 \osredm{E} t_3$, \ldots\ of reduction modulo 
$E$ steps of $\cR$.
$\cR$ \emph{is strongly normalizing modulo} $E$ if $\cR$ is strongly 
normalizing modulo $E$ on all terms.
$\cR$ \emph{is weakly confluent modulo} $E$ if for each reduction modulo 
$E$ step $t \osredm{E} s_1$ of $\cR$ and each reduction step 
$t \osred s_2$ of $\cR$ there exist reductions modulo $E$ 
$s_1 \msredm{E} s$ and $s_2 \msredm{E} s$ of $\cR$.
If $\cR$ is strongly normalizing modulo $E$ and $\cR$ is weakly confluent 
modulo $E$, all terms have a unique normal form modulo $E$ in $\cR$.

Let $\cR$ be a term rewriting system and 
let $E$ be a set of equations between terms.
A reduction ordering~$>$ for $\cR$ \emph{is $E$-compatible} if $t > s$
implies $t' > s'$ for all terms $t$, $t'$, $s$, and $s'$ for which 
$t = t'$ and $s = s'$ are derivable from $E$.
$\cR$ is strongly normalizing modulo $E$ if and only if there exists an 
$E$-compatible reduction ordering $>$ for $\cR$ such that $t > s$ for 
each rewrite rule $t \osred s$ of $\cR$.

\end{document}